\documentclass[letterpaper,11pt]{article}

\usepackage{microtype}
\usepackage[letterpaper,margin=0.97in]{geometry}
\usepackage[numbers,sort&compress]{natbib}
\usepackage[table]{xcolor}
\usepackage{amsmath}
\usepackage{amssymb}
\usepackage{enumitem}
\usepackage{textgreek}
\usepackage{graphicx}
\usepackage{framed}
\usepackage{booktabs}
\usepackage{tikz-cd}
\usepackage[framemethod=tikz]{mdframed}
\usetikzlibrary{arrows.meta}
\usepackage{amsthm}
\usepackage{mathtools}
\usepackage{xspace}
\usepackage{multirow}
\usepackage{array}
\usepackage{tabularx}
\usepackage{algorithm}
\usepackage{algpseudocode}
\usepackage{longtable}
\usepackage{thm-restate}
\makeatletter
\IfFormatAtLeastTF{2026-06-01}{%
  \renewcommand\thmt@autorefsetup{%
    \@xa\def\csname\thmt@envname autorefname\@xa\endcsname\@xa{\thmt@thmname}%
  }%
}{}
\makeatother
\usepackage{hyperref}
\usepackage{orcidlink}
\usepackage{doi}
\usepackage[capitalize, nameinlink]{cleveref}

\makeatletter
\renewcommand{\theHALG@line}{\thealgorithm.\arabic{ALG@line}}
\makeatother

\Crefname{remark}{Remark}{Remarks}
\Crefname{observation}{Observation}{Observations}
\Crefname{claim}{Claim}{Claims}

\theoremstyle{plain}
\newtheorem{theorem}{Theorem}[section]
\newtheorem{lemma}[theorem]{Lemma}

\newtheorem{corollary}[theorem]{Corollary}
\newtheorem{claim}[theorem]{Claim}

\theoremstyle{definition}
\newtheorem{definition}[theorem]{Definition}

\theoremstyle{plain}

\newcounter{open}

\theoremstyle{remark}



\newcommand{\namedref}[2]{\hyperref[#2]{#1~\ref*{#2}}}

\newcommand{\eps}{\varepsilon}
\DeclareMathOperator{\poly}{poly}
\newcommand{\LOCAL}{\ensuremath{\mathsf{LOCAL}}\xspace}
\newcommand{\CONGEST}{\ensuremath{\mathsf{CONGEST}}\xspace}
\newcommand{\set}[1]{\left\{#1\right\}}
\newcommand{\calS}{\mathcal{S}}

\newcommand{\calC}{\ensuremath{\mathcal{C}}}
\newcommand{\calA}{\ensuremath{\mathcal{A}}}

\definecolor{darkgreen}{rgb}{0,0.5,0}
\definecolor{darkred}{rgb}{0.4,0,0}
\hypersetup{
	colorlinks=true,
	linkcolor=darkred,
	citecolor=darkgreen,
	filecolor=black,
	urlcolor=[rgb]{0,0.1,0.5},
	pdftitle={Greedy-Like Defective Coloring: Distributed Algorithms and Applications},
	pdfkeywords={distributed graph algorithms, defective coloring, list coloring, deterministic algorithms, CONGEST model},
	pdfauthor={Marc Fuchs, Fabian Kuhn}
}

\title{Greedy-Like Defective Coloring:\\ Distributed Algorithms and Applications}

\date{}

\author{
   Marc Fuchs \orcidlink{0000-0003-2272-4483} \\
   \small{University of Freiburg} \\
   \small{marc.fuchs@cs.uni-freiburg.de}
   \and
   Fabian Kuhn \orcidlink{0000-0002-1025-5037} \\
   \small{University of Freiburg} \\
   \small{kuhn@cs.uni-freiburg.de}
}

\begin{document}
\maketitle
\begin{abstract}
    A $d$-defective $c$-coloring of a graph $G=(V,E)$ is a coloring of the nodes $V$ with $c$ colors such that every node has at most $d$ neighbors of the same color. Distributed algorithms for computing different variants of defective coloring are at the core of most deterministic state-of-the-art distributed coloring algorithms, and they are also an important tool in many other distributed graph algorithms. In several cases, the overall complexity could be improved if some version of defective coloring could be solved more efficiently.

    Presently, the primary challenge in understanding the distributed complexity of defective colorings is to understand which defective coloring problems can be solved in time $O(\log^{\ast} n)$ and for which such problems one needs $\Omega(\log n)$ rounds (in $n$-node graphs if the maximum degree $\Delta$ is treated as a constant). Already Barenboim and Elkin [STOC '09] showed that for any integer $p\geq 1$, a simple greedy-like algorithm that makes two passes over the nodes can color with $p^2$ colors and defect $\lfloor \Delta/p\rfloor$ in $O(\Delta+\log^{\ast} n)$ rounds. Up to now, this algorithm provides the best general tradeoff between defect and number of colors for any known algorithm with complexity $O(\log^{\ast} n)$ in bounded-degree graphs.

    In this paper, we further investigate the power of this greedy-like two-pass algorithm for computing defective colorings. We first show that the algorithm can be generalized to the \emph{list defective coloring} problem as defined by Fuchs and Kuhn in [DISC '23]. As a consequence of this, we in particular get an alternative algorithm for computing a (proper) $(\Delta+1)$-coloring of a graph in $\tilde{O}(\sqrt{\Delta}) + O(\log^{\ast} n)$ rounds in the CONGEST model. In the second part of the paper, we study a generalized version of the greedy-like two-pass algorithm for computing standard defective colorings. For many fixed numbers of colors $c$, we improve the best known defect by a constant factor within this round-complexity regime. We also show that, for any fixed $c\geq 1$, the generalized two-pass algorithm with a fixed common candidate family and a fixed global deterministic tie-breaking rule in each phase (based only on local neighbor counts) cannot guarantee a $c$-coloring with defect less than $(1-o(1))\cdot\Delta/\sqrt{c}$ as $\Delta\to\infty$. Here, global means that every vertex uses the same rule in each phase.
\end{abstract}

\newpage

\tableofcontents
\newpage

\section{Introduction and related work}
\label{sec:intro}

Distributed coloring is arguably the most studied problem in the area of distributed graph algorithms. The task is to color the nodes (or edges) of some $n$-node graph $G=(V,E)$, which also defines the network topology. The nodes of $G$ have unique $O(\log n)$-bit identifiers, and they can exchange messages over the edges of $G$ in synchronous rounds. If the messages can be of arbitrary size, this setting is known as the \LOCAL model~\cite{Linial1992,Peleg2000} and if the message size is restricted to $O(\log n)$ bits, it is known as the \CONGEST model~\cite{Peleg2000}. In its standard form, the coloring problem asks for a \emph{proper} coloring, where no two adjacent nodes are assigned the same color. The number of allowed colors is usually some function of the maximum degree $\Delta$. Of particular interest in the distributed setting is the problem of coloring with $\Delta+1$ colors, which is the number of colors used by a sequential greedy algorithm. The round complexity of distributed $(\Delta+1)$-coloring is typically either studied as a function of $n$ or (primarily) as a function of $\Delta$. As a function of $n$, the best known complexities in the \LOCAL model are $\tilde{O}(\log^{5/3} n)$ for deterministic algorithms~\cite{GhaffariGrunauFOCS24} and $\tilde{O}(\log^{5/3}\log n)$ for randomized algorithms~\cite{GhaffariGrunauFOCS24,ChangLP18}.\footnote{Throughout the paper, we use the notation $\tilde{O}(\cdot)$ to hide polylogarithmic factors in the argument, i.e., $\tilde{O}(x) = x\cdot\poly\log x$.} The best known complexities in the \CONGEST model are $O(\log^2\Delta\cdot\log n)$ for deterministic algorithms~\cite{GhaffariKuhn21} and $O(\log^3\log n)$ for randomized algorithms~\cite{GhaffariKuhn21,HalldorssonKMT21}. The best known lower bound is $\Omega(\log^* n)$, which holds in the \LOCAL model and even for randomized algorithms and if $\Delta=2$~\cite{Linial1992,Naor1991}. When considering the problem's complexity as a function of $\Delta$, a $\log^* n$ dependency therefore becomes unavoidable. In this regime, before the present work, the best known upper bounds were $O(\sqrt{\Delta\log\Delta}+\log^* n)$ in the \LOCAL~\cite{fraigniaud16local,BarenboimEG18,MausT20} and $O(\sqrt{\Delta}\cdot\log^2\Delta\cdot\log^6\log\Delta + \log^* n)$ in the \CONGEST model~\cite{FK23}. Furthermore, for problems with a deterministic complexity at most $f(\Delta) \cdot \log^*n$ (for some function $f(\cdot)$), it is known that randomization does not help when $\Delta$ is constant~\cite{chang16exponential}.

\paragraph*{Subsequent work.}
Subsequent work~\cite{FuchsK26Subsequent} presents a deterministic
$O(\sqrt{\Delta}+\log^*n)$-round algorithm for $(\Delta+1)$-coloring
in both the \LOCAL and \CONGEST models.
For a color space of size $O(\Delta)$, the same round complexity
also holds for $(\mathit{degree}+1)$-list coloring.
The proper-coloring comparisons below refer to the state of the art
before that work. The present paper develops the general
Two-Sweep OLDC primitive and analyzes its guarantees for defective
coloring.

\paragraph*{Defective colorings.} \emph{Defective} colorings are natural generalizations of proper colorings, which have become instrumental tools in many modern distributed coloring algorithms and also more generally in distributed graph algorithms. For integers $d\geq 0$ and $c\geq 1$, a $d$-defective $c$-coloring of the nodes of a graph $G=(V,E)$ is a coloring of $V$ with at most $c$ colors such that every node $v\in V$ has at most $d$ neighbors of the same color. In the context of distributed coloring algorithms, defective colorings can in particular be used to design recursive distributed coloring algorithms. A $d$-defective $c$-coloring partitions a graph into $c$ parts of maximum degree $d$. As the complexity of many of the fastest distributed coloring algorithms quite strongly depends on the maximum degree, each of the $c$ parts can be colored more efficiently (in parallel or in some cases also sequentially). In \cite{lovasz66}, Lov\'asz showed that every graph of maximum degree $\Delta$ has a defective $c$-coloring (for any $c\geq 1$) with defect $d=\lfloor \Delta/c\rfloor$.\footnote{Note that this is tight if the graph is a complete graph with $\Delta+1$ nodes.} If defective colorings with these guarantees could be computed efficiently in the distributed setting, even the most naive implementation of the above recursive coloring strategy would directly lead to efficient algorithms to color with $O(\Delta)$ colors.

\paragraph*{Distributed defective coloring algorithms.}
Defective colorings were introduced to distributed computing by Barenboim and Elkin~\cite{BarenboimE09} and by Kuhn~\cite{Kuhn2009}. Both used defective colorings to improve the complexity of $(\Delta+1)$-coloring from $O(\Delta\cdot\log\Delta + \log^* n)$ (cf.\ \cite{SzegedyV93,KuhnW06}) to $O(\Delta+\log^* n)$. These works show that, for maximum degree $\Delta\geq 1$ and any integer $0\leq d\leq\Delta$, one can compute a $d$-defective coloring with $O\big(\big(\frac{\Delta}{d+1}\big)^2\big)$ colors in time $O(f(\Delta) + \log^* n)$. The algorithm of \cite{Kuhn2009} is an adaptation of the proper $O(\Delta^2)$-coloring algorithm of Linial~\cite{Linial1992} and has no $\Delta$-dependency at all, i.e., it achieves $f(\Delta)=O(1)$. The algorithm of \cite{BarenboimE09} is slightly slower but uses fewer colors. The algorithm of \cite{BarenboimE09} is based on a simple greedy-like two-phase algorithm, which we discuss in more detail below. In the more than 15 years since \cite{BarenboimE09,Kuhn2009} were published, there has been no improvement on the achievable defective colorings if the $n$-dependency of the round complexity is restricted to be $O(\log^* n)$ (except for the special case of defective 3-colorings~\cite{BalliuHLOS19}). The only known distributed defective coloring algorithms that achieve $d$-defective colorings with $O(\Delta/d)$ colors for general $d$ require deterministic time $\poly\log n$ and randomized time $\poly\log\log n$. Bounds of this kind can be achieved by phrasing defective colorings as an instance of the Lov\'asz Local Lemma (LLL) and by using generic distributed LLL algorithms to compute the colorings~\cite{MoserT10,ChungPS14,FischerG17,GhaffariHK18}. It is known that when using generic LLL algorithms, the deterministic complexity must be at least $\Omega(\log n)$ and the randomized complexity must be at least $\Omega(\log\log n)$~\cite{Brandt2016}. Very recently, a first efficient distributed algorithm that achieves the existentially optimal trade-off between defect and number of colors was presented in \cite{potentialproblems}. The paper gives a deterministic $\poly(\Delta\cdot\log n)$-round algorithm to compute a $d$-defective $c$-coloring whenever $c\cdot(d+1)>\Delta$.

Defective colorings provide a decomposition of a graph into several graphs of lower degree on which a given graph problem (e.g., some proper coloring problem) can then be solved more efficiently. Defective colorings are thus naturally most interesting if we focus on the complexity as a function of $\Delta$ and if we thus restrict the $n$-dependency to $O(\log^* n)$. Understanding which distributed defective coloring problems can be computed in time $f(\Delta)\cdot\log^* n$ for some function $f(\cdot)$ and for which problems, we need at least $\Omega(\log_\Delta n)$ rounds deterministically (and at least $\Omega(\log_\Delta\log n)$ rounds with randomization) is therefore currently the most important open problem in the context of distributed defective coloring and possibly even in the context of distributed coloring more generally. It is known that $d$-defective $2$-coloring requires $\Omega(\log_\Delta n)$ deterministic rounds even if $d=\Delta-2$~\cite{BalliuHLOS19} and that for $d$-defective $3$-coloring, one needs $\Omega(\log_\Delta n)$ rounds whenever $d\leq (\Delta-3)/2$~\cite{Balliu0KOS25}. Apart from those hardness results, it is only known that $d$-defective $c$-coloring needs $\Omega(\log_\Delta n)$ rounds, whenever $c\cdot(d+1)\leq\Delta$~\cite{balliu2021hideandseek}. In \cite{BarenboimE09}, it is shown that for any positive integer $p$, one can compute a $\lfloor\Delta/p\rfloor$-defective coloring with $p^2$ colors in time $O(\Delta+\log^* n)$. Despite intensive research on distributed coloring problems, for general graphs, we do not know any $O(f(\Delta)+\log^{\ast} n)$-round (or $O(f(\Delta)\cdot\log^* n)$-round) algorithm with a better trade-off between defect and number of colors, apart from the known three-color improvement.

\paragraph*{Arbdefective Colorings.} Because computing defective colorings with a good defect to number of colors ratio appears to be a hard problem, researchers have considered a relaxed variant known as \emph{arbdefective coloring}. In a coloring with arbdefect $d$, each color class forms a graph of arboricity at most $d$. Arbdefective colorings were introduced by Barenboim and Elkin in \cite{BarenboimE11journal} and they have found many applications in the context of distributed coloring algorithms and more generally in the context of distributed graph algorithms (e.g., \cite{BarenboimE11journal,barenboim16sublinear,fraigniaud16local,BarenboimEG18,Kuhn20,GhaffariKuhn21,FaourGGKR25,GGHIR23,GhaffariGrunauFOCS24}). In this work, we use a slightly different definition of those kinds of colorings: a $d$-arbdefective $c$-coloring is a node coloring with $c$ colors
together with an edge orientation such that each node has at most
$d$ outneighbors of the same color.\footnote{To see why those two variants are similar, assume each color class has arboricity at most $d$, i.e., one can partition the edges of each color class into at most $d$ (spanning) forests, and orient each edge towards its root. This ensures that within a color class each node has at most $d$ parent pointers and hence an outdegree of at most $d$.} It is known that a $d$-arbdefective $c$-coloring can be computed in time $f(\Delta)\cdot\log^* n$ if and only if $c\cdot(d+1)>\Delta$~\cite{balliu2021hideandseek}. In this case, the best known algorithm requires $\sqrt{\frac{\Delta}{d+1}}\cdot\poly\log\Delta + O(\log^* n)$~\cite{FK23}.

\paragraph*{Graphs of Bounded Neighborhood Independence.} Significantly better defective coloring algorithms are known for graphs of bounded neighborhood independence, i.e., graphs where every node can only have constant-sized sets of pairwise non-adjacent neighbors. It has been shown in \cite{BarenboimE11} that in this case, a $d$-arbdefective coloring with $c$ colors directly also provides an $O(d)$-defective coloring with $c$ colors. Together with some other advances (in particular on a generalized list version of defective colorings, see below), this fact led to particularly efficient $(\Delta+1)$-coloring algorithms in graphs of bounded neighborhood independence~\cite{Kuhn20,BalliuKO20,BalliuBKO22,FuchsK25}. In \cite{FuchsK25}, it is shown that in such graphs, a $(\Delta+1)$-coloring can be computed in $\big(\log\Delta\big)^{O\big(\frac{\log\log\Delta}{\log\log\log\Delta}\big)} + O(\log^* n)$ \LOCAL rounds, i.e., in time quasi-polylogarithmic in $\Delta$. An important subclass consists of line graphs of graphs and of bounded-rank hypergraphs. For the special case of $(2\Delta-1)$-edge coloring, there even is an algorithm that solves the problem in $O(\log^{12}\Delta +\log^* n)$ rounds~\cite{BalliuBKO22}.

\paragraph*{List Defective Colorings.} The state-of-the-art distributed coloring algorithms all rely on the explicit use of list colorings: In a list coloring instance, every node $v$ obtains a list $L_v$ of colors as input and as output, each node $v$ has to choose some color from $L_v$ such that the graph is properly colored. List colorings, in particular, allow us to color a graph in several phases, where each phase extends a given partial proper coloring of the graph. In order to effectively use defective colorings in the context of list colorings, we should be able to use the coloring to divide a given list coloring instance into several (ideally independent) list coloring instances of smaller degree. If we divide the nodes into subsets of approximately equal smaller degree, we also need to divide the colors into the same number of parts in such a way that all lists are divided approximately evenly. Without global communication, it is clearly not possible to partition the colors such that the induced partition of each of the lists is close to a uniform partition.\footnote{Even with global communication, this is only possible if the lists are sufficiently large (as a function of $n$).} One can however generalize the idea in the following way. Assume that all the lists $L_v$ consist of colors from some space of size $C$ (say from the set $\set{1,\dots,C}$). One can then partition the global color space into $p>1$ parts of size $\approx C/p$. This induces a partition of the lists that is not uniform. As a consequence, each node has to select one of the $p$ remaining color spaces in such a way that the number of conflicting neighbors decreases at a rate that is comparable to the decrease of the list size. This \emph{color space reduction} idea has been introduced in \cite{Kuhn20} and it has since then been used in \cite{BalliuKO20,BalliuBKO22,FK23,GhaffariKuhn21,FuchsK25}. In \cite{FK23}, it is shown that the core task that needs to be solved can be phrased as a natural list extension of defective coloring. Formally, in \emph{list defective coloring}, each node $v$ obtains a color list $L_v$ together with a defect function $d_v:L_v \to \mathbb{N}_0$ that assigns non-negative integer values to the colors in $L_v$. The coloring has to assign a color $x_v\in L_v$ to each node $v$ such that the number of neighboring nodes $u$ that get assigned color $x_u=x_v$ is at most $d_v(x_v)$.

\paragraph*{The Two-Phase Algorithm for Computing Defective Colorings.}
A particularly simple and in our opinion very natural (distributed) algorithm to compute defective colorings is a \emph{two-phase greedy-like} algorithm that was first described by Barenboim and Elkin in \cite{BarenboimE09}. Assume that we are given a proper coloring of the nodes $V$ of a graph $G=(V,E)$ with $q$ colors. We make two passes over the nodes. In both passes, nodes of the same color are processed simultaneously. In the first pass, we process the nodes in the order given by increasing colors and in the second pass, we process the nodes in the opposite order (i.e., by decreasing colors). The algorithm has an integer parameter $p\geq 1$. In both passes, each node $v\in V$ picks a color from $\set{0,\dots,p-1}$. In the first pass, node $v$ chooses a color $x_{v,1}$ that occurs least often among the colors $x_{u,1}$ already chosen by its neighbors $u$ with smaller input color. In the second pass, it analogously chooses a color $x_{v,2}$ that occurs least often among the colors $x_{u,2}$ already chosen by its neighbors $u$ with larger input color. The final color $x_v$ of $v$ is computed as $x_v=p\cdot x_{v,1} + x_{v,2}\in \set{0,\dots,p^2-1}$. It is not hard to see that each node can have at most $\lfloor \Delta/p\rfloor$ neighbors that choose the same color. This two-pass algorithm has later in particular also been used in \cite{BalliuHLOS19}, where it is in particular shown that one can compute a $d$-defective $3$-coloring in time $O(\Delta+\log^* n)$ as long as $d\geq (2\Delta-4)/3$.

The main goal of the present paper is to understand the power of the two-pass greedy approach to defective coloring problems. For this purpose, we generalize the algorithm in the following way:
\begin{itemize}
    \item In the first pass (where nodes are still processed by increasing input color), every node $v$ chooses a nonempty set $S_v$ of candidate colors from a family $\calS$. It selects a set that minimizes an averaging-based upper bound on the defect that the second pass can guarantee. This bound uses the candidate sets already chosen by neighbors with smaller input colors.
    \item In the second pass (where nodes are processed by decreasing input colors), each node $v$ chooses a color $x_v\in S_v$ that minimizes the sum of the number of neighbors with larger input colors that have already chosen $x_v$ and the number of neighbors with smaller input colors whose candidate sets contain $x_v$. This sum bounds the final defect of $v$. For list defective coloring, the candidate sets are nonempty subsets of $L_v$, and both passes include the color-dependent defects in their calculations.
\end{itemize}

\paragraph*{Why two sweeps suffice.}
For an ordinary defective coloring, let $k_v(x)$ count earlier neighbors
whose candidate sets contain $x$, and let $r_v(x)$ count later neighbors
whose final color is $x$, and write $N_>(v)$ for the later
neighbors. If $v$ chooses $x$, its final number of conflicts is at most
$k_v(x)+r_v(x)$. Since each later neighbor chooses only one color,
\[
 \min_{x\in S_v}\bigl(k_v(x)+r_v(x)\bigr)
 \leq \frac{\sum_{x\in S_v}k_v(x)+|N_>(v)|}{|S_v|}.
\]
Thus the first sweep can minimize the right-hand side before the later
neighbors have chosen their final colors. For example, if a two-color
palette $\{a,b\}$ has earlier counts $k_v(a)=1$, $k_v(b)=0$, and there
are two later neighbors, the bound is $3/2$. Whatever those two
neighbors do, Phase~II finds a color with at most one conflict.
For list defective coloring, subtracting $d_v(x)$ from each score gives
the same averaging argument for the excess over the allowed defect.

In the following, we will refer to this algorithm as the \textit{Two-Sweep Algorithm}. A formally precise description of the algorithm is given in \Cref{sec:genericTwoSweep,alg:GenericSweep}. We next give an overview of the technical contributions of the paper.

\section{Our contributions}
\label{sec:contributions}

The goal of this paper is to understand the power of the \textit{Two-Sweep Greedy} approach for computing distributed defective colorings. We pursue this goal in two complementary directions.

First, we show that this approach extends beyond ordinary defective coloring and can solve a generalization of list defective coloring known as \emph{oriented list defective coloring}~\cite{FK23full}.  We then show that this more general form of list defective coloring, applied recursively, yields a new and simple deterministic $(\Delta+1)$-coloring algorithm in the \CONGEST model with round complexity $\tilde{O}(\sqrt{\Delta}) + O(\log^* n)$.

Second, we investigate the intrinsic limitations of the same two-sweep paradigm. More concretely, we ask for the smallest possible defect achievable when the number of output colors $C$ is fixed. We show that suitable choices of candidate palettes improve the known constants for many values of $C$, while also proving that the Two-Sweep Algorithm cannot beat the $1/\sqrt{C}$ relative-defect barrier when it uses a fixed common candidate family and fixed deterministic global tie-breaking rule based only on local neighbor counts (\Cref{thm:2sweeplower}).

The common algorithm is formalized in
\Cref{sec:genericTwoSweep}. \Cref{sec:OLDC2Sweep,sec:Two-SweepInCONGEST}
develop its OLDC implementation and the CONGEST applications, using
the color-space and list-arbdefective reductions of~\cite{FK23}.
\Cref{sec:SweepAndBuckets} studies fixed common candidate families:
the LP supplies upper bounds for bucket constructions, while the
lower-bound proof additionally uses finite trees to realize the
required behavior of the algorithm. We defer these constructions to
\Cref{sec:TechnicalDetailsLowerBound-Appendix}.

\subsection{List defective coloring and applications}
\label{sec:contributionLists}

We extend our Two-Sweep Algorithm to solve \emph{oriented list defective coloring (OLDC)} instances. In OLDC (formally defined below), the edges are oriented (possibly in both directions) and the defect of a node is measured only with respect to its outgoing edges. As we will see, we need this more general version of list defective coloring when using it as a building block for our new $(\Delta+1)$-coloring algorithm formalized in \Cref{thm:CONGESTcoloring}.

\begin{definition}\label{def:OLDC}
    Let $G=(V,E)$ be a directed graph and $\mathcal C$ be an output color space (i.e., the set of available colors). For every node $v\in V$, we are given a nonempty list $\varnothing\neq L_v\subseteq \mathcal{C}$ and a defect function $d_v:L_v\to \mathbb{N}_{0}$. For the given lists and defect functions, an oriented list defective coloring of $G$ is a coloring of the nodes $V$ such that every node $v$ is assigned a color $x \in L_v$ such that $v$ has at most $d_v(x)$ outneighbors of color $x$.
\end{definition}

We write $C:=|\mathcal C|$ for the size of the output color space. In the above definition, $\mathbb{N}_{0}$ denotes the set of natural numbers including $0$. If $G$ is an undirected graph and we replace every edge $\set{u,v}$ of $G$ by the two directed edges $(u,v)$ and $(v,u)$, then the oriented list defective coloring gives a standard list defective coloring of $G$. When solving OLDC problems in a distributed way, we assume that the communication graph $G$ itself is an undirected graph (i.e., communication can always happen in both directions) and the edge directions are given as an edge orientation of $G$ that is part of the input.

We assume a common encoding of the ambient colors using $O(\log C)$ bits per color. We assume an initial proper $q$-coloring, which only determines the processing order. Its labels need not belong to the lists $L_v$, and no relation between $q$ and $C$ is assumed. Note that the Two-Sweep Algorithm directly applies to the oriented (list) defective coloring problem. When applying the Two-Sweep Algorithm to an OLDC problem, in both sweeps, each node only considers assignments of color sets (in the first sweep) or colors (in the second sweep) to its outneighbors.

Our Two-Sweep Algorithm for OLDC uses a common integer parameter $p\geq1$ to bound the palette sizes $|S_v|$. Typically, $p\approx\sqrt{\Lambda}$ for $\Lambda:=\max_{u\in V}|L_u|$. The basic algorithm requires
\begin{align}\label{eq:SlackInBaseCase}
    \sum_{x \in L_v} (d_v(x) + 1) > \max\left\{p, \frac{|L_v|}{p}\right\} \cdot \beta_v,
\end{align}
where $\beta_v$ denotes the outdegree of $v$. For $\eps>0$, the following theorem trades an extra slack factor $1+\eps$ for fewer rounds.

\setcounter{theorem}{0}
\begin{restatable}{theorem}{restatefirst}\label{thm:2sweepalgolists}
  Let $G=(V,E)$ be a directed graph equipped with an initial proper $q$-vertex coloring. For every node $v\in V$, let $\beta_v$ be its outdegree. Assume that we are given an oriented list defective coloring instance with color lists $L_v\subseteq\mathcal C$ and defect functions $d_v:L_v\to\mathbb N_0$. Let $p\geq1$ be an integer, let $\eps\geq0$, and assume that
  \[
    \forall v \in V\,:\, \sum_{x\in L_v}(d_v(x)+1) > (1+\eps)\cdot\max\set{p,\frac{|L_v|}{p}}\cdot \beta_v.
  \]
  Then, the Two-Sweep Algorithm solves the given oriented list defective coloring instance deterministically in $O\big(\min\set{q, (p/\eps)^2 + \log^*q}\big)$ \LOCAL rounds.\footnote{We slightly abuse notation and allow $\eps=0$, in which case we assume that $\min\set{q, (p/0)^2+\log^* q}=q$.} The algorithm uses messages of $O(\log q+p\log|\mathcal C|)$ bits.
\end{restatable}

The direct implementation in \Cref{alg:GenericSweep} exchanges initial colors, palettes of at most $p$ output colors, and final output colors.
Without extra slack ($\eps=0$), the direct algorithm still solves instances satisfying \Cref{eq:SlackInBaseCase} in $O(q)$ rounds. The fast implementation (\Cref{alg:Sweep2}) proceeds in two stages. First, it utilizes the methods of \cite{KawarabayashiS18,Kuhn2009} with parameter $\alpha=\eps/p$ to compute a defective coloring with $O(p^2/\eps^2)$ colors, ensuring that each node $v$ has at most $(\eps/p)\beta_v$ outneighbors of its own color. Subsequently, it applies the Two-Sweep Algorithm to the subgraph containing only bichromatic edges, after reserving sufficient defect for the monochromatic edges. The fast implementation also exchanges intermediate colors during \Cref{lemm:defColorBlackBox}; the branch-selection rule ensures that these messages fit the same size bound, as shown in \Cref{sec:fastersweep}.

\paragraph*{List Coloring with Bounded Outdegree} It has been known since the late 1980s that an edge-oriented graph with maximum outdegree $\beta$ can be colored with $O(\beta^2)$ colors in $O(\log^* n)$ \LOCAL rounds~\cite{Linial1987}. When considering $\beta$ as a constant, we extend this result to solving list coloring instances with lists $L_v$ of size $|L_v|=\Theta(\beta^2)$. In $O(\log^* n)$ \LOCAL rounds, one can first use the algorithm of \cite{Linial1987} to compute an $O(\beta^2)$-coloring that we will then use as the initial $q$-coloring of our algorithm. By choosing $p=\beta+1$ and setting all defect values to $0$, we then get an $O(\beta^2 + \log^* n)$-round list coloring algorithm to properly color with lists $L_v$ of size $|L_v| \geq \beta^2+\beta+1$. The only previous algorithm to achieve something similar is the algorithm of \cite{MausT20}. As sketched in \cite{MausT20}, their work can be adapted to compute list colorings with lists of size $|L_v|=\Theta(\beta^2\log\beta)$ in $O(\log^* n + \log^* C)$ rounds.

\paragraph*{Recursive Oriented List Defective Coloring}

\Cref{thm:2sweepalgolists} is stated in \LOCAL because a message can contain a palette of up to $p$ colors. Recursive color-space reduction (Theorem~3 of \cite{FK23}) applies this routine to small auxiliary color spaces, reducing the message size and round complexity at the cost of requiring additional slack. The resulting algorithm solves OLDC instances with slack $O(\sqrt C)$ in time polylogarithmic in $C$, using $O(\log q+\log C)$-bit messages, as stated next. It therefore works in \CONGEST whenever $q,C\leq\poly(n)$, which holds for classical proper coloring instances, such as $(\Delta+1)$-coloring, where an initial $q= O(\Delta^2)$ coloring is precomputed with \cite{Linial1987}. The recursive algorithm and its analysis appear in \Cref{sec:Two-SweepInCONGEST}.

\begin{restatable}{theorem}{restatesecond}\label{thm:OLDC2}
  Let $G=(V,E)$ be an $n$-node graph equipped with an initial proper $q$-vertex coloring and with an edge orientation. For every node $v\in V$, assume that $\beta_v$ is its outdegree. Further, assume that we are given an oriented list defective coloring instance for $G$ with color lists $L_v \subseteq \{1, \ldots, C\}$ and defect functions $d_v:L_v\to \mathbb{N}_0$ for all $v\in V$. Assume that
  \[
    \forall v \in V\,:\, \sum_{x\in L_v}(d_v(x)+1) \geq 3 \cdot \sqrt{C}\cdot\beta_v.
  \]
  Then, our recursive Two-Sweep Algorithm deterministically solves the given oriented list defective coloring instance in $O\big(\log^3 C + \log C\cdot\log^*q\big)$ rounds. The algorithm requires the nodes to exchange messages of at most $O(\log q + \log C)$ bits.
\end{restatable}

It might be instructive to compare the conditions of \Cref{thm:2sweepalgolists,thm:OLDC2} with the constraint $c(d+1)>\Delta$, which is necessary for a $d$-defective $c$-coloring to exist for every graph of maximum degree $\Delta$ (cf.~\cite{lovasz66}). When choosing the parameters optimally and replacing $\beta_v$ with the maximum degree $\Delta$ (i.e., assuming an undirected graph), \Cref{thm:2sweepalgolists,thm:OLDC2} show that a $d$-defective $c$-coloring can be computed efficiently by a simple greedy-like algorithm as long as $\sqrt{c}(d+1)>\alpha \Delta$ for some constant $\alpha>1$ and this can also be generalized to the list defective coloring setting.

\paragraph*{Proper \boldmath$(\Delta+1)$-Coloring in the \CONGEST Model}
As a direct application, we can plug \Cref{thm:OLDC2} into the framework of \cite{fraigniaud16local,Kuhn20,FK23} to obtain a new efficient $(\Delta+1)$-coloring (and more generally a new efficient $(\mathit{degree}+1)$-list coloring algorithm) for the \CONGEST model.

\begin{restatable}{theorem}{restatethird}\label{thm:CONGESTcoloring}
  Let $G=(V,E)$ be an $n$-node graph of maximum degree $\Delta$ and assume that all nodes $v\in V$ have color lists $L_v$ satisfying $|L_v|\geq \deg(v)+1$, consisting of colors from a color space of size $O(\Delta)$. Then the given list coloring instance can be solved in $O(\sqrt{\Delta}\cdot\log^4\Delta + \log^* n)$ deterministic rounds in the \CONGEST model.
\end{restatable}

\paragraph*{Coloring of Graphs of Bounded Neighborhood Independence}

An application of the general OLDC result \Cref{thm:OLDC2}, through the list-arbdefective reduction described in \Cref{sec:BoundedNeighborhood}, is that we can efficiently compute $(\Delta+1)$-colorings in graphs of bounded neighborhood independence in the \CONGEST model.\footnote{The neighborhood independence $\theta$ of a graph is the maximum number of pairwise non-adjacent neighbors of any node of the graph. As a prominent example, line graphs of hypergraphs of rank $r$ have neighborhood independence at most $r$.} The authors of \cite{FuchsK25} present a recursive \LOCAL algorithm that uses the algorithm of \cite{FK23} as a subroutine for the base cases. The full version~\cite{FuchsK25full} also notes that a \CONGEST implementation is possible with additional work. Here, we provide a suitable Two-Sweep base routine and analyze the message sizes throughout the recursion, obtaining a \CONGEST algorithm with the same round complexity up to $\poly\log\Delta$ factors. This result is stated in \Cref{thm:boundedNeighborhoodMain}. The details are deferred to \Cref{sec:BoundedNeighborhood}.

In asymptotic complexity bounds, we use clipped base-two
logarithms, defined by $\log x:=\log_2\max\{2,x\}$ at each iteration.
This includes $\theta=1$ and bounded $\Delta$. Logarithms such as
$\log_\lambda C$ in exact parameter choices retain their usual meaning.

\begin{theorem}\label{thm:boundedNeighborhoodMain}
  There is a deterministic \CONGEST algorithm with round complexity
  \begin{align*}
    \min \left\{
      (\theta\log\Delta)^{O\left(\frac{\log\log\Delta}
      {\max\{1,\log\log\log\Delta-\log\log\theta\}}\right)},
      O\left(\theta^2\Delta^{1/4}\log^8\Delta\right)
    \right\} + O(\log^*n)
  \end{align*}
  to solve $(\Delta+1)$-coloring in any $n$-node graph of maximum degree $\Delta$ and neighborhood independence at most $\theta$.
\end{theorem}

\subsection{Optimizing the defect as a function of the number of colors}
\label{sec:contribsPart2}

As discussed e.g.\ in \cite{FK23,FuchsK25}, we believe that in order to make significant progress on the $\Delta$-dependency of the complexity of (deterministic) distributed $(\Delta+1)$-coloring, we might need to understand which defective coloring problems can be solved in time $f(\Delta)\cdot\log^* n$ and for which such problems, there is an $\Omega(\log n)$ lower bound if $\Delta=O(1)$. Further, the question seems most relevant for the case, where the number of colors is relatively small as in this case, it seems most likely to also get a reasonably small $\Delta$-dependency for the defective coloring solution. For example, in the related arbdefective coloring problem, for which this question is understood exactly~\cite{balliu2021hideandseek}, all existing algorithms with close to optimal number of colors $C$ have a round complexity that is at least linear in $\sqrt{C}$~\cite{BarenboimEG18,FK23}.

In the regime $C=\Theta(\Delta/(d+1))$, the known algorithms cited here have round bounds with a $\sqrt C$ dependence, up to polylogarithmic factors in $\Delta$, plus $O(\log^*n)$~\cite{BarenboimEG18,FK23}; this describes their guarantees rather than a lower bound for the problem.

To make progress toward this question, we establish almost tight bounds on the best defective colorings that can be achieved by the Two-Sweep Algorithm in the model of \Cref{thm:2sweeplower} if the number of colors $C$ is fixed. In this part of the paper, we focus on the case where the maximum degree $\Delta$ is large and we concentrate on the value of the achievable defect relative to the degree of the node. We say that an algorithm achieves \emph{relative defect} $\rho\in[0,1]$ if the actual number $\operatorname{def}(v)$ of same-colored neighbors satisfies, at every vertex $v$,
\begin{align}
    \label{eq:defRelativeDefect}
    \operatorname{def}(v) \leq \rho\cdot \deg(v) + o(\deg(v)).
\end{align}
Here the additive term is bounded by a single function of $\deg(v)$,
independent of the input graph and the node, whose ratio to $\deg(v)$
tends to zero as the degree tends to infinity.
Let $\rho^*(C)$ be the infimum of all $\rho\in[0,1]$ for which there exists an $f(\Delta)\cdot\log^* n$-round deterministic distributed algorithm to compute a $C$-coloring with relative defect $\rho$. There are existing lower bounds for defective colorings with $2$ and $3$ colors. In \cite{BalliuHLOS19} and \cite{Balliu0KOS25}, it is shown that $(\Delta-2)$-defective $2$-colorings and $(\Delta-3)/2$-defective $3$-colorings require $\Omega(\log_\Delta n)$ rounds deterministically. We therefore have $\rho^*(2)=1$ and $\rho^*(3)\geq 1/2$. The best known upper bounds on $\rho^*(C)$ are $\rho^*(3)\leq 2/3$ by using the algorithm of \cite{BalliuHLOS19} and $\rho^*(C)\leq 1/\lfloor\sqrt{C}\rfloor$ by using the base two-sweep algorithm of Barenboim and Elkin~\cite{BarenboimE09} with parameter $p=\lfloor\sqrt{C}\rfloor$ and $p^2\leq C$ colors. While we do not provide any lower bounds for general defective coloring algorithms, we show that for large $C$, the generalized Two-Sweep Algorithm under the assumptions of \Cref{thm:2sweeplower} cannot be significantly better than the algorithm of \cite{BarenboimE09}.

The class of Two-Sweep Algorithms considered in the
following lower bound is defined precisely in \Cref{sec:lowerBoundModel}.
We fix $C\geq1$, a common nonempty
candidate family $\calS$, and a deterministic tie-breaking rule in
each phase. Every vertex uses the same rules, which depend only on
local neighbor counts. The initial proper coloring is supplied as
part of the input and determines the sweep order.

\begin{restatable}{theorem}{restateSweepLowerBound}\label{thm:2sweeplower}
 No Two-Sweep Algorithm can guarantee a $C$-coloring with relative defect below
  $1/\sqrt C$, even on finite trees.
\end{restatable}

More precisely, for every $0\leq\rho<1/\sqrt C$, there
are trees of arbitrarily large maximum degree $\Delta$ and supplied
proper initial colorings for which some vertex has more than
$\rho\Delta$ same-colored neighbors. For each fixed $\rho$, the number
of initial colors is bounded independently of $\Delta$.
The statement concerns guarantees over all supplied processing
colorings; preprocessing that selects the order is outside this model.

The lower-bound proof has two steps. A numerical argument in \Cref{sec:lowerBound} finds earlier-neighbor counts for which every candidate palette can be defeated by suitable later-neighbor colors. A graph construction in \Cref{sec:TechnicalDetailsLowerBound-Appendix} then forces the required palettes and final colors to occur on finite trees of arbitrarily large degree. This second step turns the numerical obstruction into a large actual defect, contradicting even the additive allowance in \Cref{eq:defRelativeDefect}.

While already the basic two-sweep algorithm of \cite{BarenboimE09} achieves a relative defect $1/\lfloor\sqrt{C}\rfloor$ that is close to $1/\sqrt{C}$, especially for small values of $C$, there is still room for improvement. For example, for $C\in\{5,6,7,8\}$, the algorithm of \cite{BarenboimE09} only guarantees a relative defect of $1/2$, which can already be achieved with only $4$ colors. A small improvement can already be achieved by our Two-Sweep Algorithm for list defective coloring (cf.\ \Cref{sec:contributionLists}). If each node $v$ has all colors $\set{1,\dots,C}$ in its list $L_v$ and we use an integer parameter $1\leq p\leq C$ (i.e., $\calS$ contains all color sets of size $p$), \Cref{thm:2sweepalgolists} guarantees that we can compute a $d$-defective $C$-coloring if $C(d+1)>\max\set{p,C/p}\cdot\Delta$. The guaranteed relative defect is therefore at most $\max\set{p,C/p}/C=\max\set{p/C,1/p}$.

We show that this can be further improved by using color sets of two different sizes in the first pass of the Two-Sweep Algorithm. More precisely, we consider the following \emph{Two-Bucket Two-Sweep Algorithm} in which the colors are partitioned into two buckets and we use sets of different sizes in the two buckets.
\begin{itemize}
    \item We define $s_1:=\lfloor \sqrt{C}\rfloor$ and $s_2:=\lceil\sqrt{C}\rceil$.
    \item We partition the $C$ colors $\set{1,\dots,C}$ into two parts $\set{1,\dots,C_1}$ and $\set{C_1+1,\dots,C}$, where $C_1$ is chosen such that $C_1\geq s_1$ and $C_2:=C-C_1\geq s_2$.
    \item The set $\calS$ consists of all $s_1$-element subsets of $\set{1,\dots,C_1}$ and all $s_2$-element subsets of $\set{C_1+1,\dots,C}$.
\end{itemize}

In the following, we also refer to the algorithm that only uses sets of a single size $p$ as the Single-Bucket Two-Sweep Algorithm (i.e., the algorithm of \Cref{thm:2sweepalgolists} when all nodes have $[C]$ as their color list). For the Two-Bucket Two-Sweep Algorithm, a relaxed linear program provides an upper bound on every execution and can be solved explicitly.

\begin{restatable}{theorem}{restateBucketExact}\label{thm:2bucketexact}
      Let $C_1+C_2=C$, $C_1\geq\lfloor\sqrt C\rfloor$, and $C_2\geq\lceil\sqrt C\rceil$. The Two-Bucket Two-Sweep Algorithm achieves a relative defect $\rho_{2B}$ satisfying
    \[
    \rho_{2B} \leq \max\set{
    \frac{\lfloor\sqrt{C}\rfloor\cdot\lceil\sqrt{C}\rceil}{\lfloor\sqrt{C}\rfloor\cdot C_2 + \lceil\sqrt{C}\rceil\cdot C_1},
    \frac{\lceil\sqrt{C}\rceil^2}{C_2+\lfloor\sqrt{C}\rfloor\cdot\lceil\sqrt{C}\rceil^2}
    }.
    \]
\end{restatable}

For many values of $C\geq 4$, choosing between the single-bucket algorithm and \Cref{thm:2bucketexact} improves the previously best known relative defect $1/\lfloor\sqrt C\rfloor$ of \cite{BarenboimE09}, with substantial improvements for several small values of $C$. For squares and $C=s^2+1$, the resulting bound matches this baseline; for $C=3$, it recovers the known relative defect $2/3$ of \cite{BalliuHLOS19}. \Cref{tab:short_defect_data} compares the previous bounds with the best upper bounds obtained by choosing between the single-bucket designs and all feasible integer splits in \Cref{thm:2bucketexact}. In particular, for $C=6,7,8$, our bounds are $3/7$, $9/23$, and $3/8$, respectively. Relative to the previous bound $1/2$, these reduce the guaranteed defect by approximately $14.3\%$, $21.7\%$, and $25\%$, respectively. The following theorem simplifies this upper bound.

\begin{restatable}{theorem}{restatetwobucketasymptotic}
\label{thm:2bucketasymptotic}
    By choosing between the Single-Bucket and the Two-Bucket Two-Sweep Algorithms, we can achieve a $C$-coloring with a relative defect $\rho$ satisfying
    \[
    \rho \leq \begin{cases}
    1/\sqrt{C} & \text{if $C=s^2$ for $s\in \mathbb{N}$,}\\
    1/\sqrt{C}+1/(2C^{3/2})+O(1/C^{5/2}) & \text{if $C=s^2\pm 1$ for $s\in \mathbb{N}$}\\
    1/\sqrt{C} + 1/(8C^{3/2}) + O(1/C^2) & \text{otherwise}.
    \end{cases}
    \]
\end{restatable}

 \begin{table}[ht]
\centering
\small % Slightly reduces the font size
\setlength{\tabcolsep}{4pt} % Reduces the horizontal space between columns
\renewcommand{\arraystretch}{1.4}
\begin{tabular}{|c|c|c|c|c|c|c|}
\hline
\textbf{C} & \textbf{previous bound} & \textbf{one-/two-bucket bound} & $\approx$ & \textbf{$C_1$} & \textbf{$C_2$} & \textbf{single bucket} \\ \hline
\textbf{3} & $2/3$ & $2/3$ & 0.667 & 0 & 3 & * \\ \hline
\textbf{4} & $1/2$ & $1/2$ & 0.500 & 4 & 0 & * \\ \hline
\textbf{5} & $1/2$ & $1/2$ & 0.500 & 5 & 0 & * \\ \hline
\textbf{6} & $1/2$ & $3/7$ & 0.429 & 2 & 4 &  \\ \hline
\textbf{7} & $1/2$ & $9/23$ & 0.391 & 2 & 5 &  \\ \hline
\textbf{8} & $1/2$ & $3/8$ & 0.375 & 0 & 8 & * \\ \hline
\textbf{9} & $1/3$ & $1/3$ & 0.333 & 9 & 0 & * \\ \hline
\textbf{10} & $1/3$ & $1/3$ & 0.333 & 10 & 0 & * \\ \hline
\textbf{11} & $1/3$ & $4/13$ & 0.308 & 6 & 5 &  \\ \hline
\textbf{12} & $1/3$ & $12/41$ & 0.293 & 5 & 7 &  \\ \hline
\textbf{13} & $1/3$ & $16/57$ & 0.281 & 4 & 9 &  \\ \hline
\textbf{14} & $1/3$ & $16/59$ & 0.271 & 3 & 11 &  \\ \hline
\textbf{15} & $1/3$ & $4/15$ & 0.267 & 0 & 15 & * \\ \hline
\textbf{16} & $1/4$ & $1/4$ & 0.250 & 16 & 0 & * \\ \hline
\end{tabular}
\caption{Upper bounds on the relative defect for a given number of colors $C$. The previous bound is $2/3$ for $C=3$~\cite{BalliuHLOS19} and $1/\lfloor\sqrt C\rfloor$ for $C\geq4$~\cite{BarenboimE09}. The one-/two-bucket bound is the best relaxed-LP bound within the analyzed Single-Bucket and Two-Bucket Two-Sweep constructions. The values $C_1$ and $C_2$ are the bucket sizes. An asterisk denotes a single-bucket construction. An extension of this table is given in \Cref{sec:full-table}.}
\label{tab:short_defect_data}
\end{table}

The results of \Cref{thm:2bucketexact,thm:2bucketasymptotic} also hold for the more general oriented defective coloring problem where the input graph is directed and the defect only limits the number of equally colored outneighbors. The relative defect then refers to the fraction of outneighbors that have the same color. When formally stating and analyzing the Two-Sweep Algorithm in the following, we directly use the oriented version as we need this when applying our list defective algorithm to get a fast proper $(\Delta+1)$-coloring algorithm in \CONGEST (cf.\ \Cref{thm:2sweepalgolists,thm:CONGESTcoloring}). Also note that in the first case of \Cref{thm:2bucketasymptotic}, i.e., where $C$ is a square number, the relative defect is optimal among Two-Sweep Algorithms satisfying the assumptions of \Cref{thm:2sweeplower}. For $C=s^2\pm1$, using three buckets slightly improves the bound; see \Cref{thm:threeBucketsImprovement}.
\section{The Two-Sweep Algorithm}
\label{sec:genericTwoSweep}

Assume that the input graph $G$ is equipped with an initial proper $q$-coloring and an edge orientation. The initial colors only schedule the sweeps (i.e., their labels need not belong to the output lists). Each node $v$ has a nonempty color list $L_v\subseteq\mathcal C$ and a defect function $d_v:L_v\to\mathbb N_0$: choosing $x\in L_v$ allows at most $d_v(x)$ equally colored outneighbors. A prescribed, locally known nonempty family $\varnothing\neq\mathcal S_v\subseteq2^{L_v}\setminus\{\varnothing\}$ specifies the possible palettes. In this section, the candidate family will be treated as part of the input; in our later implementations, we will present different approaches to construct $\mathcal{S}_v$.

\begin{figure}[ht]
  \centering
  \tikzset{every picture/.style={line width=0.75pt}} %set default line width to 0.75pt        

\begin{tikzpicture}[x=0.85pt,y=0.85pt, yscale=-1, xscale=1, scale=0.75]
%uncomment if required: \path (0,273); %set diagram left start at 0, and has height of 273

% Background shapes
%Shape: Ellipse [id:dp1808193399805469] 
\draw   (136,95.5) .. controls (136,52.15) and (164.09,17) .. (198.75,17) .. controls (233.41,17) and (261.5,52.15) .. (261.5,95.5) .. controls (261.5,138.85) and (233.41,174) .. (198.75,174) .. controls (164.09,174) and (136,138.85) .. (136,95.5) -- cycle ;

%Shape: Ellipse [id:dp41510974819766333] 
\draw   (406,98.5) .. controls (406,55.15) and (434.09,20) .. (468.75,20) .. controls (503.41,20) and (531.5,55.15) .. (531.5,98.5) .. controls (531.5,141.85) and (503.41,177) .. (468.75,177) .. controls (434.09,177) and (406,141.85) .. (406,98.5) -- cycle ;

% Arrows
%Straight Lines [id:da5120836814327626] 
\draw [shorten >= 4pt, shorten <= 4pt, ->, >=Stealth, line width=0.80pt]   (332.5,93) -- (203.25,55.25) ;

%Straight Lines [id:da5985151630745859] 
\draw [shorten >= 4pt, shorten <= 4pt, ->, >=Stealth, line width=0.80pt]   (332.5,93) -- (155.25,77.25) ;

%Straight Lines [id:da7905792394469741] 
\draw [shorten >= 4pt, shorten <= 4pt, ->, >=Stealth, line width=0.80pt]   (332.5,93) -- (167.5,99) ;

%Straight Lines [id:da7540478309550215] 
\draw [shorten >= 4pt, shorten <= 4pt, ->, >=Stealth, line width=0.80pt]   (332.5,93) -- (171.25,123.25) ;

%Straight Lines [id:da8211828387797624] 
\draw [shorten >= 4pt, shorten <= 4pt, ->, >=Stealth, line width=0.80pt]   (332.5,93) -- (215.25,137.25) ;

%Straight Lines [id:da2191444219648202] 
\draw [shorten >= 4pt, shorten <= 4pt, ->, >=Stealth, line width=0.80pt]   (332.5,93) -- (473.25,58.25) ;

%Straight Lines [id:da3817558231721184] 
\draw [shorten >= 4pt, shorten <= 4pt, ->, >=Stealth, line width=0.80pt]   (332.5,93) -- (425.25,80.25) ;

%Straight Lines [id:da5263291712633092] 
\draw [shorten >= 4pt, shorten <= 4pt, ->, >=Stealth, line width=0.80pt]   (332.5,93) -- (468.75,93) ;

%Straight Lines [id:da7501448263920905] 
\draw [shorten >= 4pt, shorten <= 4pt, ->, >=Stealth, line width=0.80pt]   (332.5,93) -- (508.25,118.25) ;

%Straight Lines [id:da4676048887339507] 
\draw [shorten >= 4pt, shorten <= 4pt, ->, >=Stealth, line width=0.80pt]   (332.5,93) -- (441.25,126.25) ;

%Straight Lines [id:da30586574071286643] 
\draw [shorten >= 4pt, shorten <= 4pt, ->, >=Stealth, line width=0.80pt]   (332.5,93) -- (471.25,152.25) ;

% Foreground circles
\begin{scope}
    %Shape: Circle [id:dp9123503355851562] 
    \draw  [fill={rgb, 255:red, 74; green, 144; blue, 226 }  ,fill opacity=1 ] (150,77.25) .. controls (150,74.35) and (152.35,72) .. (155.25,72) .. controls (158.15,72) and (160.5,74.35) .. (160.5,77.25) .. controls (160.5,80.15) and (158.15,82.5) .. (155.25,82.5) .. controls (152.35,82.5) and (150,80.15) .. (150,77.25) -- cycle ;

    %Shape: Circle [id:dp6209149062251162] 
    \draw  [color={rgb, 255:red, 0; green, 0; blue, 0 }  ,draw opacity=1 ][fill={rgb, 255:red, 74; green, 144; blue, 226 }  ,fill opacity=1 ] (220.4,136.21) .. controls (220.97,139.05) and (219.13,141.82) .. (216.29,142.4) .. controls (213.45,142.97) and (210.68,141.13) .. (210.1,138.29) .. controls (209.53,135.45) and (211.37,132.68) .. (214.21,132.1) .. controls (217.05,131.53) and (219.82,133.37) .. (220.4,136.21) -- cycle ;

    %Shape: Circle [id:dp9789724075002134] 
    \draw  [fill={rgb, 255:red, 74; green, 144; blue, 226 }  ,fill opacity=1 ] (198,55.25) .. controls (198,52.35) and (200.35,50) .. (203.25,50) .. controls (206.15,50) and (208.5,52.35) .. (208.5,55.25) .. controls (208.5,58.15) and (206.15,60.5) .. (203.25,60.5) .. controls (200.35,60.5) and (198,58.15) .. (198,55.25) -- cycle ;

    %Shape: Circle [id:dp8540792983724651] 
    \draw  [fill={rgb, 255:red, 74; green, 144; blue, 226 }  ,fill opacity=1 ] (166,123.25) .. controls (166,120.35) and (168.35,118) .. (171.25,118) .. controls (174.15,118) and (176.5,120.35) .. (176.5,123.25) .. controls (176.5,126.15) and (174.15,128.5) .. (171.25,128.5) .. controls (168.35,128.5) and (166,126.15) .. (166,123.25) -- cycle ;

    %Shape: Circle [id:dp055401415244140106] 
    \draw  [fill={rgb, 255:red, 227; green, 12; blue, 38 }  ,fill opacity=1 ] (325,93) .. controls (325,97.14) and (328.36,100.5) .. (332.5,100.5) .. controls (336.64,100.5) and (340,97.14) .. (340,93) .. controls (340,88.86) and (336.64,85.5) .. (332.5,85.5) .. controls (328.36,85.5) and (325,88.86) .. (325,93) -- cycle ;

    %Shape: Circle [id:dp4515696705732566] 
    \draw  [fill={rgb, 255:red, 74; green, 144; blue, 226 }  ,fill opacity=1 ] (163.12,95.68) .. controls (161.28,98.1) and (161.76,101.55) .. (164.18,103.38) .. controls (166.6,105.22) and (170.05,104.74) .. (171.88,102.32) .. controls (173.72,99.9) and (173.24,96.45) .. (170.82,94.62) .. controls (168.4,92.78) and (164.95,93.26) .. (163.12,95.68) -- cycle ;

    %Shape: Circle [id:dp8688748005132628] 
    \draw  [fill={rgb, 255:red, 126; green, 211; blue, 33 }  ,fill opacity=1 ] (420,80.25) .. controls (420,77.35) and (422.35,75) .. (425.25,75) .. controls (428.15,75) and (430.5,77.35) .. (430.5,80.25) .. controls (430.5,83.15) and (428.15,85.5) .. (425.25,85.5) .. controls (422.35,85.5) and (420,83.15) .. (420,80.25) -- cycle ;

    %Shape: Circle [id:dp7463031705679719] 
    \draw  [fill={rgb, 255:red, 126; green, 211; blue, 33 }  ,fill opacity=1 ] (476.4,151.21) .. controls (476.97,154.05) and (475.13,156.82) .. (472.29,157.4) .. controls (469.45,157.97) and (466.68,156.13) .. (466.1,153.29) .. controls (465.53,150.45) and (467.37,147.68) .. (470.21,147.1) .. controls (473.05,146.53) and (475.82,148.37) .. (476.4,151.21) -- cycle ;

    %Shape: Circle [id:dp7646880146471892] 
    \draw  [fill={rgb, 255:red, 126; green, 211; blue, 33 }  ,fill opacity=1 ] (513.4,117.21) .. controls (513.97,120.05) and (512.13,122.82) .. (509.29,123.4) .. controls (506.45,123.97) and (503.68,122.13) .. (503.1,119.29) .. controls (502.53,116.45) and (504.37,113.68) .. (507.21,113.1) .. controls (510.05,112.53) and (512.82,114.37) .. (513.4,117.21) -- cycle ;

    %Shape: Circle [id:dp7685755992504011] 
    \draw  [fill={rgb, 255:red, 126; green, 211; blue, 33 }  ,fill opacity=1 ] (468,58.25) .. controls (468,55.35) and (470.35,53) .. (473.25,53) .. controls (476.15,53) and (478.5,55.35) .. (478.5,58.25) .. controls (478.5,61.15) and (476.15,63.5) .. (473.25,63.5) .. controls (470.35,63.5) and (468,61.15) .. (468,58.25) -- cycle ;

    %Shape: Circle [id:dp007498907683794465] 
    \draw  [fill={rgb, 255:red, 126; green, 211; blue, 33 }  ,fill opacity=1 ] (468,58.25) .. controls (468,55.35) and (470.35,53) .. (473.25,53) .. controls (476.15,53) and (478.5,55.35) .. (478.5,58.25) .. controls (478.5,61.15) and (476.15,63.5) .. (473.25,63.5) .. controls (470.35,63.5) and (468,61.15) .. (468,58.25) -- cycle ;

    %Shape: Circle [id:dp8077109542426005] 
    \draw  [fill={rgb, 255:red, 126; green, 211; blue, 33 }  ,fill opacity=1 ] (436,126.25) .. controls (436,123.35) and (438.35,121) .. (441.25,121) .. controls (444.15,121) and (446.5,123.35) .. (446.5,126.25) .. controls (446.5,129.15) and (444.15,131.5) .. (441.25,131.5) .. controls (438.35,131.5) and (436,129.15) .. (436,126.25) -- cycle ;

    %Shape: Circle [id:dp41447311285276933] 
    \draw  [fill={rgb, 255:red, 126; green, 211; blue, 33 }  ,fill opacity=1 ] (463.25,93) .. controls (463.25,96.04) and (465.71,98.5) .. (468.75,98.5) .. controls (471.79,98.5) and (474.25,96.04) .. (474.25,93) .. controls (474.25,89.96) and (471.79,87.5) .. (468.75,87.5) .. controls (465.71,87.5) and (463.25,89.96) .. (463.25,93) -- cycle ;
\end{scope}

%Straight Lines [id:da5812144123660125] 
\draw [line width=2.25]    (102,211.5) -- (585,211.5) ;
\draw [shift={(589,211.5)}, rotate = 180] [color={rgb, 255:red, 0; green, 0; blue, 0 }  ][line width=2.25]    (17.49,-5.26) .. controls (11.12,-2.23) and (5.29,-0.48) .. (0,0) .. controls (5.29,0.48) and (11.12,2.23) .. (17.49,5.26)   ;

% Text Node
\draw (182,23) node [anchor=north west][inner sep=0.75pt]   [align=left] {$N_<(v)$};
% Text Node
\draw (451,27) node [anchor=north west][inner sep=0.75pt]   [align=left] {$N_>(v)$};
% Text Node
\draw (231,215) node [anchor=north west][inner sep=0.75pt]   [align=left] {increasing initial coloring};
% Text Node
\draw (341,71.75) node   [align=left] {\begin{minipage}[lt]{21.76pt}\setlength\topsep{0pt}
$v$
\end{minipage}};

\end{tikzpicture}
  \caption{The figure illustrates a node $v$ (in red) together with its outneighbors in $N_<(v)$ (in blue) and $N_>(v)$ (in green). Note that when node $v$ has to pick a set $S_v \in \mathcal{S}_v$ in Phase $I$, all blue nodes $u$ have already picked their sets $S_u \in \mathcal{S}_u$ and have sent them to $v$. In Phase $II$, $v$ has to output a final color from $S_v$. At this point, the green nodes have already picked their final colors. Node $v$ therefore has to pick its color based on knowing the sets $S_u$ for nodes in $N_<(v)$ and the final colors of nodes in $N_>(v)$.}
  \label{fig:sweepidea}
\end{figure}
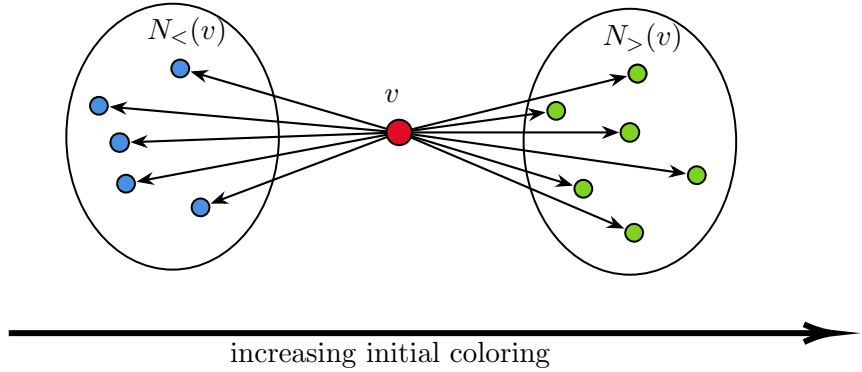

We note that throughout the algorithm, nodes do not consider their inneighbors to guide their decisions; in other words, nodes only look towards their outneighbors at all times. The algorithm proceeds in two phases that we call \textbf{Phase I} and \textbf{Phase II}. In both phases, we iterate (sweep) through the $q$ initial colors, in ascending order in Phase I and descending order in Phase II. Let $v$ be a node with color $c \in \{1, \ldots, q\}$ from the initial $q$-coloring. Then, $N_<(v)$ is defined as the set of outneighbors of $v$ that have an initial color $\in\{1, \ldots, c-1\}$ and similarly, $N_>(v)$ is the set of outneighbors that have an initial color $\in\{c+1, \ldots, q\}$. Note that by the definition of a proper coloring, there is no outneighbor that is also colored with color $c$ and thus $\beta_v = |N_<(v)| + |N_>(v)|$. This notation is illustrated in \Cref{fig:sweepidea}. The pseudocode of the algorithm is given in \Cref{alg:GenericSweep}. \newpage

\paragraph*{Phase I} In this phase, we iterate over the initial $q$-coloring of the nodes (in ascending order). In each step, all nodes $v$ sharing the same initial color will select the best palette $S_v \in \mathcal{S}_v$ --- with respect to their quality defined in \Cref{def:qualityGeneric} --- in parallel and send these sets to their neighbors. We call these nodes \emph{active}. For an active node, we define $k_v(x)$ for a color $x \in L_v$ as the number of outneighbors $u$ of $v$ in $N_<(v)$ such that $x \in S_u$. Note that all nodes in $N_<(v)$ have already selected their palettes when $v$ becomes active. The quality of a palette $S \in \mathcal{S}_v$ is formally defined as follows.
\begin{align}
    \label{def:qualityGeneric}
    Q_v(S) := \frac{1}{|S|} \left( |N_>(v)| + \sum_{x \in S} (k_v(x) - d_v(x))  \right)
\end{align}
For $x\in S$, the count $k_v(x)$ bounds the possible conflicts from $N_<(v)$. Negative terms $k_v(x)-d_v(x)$ are intentional: they record defect budget still available for conflicts with $N_>(v)$. Each outneighbor in $N_>(v)$ contributes at most one final conflict to the sum over $S$. Counting only outneighbors whose lists intersect $S$ could tighten this bound; we use $|N_>(v)|$, which requires no knowledge of their lists and suffices for our guarantees. Thus, $Q_v(S)$ bounds the average excess over the allowed defects in the Phase~II objective, and the minimum of that objective is no larger than its average.

In this generic Two-Sweep Algorithm we consider $\mathcal{S}_v$ to be part of the input. 
However, in our implementations in \Cref{sec:OLDC2Sweep} and \Cref{sec:SweepAndBuckets}, we give explicit constructions for $\mathcal{S}_v$ such that the set with the best quality (i.e., the minimum value) is guaranteed to meet necessary requirements (e.g., \Cref{eq:requirementPhaseIGeneric}, as shown later).
\paragraph*{Phase II} Every node $v\in V$ already holds its nonempty palette $S_v \subseteq L_v$. We now revisit the $q$ colors in \textit{reverse} (i.e., descending) order. When it is $v$'s turn, it selects its final color from $S_v$ and announces this choice to its neighbors. Consequently: (1) All outneighbors in $N_>(v)$ have already fixed their colors. (2) All outneighbors in $N_<(v)$ are still undecided about their final color. Let $r_v(x)$ denote the number of outneighbors in $N_>(v)$ that decided on color $x \in S_v$. We call node $v$ \emph{active in Phase II} in the iteration where $v$ has to decide on such a color. When $v$ is active, it chooses a color $x \in S_v$ that minimizes:
\begin{align}\label{eq:minimizePhaseIIGeneric}
    k_v(x) + r_v(x) - d_v(x). 
\end{align}

\begin{algorithm}[hbt]
    \caption{TwoSweep($q, \mathcal{S}_v$)} \label{alg:GenericSweep}
    \begin{algorithmic}[1]
        \Require Proper $q$-coloring, constraint arcs, input lists $L_v$, defect functions $d_v$,
        and nonempty candidate families
        $\varnothing\neq\mathcal{S}_v\subseteq 2^{L_v}\setminus\{\varnothing\}$.
        \Ensure An OLDC solution if every selected $S_v$ satisfies
        \Cref{eq:requirementPhaseIGeneric}; runtime $O(q)$ in \LOCAL.
        \State Send the initial color to all neighbors.
        \State \textit{/** Phase I **/}
        \For{$c=1,\ldots,q$}
            \State All nodes $v$ with initial color $c$ act in parallel.
            \State For all $x \in L_v$, let $k_v(x)$ be the number of outneighbors $u$ of $v$ in $N_<(v)$ with $x \in S_u$. 
            \State Select a palette $S_v \in \mathcal{S}_v$ as follows:
            \Statex \hspace{\algorithmicindent} $S_v$ minimizes the quality $Q_v(S_v)$ among all sets in $\mathcal{S}_v$. \Comment{\Cref{def:qualityGeneric}}
            \State Send $S_v$ to all neighbors.
        \EndFor
        \State \textit{/** Phase II **/}
        \For{$c=q,\ldots,1$}
            \State All nodes $v$ with initial color $c$ act in parallel.
            \State For all $x \in S_v$, let $r_v(x)$ be the number of outneighbors in $N_>(v)$ that already decided on color $x$.
            \State Compute a color $x_v \in S_v$ that minimizes
            \begin{align*}
                k_v(x_v) + r_v(x_v) - d_v(x_v).
            \end{align*}
            \State Decide on color $x_v$.
            \State Send $x_v$ to all neighbors.
        \EndFor
    \end{algorithmic}
\end{algorithm}

\subsection{Analysis of the generic Two-Sweep Algorithm}
\label{sec:genericAlgo}

We next show that $Q_v(S_v)<1$ ensures that the color chosen in Phase~II satisfies its defect bound.

\begin{lemma}\label{lemm:corectnessOfGenericSweep}
    Consider a node $v$ in Phase II. If
    \begin{align}
        \label{eq:requirementPhaseIGeneric}
        Q_v(S_v) < 1
    \end{align}
    holds, then the color $x_v \in S_v$ chosen by the algorithm satisfies
    \begin{align}
    \label{eq:ConditionPhase2Generic}
        k_v(x_v) + r_v(x_v) \leq d_v(x_v)
    \end{align}
    when $v$ becomes active in Phase $II$. Further, at most $d_v(x_v)$ outneighbors of $v$ will have decided on the same color after termination of Phase $II$.
\end{lemma}
\begin{proof}
    Observe that $\sum_{x \in S_v} r_v(x)$ is at most the number of outneighbors in $N_>(v)$. Combining this with \Cref{eq:requirementPhaseIGeneric}, we get
    \begin{align*}
        \sum_{x \in S_v} \left(r_v(x) + k_v(x) -d_v(x) \right)
        \leq |N_>(v)| + \sum_{x \in S_v} (k_v(x) - d_v(x))
        = |S_v| \cdot Q_v(S_v) < |S_v|.
    \end{align*}
    Hence, some color $y\in S_v$ satisfies
    $k_v(y)+r_v(y)-d_v(y)<1$. Let $x_v$ be the minimizer selected by the
    algorithm. Its objective is no larger than that of $y$ and is an integer;
    therefore, $k_v(x_v)+r_v(x_v)-d_v(x_v)\leq0$.

    At most $r_v(x_v)$ nodes in $N_>(v)$ have already chosen $x_v$, and no
    further node from this set becomes active. Among the nodes in $N_<(v)$,
    only the at most $k_v(x_v)$ nodes whose palettes contain $x_v$ can later
    choose it. Thus, at most $k_v(x_v)+r_v(x_v)\leq d_v(x_v)$ outneighbors of
    $v$ receive color $x_v$.
\end{proof}

The following lemma concludes the analysis of the algorithm. 

\begin{lemma}\label{lemm:VanillaGenericColoring}
    Consider a graph $G$ with an initial proper $q$-coloring and a given edge
    orientation. If every palette selected by \Cref{alg:GenericSweep}
    satisfies \Cref{eq:requirementPhaseIGeneric}, then the algorithm solves the
    oriented list defective coloring instance in $O(q)$ \LOCAL rounds. Its
    largest message has
    \[
        O\left(\max_{v\in V}|S_v|\cdot\log|\calC|+\log q\right)
    \]
    bits. Consequently, the same implementation works in \CONGEST whenever
    this quantity is $O(\log n)$.
\end{lemma}
\begin{proof}
    The correctness follows from \Cref{lemm:corectnessOfGenericSweep}. An
    initial exchange lets every node learn its neighbors' initial colors. Each
    phase then uses one round per initial color: active nodes send their palette
    in Phase~I and their final color in Phase~II. A palette contains at most
    $\max_v|S_v|$ colors from $\calC$; the initial and final colors require
    $O(\log q)$ and $O(\log|\calC|)$ bits, respectively. This proves both
    complexity bounds.
\end{proof}

\section{Solving OLDC instances using the Two-Sweep Algorithm}
\label{sec:OLDC2Sweep}

We now specify the candidate families for \Cref{alg:GenericSweep}. Let $p\geq1$ be a common integer parameter. If $|L_v|\leq p$, set $\mathcal S_v:=\{L_v\}$; otherwise, let $\mathcal S_v:=\binom{L_v}{p}$ consist of all $p$-element subsets of $L_v$. Thus, the selected palette has $\min\{p,|L_v|\}$ colors, and different input lists can yield different candidate families. Algorithmically, we just pick those sublists greedily:
if $|L_v|>p$, select the $p$ colors with the smallest values
$k_v(x)-d_v(x)$; otherwise select the whole list. We assume the slack condition in \Cref{eq:SlackInBaseCase}. A different choice of candidate families is studied in \Cref{sec:SweepAndBuckets}.

\begin{lemma}\label{lemm:VanillaColoring}
    Consider a graph $G$ with an initial proper $q$-coloring and a given edge orientation. Given that for an integer parameter $p \geq 1$ it holds
    \begin{align}
        \sum_{x \in L_v} (d_v(x)+1) > \max\left\{p, \frac{|L_v|}{p}\right\} \cdot \beta_v
    \end{align}
    for all nodes $v \in V$, this implementation of \Cref{alg:GenericSweep}
    solves an oriented list defective coloring instance in $O(q)$ \LOCAL
    rounds with messages of $O(p\log|\calC|+\log q)$ bits.
\end{lemma}
\begin{proof}
    Since every candidate palette is a nonempty subset of the corresponding
    list and has size at most $p$, it remains to verify
    \Cref{eq:requirementPhaseIGeneric}. Fix a node $v$. First suppose that
    $|L_v|\leq p$, and hence $S_v=L_v$. By construction,
    $\sum_{x\in L_v}k_v(x)\leq p|N_<(v)|$. Therefore,
    \begin{align*}
        |S_v| \cdot Q_v(S_v)
        &= |N_>(v)| + \sum_{x \in S_v}\bigl(k_v(x)-d_v(x)\bigr) \\
        &\leq \beta_v - |N_<(v)| + |N_<(v)| \cdot p
            - \sum_{x \in L_v} d_v(x) \\
        &= \beta_v + (p-1) \cdot |N_<(v)| - \sum_{x \in S_v} d_v(x) \\
        &\leq p \cdot \beta_v - \sum_{x \in S_v} d_v(x)\\
        &<\sum_{x \in L_v} (d_v(x)+1) - \sum_{x \in S_v} d_v(x) \\
        &= |S_v|.
    \end{align*}

    Now suppose that $\ell:=|L_v|>p$ and define
    $a_x:=k_v(x)-d_v(x)$. Since all candidate palettes have size $p$, the
    palette $S_v$ minimizing the quality consists of $p$ colors with the
    smallest values $a_x$. Consequently,
    \[
        \sum_{x\in S_v}a_x
        \leq \frac{p}{\ell}\sum_{x\in L_v}a_x.
    \]
    Moreover, $\sum_{x\in L_v}k_v(x)\leq p|N_<(v)|$. Set
    $m_v:=\max\{p,\ell/p\}$. Multiplying the resulting upper bound on
    $pQ_v(S_v)$ by $\ell/p$ gives
    \begin{align*}
        \frac{\ell}{p}\,pQ_v(S_v)
        &\leq \frac{\ell}{p}|N_>(v)|+p|N_<(v)|
            -\sum_{x\in L_v}d_v(x)\\
        &\leq m_v\beta_v-\sum_{x\in L_v}d_v(x)\\
        &<\sum_{x\in L_v}(d_v(x)+1)-\sum_{x\in L_v}d_v(x)
        =\ell.
    \end{align*}
    Hence, $pQ_v(S_v)<p=|S_v|$. In both cases,
    \Cref{eq:requirementPhaseIGeneric} holds, and the claim follows from
    \Cref{lemm:VanillaGenericColoring}.
\end{proof}

The preceding lemma gives the $O(q)$-round baseline under \Cref{eq:SlackInBaseCase}. Unfortunately, when $\eps=0$ we must sweep over all colors of a proper $q$-coloring. Relying on computing a proper coloring makes this approach too slow for OLDC instances. Even with highly efficient algorithms to compute an initial $\Theta(\Delta)$-coloring, the round complexity of \Cref{alg:GenericSweep} would still be of order $\Omega(\Delta)$ and thus far from our $\tilde{O}(\sqrt{\Delta})$ target outlined in \Cref{thm:CONGESTcoloring}. \par
To achieve the desired speed-up, we instead use a defective coloring. This strategy employs far fewer colors for the initial coloring at the acceptable cost of allowing monochromatic arcs. For a given parameter $0<\alpha\leq1$, a coloring with $O(1/\alpha^2)$ colors in which every node has at most an $\alpha$-fraction of its outneighbors of the same color can be computed in $O(\log^* q)$ rounds from the initial $q$-coloring (cf.\ \cite{Kuhn2009,KawarabayashiS18} and \Cref{lemm:defColorBlackBox}). We will later show that the additional $(1+\eps)$-factor in the slack, combined with this defective coloring algorithm using $\alpha := \eps/p$, achieves a sufficient trade-off between the size of the initial coloring and the monochromatic arcs. \par
The detailed implementation for $0<\eps\leq p$ is provided in \Cref{alg:Sweep2}.

\begin{algorithm}
    \caption{FastTwoSweep($q, p, \eps$)} \label{alg:Sweep2}
    \begin{algorithmic}[1]
        \Require Proper $q$-coloring of the nodes. Integer $p \geq 1$ and
        $0<\varepsilon\leq p$. For each node $v \in V$, \Cref{eq:SlackInSweep} holds.
        \Ensure OLDC in time $O(\min \{q, (p/\eps)^2 + \log^* q\})$
        \If{$q \leq \frac{p^2}{\varepsilon^2} + \log^* q$}
            \State If $|L_v| \leq p$, let $\mathcal{S}_v := \{L_v\}$, otherwise $\mathcal{S}_v := \binom{L_v}{p}$.
            \State Solve the OLDC instance using our TwoSweep$(q, \mathcal{S}_v)$ \Comment{\Cref{alg:GenericSweep}}
        \Else
            \State Compute a defective coloring $\Psi:V\to[q']$ using \Cref{lemm:defColorBlackBox}
            \Statex \hspace{\algorithmicindent}with $\alpha := \varepsilon/p$ and $q':=O(p^2/\varepsilon^2)$.
            \State Let $G'$ be obtained from $G$ by removing the edges monochromatic under $\Psi$.
            \State Define 
            \begin{align*}
                d_v'(x) := d_v(x) - \left\lfloor \beta_v \cdot \frac{\varepsilon}{p}\right\rfloor \text{ and }
                L_v' := \{ x \in L_v \mid d_v'(x) \geq 0\}
            \end{align*} 
            \State Define OLDC instance with defects $d_v'(x)$ and color lists $L_v'$ on $G'$. 
            \State If $|L_v'| \leq p$, let $\mathcal{S}_v := \{L_v'\}$, otherwise $\mathcal{S}_v := \binom{L_v'}{p}$.
            \State Solve this OLDC instance on $G'$:
            \Statex \hspace{\algorithmicindent} Use our TwoSweep($q'$, $\mathcal{S}_v$) \Comment{\Cref{alg:GenericSweep}}
        \EndIf
    \end{algorithmic}
\end{algorithm}

\subsection{Working with a nonproper initial coloring}\label{sec:fastersweep}
The weakness of the previous algorithm is the dependence of the complexity on the $q$-coloring. We do not have sufficiently fast (proper) coloring algorithms to use here. However, we overcome this problem by using a defective coloring as described in \Cref{lemm:defColorBlackBox} instead of a proper coloring, and we show that monochromatic edges (induced by this defective coloring) can be safely ignored under a slightly increased slack (cf.\ \Cref{eq:SlackInSweep}).

\begin{lemma}[Directed adaptation of \cite{Kuhn2009,KawarabayashiS18}]
    \label{lemm:defColorBlackBox}
    For every $0<\alpha\leq1$, a directed graph with an initial proper
    $q$-coloring admits a deterministic $O(1/\alpha^2)$-coloring in which each
    node $v$ has at most $\alpha\beta_v$ outneighbors of its own color. The
    coloring can be computed in $O(\log^*q)$ rounds using
    $O(\log q+\log(1/\alpha))$-bit messages. In particular, this is a
    \CONGEST algorithm whenever $q,1/\alpha\leq\poly(n)$.
\end{lemma}
The statement follows almost directly from the polynomial color-reduction and
relative-defect arguments in \cite{Kuhn2009,KawarabayashiS18}. For
completeness, and to record the directed adaptation, the dependence on a
possibly nonconstant $\alpha$, and the message-size bound, we provide a full
proof in \Cref{sec:directedDefectiveColoringProof}.

Let $0<\eps\leq p$ be a real parameter and let $p\geq1$ be an integer as before. Subsequently, we consider OLDC instances where each node~$v$ satisfies the following property.
\begin{align}\label{eq:SlackInSweep}
    \sum_{x \in L_v} (d_v(x) + 1) > \left( 1 + \eps \right) \cdot \max\left\{p, \frac{|L_v|}{p} \right\} \cdot \beta_v
\end{align}
Assume we have a defective coloring given on $G$ such that every node $v$ has at most $\varepsilon/p \cdot \beta_v$ outneighbors of the same color. To run \Cref{alg:GenericSweep}, we remove all edges whose endpoints have the same defective color. The resulting subgraph $G'$ is properly colored. We adjust the defects and color lists as specified in \Cref{alg:Sweep2}, thereby reserving enough defect for the removed edges. We can now prove our first main theorem.

\restatefirst*
\begin{proof}
    The case $\eps=0$ is \Cref{lemm:VanillaColoring}. If $\eps\ge p$, the slack condition implies $\sum_{x\in L_v}(d_v(x)+1)>|L_v|\beta_v$, so each node can choose a color $x\in L_v$ with $d_v(x)\ge\beta_v$ without communication. Hence, assume $0<\eps<p$.
    Apply \Cref{lemm:defColorBlackBox} with
    $\alpha=\varepsilon/p\leq1$. This takes $O(\log^*q)$ rounds and produces
    $O(p^2/\varepsilon^2)$ colors. Let
    \[
        t_v:=\left\lfloor\frac{\varepsilon\beta_v}{p}\right\rfloor,
        \qquad
        m_v:=\max\left\{p,\frac{|L_v|}{p}\right\},
    \]
    and define $d_v'(x)=d_v(x)-t_v$ and
    $L_v'=\{x\in L_v:d_v'(x)\geq0\}$ as in \Cref{alg:Sweep2}. Let
    $\beta_v'$ denote the outdegree of $v$ in $G'$. If $x\notin L_v'$, then
    integrality gives $d_v(x)+1-t_v\leq0$. Hence,
    \begin{align*}
        \sum_{x\in L_v'}(d_v'(x)+1)
        &\geq \sum_{x\in L_v}(d_v(x)+1-t_v)\\
        &=\sum_{x\in L_v}(d_v(x)+1)-|L_v|t_v\\
        &>(1+\varepsilon)m_v\beta_v
            -\varepsilon\frac{|L_v|}{p}\beta_v\\
        &\geq m_v\beta_v\\
        &\geq \max\left\{p,\frac{|L_v'|}{p}\right\}\beta_v'.
    \end{align*}
    The last two inequalities use $m_v\geq |L_v|/p$,
    $|L_v'|\leq|L_v|$, and $\beta_v'\leq\beta_v$. Thus, the reduced instance
    satisfies \Cref{eq:SlackInBaseCase} (in particular, $L_v'\neq\varnothing$) and is solved by
    \Cref{lemm:VanillaColoring} in $O(p^2/\varepsilon^2)$ rounds.

    It remains to verify the original defects. Let $h_v$ be the number of
    removed outgoing edges at $v$. By \Cref{lemm:defColorBlackBox},
    $h_v\leq\varepsilon\beta_v/p$. Since $h_v$ is an integer, $h_v\leq t_v$.
    If $v$ receives color $x$, it therefore has at most
    $d_v'(x)+h_v\leq d_v(x)$ equally colored outneighbors in $G$. Taking the
    faster of the direct and reduced implementations gives the claimed
    $O(\min\{q,p^2/\varepsilon^2+\log^*q\})$ runtime.

    The direct implementation exchanges initial colors, palettes, and
    final colors, using $O(\log q+p\log|\mathcal C|)$ bits per message
    by \Cref{lemm:VanillaColoring}. The fast branch additionally sends
    intermediate colors using $O(\log q+\log(p/\varepsilon))$ bits
    by \Cref{lemm:defColorBlackBox}. This branch is used only when
    $q>(p/\varepsilon)^2+\log^*q$, which implies
    $\log(p/\varepsilon)=O(\log q)$. Its initial sweep coloring has
    $O((p/\varepsilon)^2)$ colors, so its encoding also fits this bound.
    Thus every message in either branch has the claimed size.
\end{proof}

We can now use our Two-Sweep Algorithm in a recursive manner to compute a $(\Delta+1)$-coloring in the \CONGEST model as stated in \Cref{thm:CONGESTcoloring}. For that we use a known (\cite{fraigniaud16local,Kuhn20,FK23}) technique, called color-space reduction, on \Cref{thm:2sweepalgolists}. The details appear in the subsequent section.

%%% Local Variables:
%%% mode: latex
%%% TeX-master: "../main"
%%% End:

\section{Applications of the Two-Sweep Algorithm} \label{sec:Two-SweepInCONGEST}

\paragraph*{Oriented List Defective Coloring in CONGEST} In this section, we \emph{reduce} the maximum message size of our algorithm using the color space reduction technique as stated in \Cref{lemm:colorSpaceReductionBox}. The lemma is the strict-slack specialization of the more general Theorem 3 in \cite{FK23}; its proof in \Cref{sec:strictSlackColorSpaceReductionProof} adapts the proof of Theorem~1.2 in \cite{FK23full}.

\begin{samepage}
\begin{lemma}[Strict-slack specialization of Theorem 3 in \cite{FK23}] \label{lemm:colorSpaceReductionBox}
    Let $\kappa$, $M$, and $T$ be nondecreasing functions of the maximum list size $\Lambda$, where $\kappa(\Lambda)>0$. Assume we are given a deterministic algorithm $\calA$ solving every oriented list defective coloring instance with nonempty lists for which each node $v \in V$ satisfies
    \begin{align*}
        \sum_{x \in L_v} (d_v(x) + 1) > \beta_v \cdot \kappa(\Lambda)
    \end{align*}
    with round complexity $T(\Lambda)$ and messages of at most $M(\Lambda)$ bits. Then, for every ambient color-space size $C$ and every integer $2\leq\lambda\leq C$, there exists a deterministic algorithm $\calA'$ that solves every oriented list defective coloring instance with nonempty lists whose colors come from that space and for which each node $v \in V$ satisfies
    \begin{align*}
        \sum_{x \in L_v} (d_v(x) + 1) > \beta_v \cdot \kappa(\lambda)^{\lceil \log_\lambda C \rceil}
    \end{align*}
    in time $O((T(\lambda)) \cdot \lceil \log_\lambda C \rceil)$ using messages of $O(M(\lambda)+\log\lambda)$ bits, assuming a common encoding of the ambient colors and partition labels.
\end{lemma}
\end{samepage}

The formal proof of \Cref{lemm:colorSpaceReductionBox} is given in \Cref{sec:strictSlackColorSpaceReductionProof}.

Using our \emph{Fast-Two-Sweep} routine (\Cref{alg:Sweep2}) as the deterministic oriented list defective coloring algorithm $\calA$ in \Cref{lemm:colorSpaceReductionBox} and choosing suitable parameters, we can drastically reduce the dependence on $p$ and $|L_v|$ in the number of rounds and the maximum message size. In particular, the size of messages drops to $O(\log q+\log(C))$ bits, making the algorithm viable in \CONGEST whenever both the initial color space and the color space are $\poly(n)$. The detailed proof follows.

\restatesecond*
\begin{proof}[Proof of \Cref{thm:OLDC2}]
    Since the case that $C = 1$ is trivial, first suppose that $2\leq C \leq3$. In this case, by the assumption of the statement we have $\sum_{x\in L_v} (d_v(x) + 1) \geq 4 \beta_v$ and thus, one of the (at most three) colors in $L_v$ must satisfy that $d_v(x) + 1 \geq \lceil 4/3 \cdot \beta_v \rceil$ (as defects and outdegrees are integers) and thus $d_v(x) \geq \beta_v$ (for $\beta_v=0$, this follows directly from nonnegative defects). This color can be assigned to $v$ without violating any constraint and we are done. Subsequently, assume $C \geq 4$. \par

    Now consider \Cref{thm:2sweepalgolists}. We set the parameter $p := \left\lceil \sqrt{\Lambda} \right\rceil$, where $\Lambda$ is the maximum list size as in \Cref{lemm:colorSpaceReductionBox}, and we set $\eps := \frac{1}{3 \lceil \log_4 C \rceil}$. We get an algorithm that solves a given OLDC instance for which 
    \begin{align*}
        \sum_{x \in L_v} (d_v(x) + 1) > \beta_v \cdot \left( 1+ \frac{1}{3 \lceil \log_4 C \rceil} \right) \cdot \left\lceil \sqrt{\Lambda } \right\rceil
    \end{align*}   
    holds for every node $v \in V$. The round complexity is $O(\Lambda \cdot \log^2 C + \log^* q)$ and messages are of size at most $O(\log q + \sqrt{\Lambda} \cdot \log C)$ bits. These bounds include initial and intermediate preprocessing colors, palettes of at most $p$ colors, and final output colors, by the full message accounting in \Cref{thm:2sweepalgolists}. This algorithm (with those parameters) is used as algorithm $\calA$ when applying the strict-slack version of \Cref{lemm:colorSpaceReductionBox}. Here we set the splitting parameter $\lambda = 4$. The outcome is an OLDC algorithm $\calA'$ with an initial slack requirement of
    \begin{align} \label{eq:afterColorSpaceReduction}
        \sum_{x \in L_v} (d_v(x) + 1) > \beta_v \cdot \left( \left( 1+ \frac{1}{3 \cdot \lceil \log_4 C \rceil} \right) \cdot 2 \right)^{ \lceil \log_4 C \rceil}
    \end{align}   
    for every node $v \in V$. Because \Cref{lemm:colorSpaceReductionBox}
    invokes the base algorithm once per level, its round complexity is
    \begin{align*}
        O\left(\bigl(\log^2 C+\log^*q\bigr)
            \left\lceil\log_4 C\right\rceil\right)
        =O\bigl(\log^3(C)+\log(C)\cdot\log^*q\bigr).
    \end{align*}
    It uses messages of size at most $O(\log q + \log C)$ bits. It remains to
    verify that the hypothesis of the theorem implies the strict condition
    in \Cref{eq:afterColorSpaceReduction}. We have
    \begin{align*}
        \left( \left( 1+ \frac{1}{3 \lceil \log_4 C \rceil} \right) \cdot 2 \right)^{ \lceil \log_4 C \rceil} &\leq 2 \cdot \left( 1+ \frac{1}{3 \lceil \log_4 C \rceil} \right)^{\lceil \log_4 C \rceil} \cdot 2^{\log_4 C} \\
        &< 2 \cdot e^{1/3} \cdot \sqrt{C} < 3 \sqrt{C}.
    \end{align*}
    If $\beta_v>0$, the theorem's non-strict lower bound with factor
    $3\sqrt C$ therefore implies the strict inequality in
    \Cref{eq:afterColorSpaceReduction}. If $\beta_v=0$, the same inequality
    follows from $L_v\neq\varnothing$, because then
    $\sum_{x\in L_v}(d_v(x)+1)\geq1>0$. Hence, the strict-slack reduction
    applies at every node.
\end{proof}

Using \Cref{thm:OLDC2} as OLDC subroutine of Theorem 4 in \cite{FK23}, we directly get an efficient $(\mathit{degree}+1)$-list coloring algorithm for the \CONGEST model. The details appear in the proof of \Cref{thm:CONGESTcoloring}.

\restatethird*
\begin{proof}[Proof of \Cref{thm:CONGESTcoloring}]
    In this given coloring instance all defects are zero, i.e., we have a list defective coloring instance where for all nodes $v \in V$ it holds $\sum_{x \in L_v} (d_v(x) + 1) = |L_v| \geq \deg(v) + 1$. Thus, we can apply Theorem 4 from \cite{FK23} with $\nu=0$ and $\kappa=3\sqrt C$, where we plug in \Cref{thm:OLDC2} as list defective coloring algorithm.
    This results in an algorithm that computes the required $(\mathit{degree}+1)$-list coloring in time
    \begin{align*}
        O(\sqrt{C} \cdot \log \Delta \cdot T + \log^* n )
    \end{align*}
    where $T$ is the round complexity from \Cref{thm:OLDC2}. Since they start by computing an initial $O(\Delta^2)$ coloring (e.g. by using Linial's coloring algorithm \cite{Linial1987}), we can also use this as the $q$-coloring in \Cref{thm:OLDC2}. The theorem's assumption $C = O(\Delta)$ and the bound $q=O(\Delta^2)$ give
    \begin{align*}
        T
        &=O\bigl(\log^3(C)+\log(C)\cdot\log^*(\Delta^2)\bigr)\\
        &=O(\log^3\Delta),
    \end{align*}
    where we use $\log^*(\Delta^2)=O(\log\Delta)$. Plugging this
    value of $T$ into the round complexity above concludes the analysis of the
    round complexity. It remains to argue that the algorithm uses messages of
    size at most $O(\log n)$. Theorem 4 from \cite{FK23} states that if each message of the
    given algorithm has at most $B$ bits, the resulting
    algorithm uses messages of size at most $O(B + \log n)$ bits. Due to
    \Cref{thm:OLDC2}, this leads to messages of size at most
    $O(\log q + \log C + \log n)$ bits. As we have $q = O(\Delta^2)$ and
    $C = O(\Delta)$, this algorithm is viable in \CONGEST.
\end{proof}

 \paragraph*{Comparison to the Results of \cite{FK23,MausT20}}
Before the present work, the distributed algorithms to compute oriented list defective colorings and to compute a proper $(\Delta+1)$-coloring in $\tilde{O}(\sqrt{\Delta}) + O(\log^* n)$ \CONGEST rounds appeared in \cite{FK23}, which was based on techniques developed in \cite{MausT20}. It therefore makes sense to compare \Cref{thm:2sweepalgolists,thm:OLDC2,thm:CONGESTcoloring} with the results of \cite{FK23,MausT20}.

We start by discussing our distributed OLDC algorithms. First, our Two-Sweep Algorithm is essentially a greedy algorithm, and in our opinion, it is conceptually significantly simpler than the algorithm of \cite{FK23,MausT20}. In the algorithm of \cite{FK23,MausT20}, each node $v$ has to compute an appropriate set of size at most $2^{2^{|L_v|}}$, consisting of sets of subsets of its list $L_v$. This collection of sets is computed in a non-trivial way as a function of the node's initial color and of $L_v$. To simplify the comparison with \cite{FK23}, we consider uniform
defects $0<d<\beta$, where $\beta$ is the maximum outdegree, and
assume that the supplied initial proper $q$-coloring satisfies
$q\leq\poly(\beta)$.
Under these assumptions, the uniform-defect routine of \cite{FK23}
has a sufficient condition requiring lists of size
$\Omega\big(\big(\frac{\beta}{d+1}\big)^2
\cdot(\log\beta+\log\log C)\big)$,
while our Two-Sweep Algorithm (\Cref{thm:2sweepalgolists})
admits the following parameter choice. Put $r:=\beta/(d+1)$ and fix $\eps\in[0,1]$. Choosing $p=\lfloor(1+\eps)r\rfloor+1$ and lists of size at least $p^2$ gives $p>(1+\eps)r$ and $|L_v|\geq p^2>(1+\eps)pr$, satisfying the strict hypothesis of \Cref{thm:2sweepalgolists}. Thus, the sufficient list size is $O(r^2)=O\big(\big(\frac{\beta}{d+1}\big)^2\big)$. If we further assume that the size of the color space $C$ is proportional to the maximum list size, then also the recursive Two-Sweep Algorithm (\Cref{thm:OLDC2})  uses lists of size $O(r^2)$.
 The tradeoff in rounds and communication is summarized below for a
supplied proper $q$-coloring. The round bound for
\Cref{thm:2sweepalgolists} in the following table uses the choice $\eps=1$.
The uniform-defect list-size comparison above is a comparison of
sufficient conditions, not a necessity result for either algorithm.
\begin{center}
\small
\begin{tabularx}{\linewidth}{@{}p{0.20\linewidth}p{0.29\linewidth}X@{}}
\toprule
Routine & Rounds & Bits per message \\
\midrule
\cite{FK23} & $O(\log\beta)$ &
$O(\min\{C,\Lambda\log C\}+\log\beta+\log q)$ \\
\addlinespace
\Cref{thm:2sweepalgolists} &
$O(\min\{q,p^2+\log^*q\})$ &
$O(\log q+p\log C)$ \\
\addlinespace
\Cref{thm:OLDC2} &
$O(\log^3C+\log C\log^*q)$ &
$O(\log q+\log C)$ \\
\bottomrule
\end{tabularx}
\end{center}
For polynomial $C,q$ in $\beta$, the recursive Two-Sweep bound is
$O(\log^3\beta)$ rather than the $O(\log\beta)$ bound of~\cite{FK23}.
The benefit is the simpler palette selection and the sufficient
list-size guarantee described above.
Although the internal computational complexity of a distributed algorithm is usually not studied, it might still be relevant in practical applications.  The two sweeps do not require enumerating the candidate family:
node $v$ accumulates palette-membership counts and selects the
$\min\{p,|L_v|\}$ colors with smallest scores $k_v(x)-d_v(x)$.
Using sorted lists and binary search, their local work is
$O((|L_v|+p\deg(v))\log(|L_v|+1))$ arithmetic and comparison operations
on color labels, counts, and the supplied defect values. Bit complexity
also depends on the encodings of those values.
When the accelerated branch uses defective-coloring preprocessing, it has
a deterministic implementation
with local computation polynomial in $q$, $p/\eps$, and $\deg(v)$;
see \Cref{sec:directedDefectiveColoringProof}.
In the recursive application, the base calls use auxiliary output color spaces of size at most four,
$p\leq2$, and $p/\eps=O(\log C)$. Thus for $q=O(\Delta^2)$ and
$C=O(\Delta)$, these Two-Sweep calls have polynomial local computation
in $\Delta$ once the initial proper coloring is available.  The computational complexity of the algorithm of \cite{FK23} is discussed in the appendix of \cite{FK23full} (the full version of \cite{FK23}). For the parameter choices of the underlying construction of \cite{MausT20, FK23}, an implementation that explicitly enumerates the candidate family requires superpolynomial local computation in $n$ when the color lists have size $\Omega(\log n)$. The method described in the appendix of \cite{FK23full} reduces the local computational cost at the price of additional $\log\Delta$ factors in the round complexity.

We next also compare the resulting $\tilde{O}(\sqrt{\Delta})+O(\log^* n)$-round \CONGEST algorithm for the $(\Delta+1)$-coloring problem with the previous such algorithm from \cite{FK23}. The degree-dependent term in our stated round bound is larger by a factor of $\log^2\Delta/(\log\log\Delta)^6$ than the corresponding term in \cite{FK23}. However, by replacing the oriented list defective coloring algorithm in \cite{FK23} with the algorithm of \Cref{thm:OLDC2}, as discussed above, our algorithm becomes simpler and computationally significantly more efficient than the algorithm of \cite{FK23}. Moreover, while the appendix of the full version \cite{FK23full} (of \cite{FK23}) outlines a method to reduce the internal node complexity to $o(n)$, with their suggested parameter choice, the resulting round bound incurs additional $\log \Delta$ factors\footnote{Choosing $\eps=1/6$, as suggested in Appendix~C of \cite{FK23full}, increases the degree-dependent term in their original round bound by a factor of $\Omega(\log^2\Delta)$.} and therefore exceeds our bound stated in \Cref{thm:CONGESTcoloring}. Furthermore, the OLDC-based component of the algorithm of \cite{FK23} only works in the \CONGEST model as long as $\Delta\leq \poly\log n$. For larger $\Delta$, they use the algorithm of \cite{GhaffariKuhn21}, which has a round complexity of $O(\log^2\Delta\cdot \log n)$. Our algorithm from \Cref{thm:CONGESTcoloring} is a $\tilde{O}(\sqrt{\Delta})+O(\log^* n)$-round coloring algorithm in \CONGEST for the whole range of degrees.

Finally, we would like to point out one additional reason why our OLDC algorithm and the resulting proper $(\Delta+1)$-coloring algorithm might be useful: low local computational cost is also important when transforming distributed algorithms into parallel ones in models such as the standard PRAM model. We believe that the Two-Sweep Algorithm is simple enough to serve as a building block for an efficient PRAM algorithm. In particular, in the setting of our $(\Delta+1)$-coloring application with $\Delta\leq\poly\log n$, the local counting for a single initial-color class can be implemented in $O(\log\Delta)$ parallel time using $\poly(\Delta\log n)$ processors per node, hence $\tilde{O}(n)$ processors in total. This observation concerns the local computation; obtaining a bound on the depth of the full algorithm also requires handling the sequential dependencies within the sweeps and the outer reductions. Nevertheless, we hope that a suitable parallel implementation can compute a $(\Delta+1)$-coloring in $O(\poly\log\Delta+\log^*n)$ parallel time and $\tilde{O}(n+m)$ work in this degree regime, where $m=|E|$. Developing such an implementation and a complete depth and work analysis is beyond the scope of this paper and left for future research. This would also require parallel implementations of the other distributed building blocks that we use; see \cite{ElkinK26CoRR} for related recent work on parallel coloring. We do not currently see an equally direct approach based on the previous distributed $\tilde{O}(\sqrt\Delta)+O(\log^*n)$-round $(\Delta+1)$-coloring algorithms.

\section{Optimizing the defect for \boldmath\texorpdfstring{$C$}{C} Colors}
\label{sec:SweepAndBuckets}
\enlargethispage{3pt}
In this section, we study the relative defect achievable by the Two-Sweep
Algorithm as a function of the number $C$ of colors. We specialize
\Cref{alg:GenericSweep} to the common color list
$L_v=\mathcal C:=\set{1,\dots,C}$ and use a fixed nonempty candidate family
$\calS=\set{S_1,\dots,S_k}$ of nonempty subsets of $\mathcal C$ at every
node. We seek the same target relative defect $\rho\in[0,1]$ for every color,
with the asymptotic guarantee in \Cref{eq:defRelativeDefect}, and study how
the choice of $\calS$ affects the achievable guarantee.

At nodes of positive degree, we normalize local neighbor counts by
$\deg(v)$\footnote{For oriented graphs,
the same definitions use the outdegree $\beta_v$ in place of $\deg(v)$.}. In
Phase~I, $\alpha_S$ is the number of neighbors in $N_<(v)$ that have already
announced $S$, divided by $\deg(v)$. We write
$\vec\alpha=(\alpha_S)_{S\in\calS}$ for these normalized counts.
Isolated nodes have defect zero and may choose an arbitrary color.

\subsection{Normalized quality and the actual defect}
Write $d_v$ for the common value of $d_v(x)$ at node $v$. Since this
value is independent of $x$, it contributes the same additive constant
to every candidate within each phase and does not affect either minimization.
For a node $v$ of positive degree, put
\[
    \alpha:=\sum_{A\in\mathcal S}\alpha_A,
    \qquad \alpha_x:=\sum_{A\in\mathcal S:x\in A}\alpha_A
        =\frac{k_v(x)}{\deg(v)}.
\]
The normalized quality of a candidate palette $S\in\mathcal S$ is
\begin{align}
    Q_{\vec\alpha}(S)
    &:=\frac{1-\alpha+\sum_{x\in S}\alpha_x}{|S|}\notag\\
    &=\frac1{|S|}\left(1+\sum_{A\in\mathcal S}
        (|S\cap A|-1)\alpha_A\right).
    \label{eq:qualityOfSet}
\end{align}
Here $1-\alpha=|N_>(v)|/\deg(v)$, so the quality in
\Cref{def:qualityGeneric} satisfies the exact identity
\[
    Q_v(S)=\deg(v)Q_{\vec\alpha}(S)-d_v.
\]
Thus Phase~I minimizes $Q_{\vec\alpha}$ and Phase~II minimizes
$k_v(x)+r_v(x)$. We analyze the actual relative defect produced by these
choices.

\begin{lemma}[Quality upper-bounds the actual defect]
    \label{eq:NearEqualityOfDefectAndQuality}
    Let $S_v$ be the set chosen by a node $v$ of positive degree in
    Phase~I, and let $\rho_v$ denote its final relative defect.  Then
    \begin{align*}
        \rho_v
        \leq Q_{\vec{\alpha}}(S_v) = \min_{S\in\mathcal S}Q_{\vec\alpha}(S).
    \end{align*}
\end{lemma}
\begin{proof}
    For $x\in\mathcal C$, let $h_v(x)$ be the number of neighbors in
    $N_<(v)$ that eventually choose $x$, and let $r_v(x)$ be the number of
    neighbors in $N_>(v)$ that choose $x$. Since a node counted by
    $h_v(x)$ must have announced a set containing $x$, we have
    $h_v(x)\leq k_v(x)$. In Phase~II, node $v$ chooses a color
    $x_v\in S_v$ minimizing $k_v(x)+r_v(x)$. Hence,
    \begin{align*}
        \rho_v
        &=\frac{h_v(x_v)+r_v(x_v)}{\deg(v)}\\
        &\leq \frac{k_v(x_v)+r_v(x_v)}{\deg(v)}\\
        &\leq \frac{1}{|S_v|\deg(v)}
        \sum_{x\in S_v}\bigl(k_v(x)+r_v(x)\bigr)\\
        &\leq \frac{1}{|S_v|\deg(v)}
        \left(\sum_{x\in S_v}k_v(x)+|N_>(v)|\right)
        =Q_{\vec{\alpha}}(S_v).
    \end{align*}
    Here, the last inequality uses
    $\sum_{x\in S_v}r_v(x)\leq|N_>(v)|$.  The equality in the lemma statement
    follows because $S_v$ minimizes $Q_{\vec\alpha}$.
\end{proof}

We first establish the lower bound of \Cref{thm:2sweeplower}; the analysis of
candidate families giving upper bounds follows in \Cref{sec:towardsUpperBound}.

\subsection{Two-Sweep lower bound on the defect}
\label{sec:lowerBound}
For the lower bound, we consider undirected graphs and first allow arbitrary
nonnegative earlier-neighbor masses $\alpha_S$ with total
$\alpha:=\sum_{S\in\calS}\alpha_S\leq1$. The remaining mass $1-\alpha$
represents later neighbors. For each color $x\in\mathcal C$, define
$\alpha_x:=\sum_{S\in\calS:x\in S}\alpha_S$, and let $\delta_x$ denote the
mass of later neighbors choosing $x$. In an execution,
$\alpha_x=k_v(x)/\deg(v)$ and $\delta_x=r_v(x)/\deg(v)$, so the bound used in
\Cref{eq:NearEqualityOfDefectAndQuality} becomes $\alpha_x+\delta_x$.
The following lemma characterizes when a candidate palette guarantees that this
bound is at most $\rho$ for some color, for every allocation of
later-neighbor colors.

\begin{lemma}[Feasible candidate palettes]
    \label{lem:goodset}
    Let $\alpha_S\geq 0$ for all $S\in\calS$, let $\alpha:=\sum_{S\in\calS}\alpha_S<1$, and define $\alpha_x$ as above. A nonempty candidate palette $S\in\calS$ guarantees that, for every vector $(\delta_x)_{x\in S}$ of nonnegative later-neighbor masses (i.e., with $\sum_{x\in S}\delta_x\leq 1-\alpha$), there exists a color $x\in S$ with $\alpha_x+\delta_x\leq\rho$ if and only if
    \begin{equation}\label{eq:goodset}
        \sum_{x\in S}\max\set{0,\rho-\alpha_x}\geq 1-\alpha.
    \end{equation}
\end{lemma}
\begin{proof}
    By assumption, the later-neighbor masses satisfy $\sum_{x\in S}\delta_x\leq 1-\alpha$.

    Let $S':=\set{x\in S:\alpha_x<\rho}$ and write $R:=\sum_{x\in S'}(\rho-\alpha_x)$. If \Cref{eq:goodset} holds, then $R\geq 1-\alpha>0$, so $S'$ is nonempty. If every $x\in S'$ satisfied $\alpha_x+\delta_x>\rho$, we would have $\sum_{x\in S'}\delta_x>R\geq 1-\alpha$, a contradiction.

    Conversely, suppose that $R<1-\alpha$ and set $\eps:=(1-\alpha-R)/|S|>0$. The choice $\delta_x:=\max\set{0,\rho-\alpha_x}+\eps$ for every $x\in S$ has total mass $1-\alpha$ and satisfies $\alpha_x+\delta_x>\rho$ for every $x\in S$. Hence, $S$ cannot guarantee the desired bound.
\end{proof}

If $\alpha=1$, there are no later neighbors, and a set $S$ guarantees the desired bound precisely when $\alpha_x\leq\rho$ for some $x\in S$. The lower-bound argument below constructs a state with $\alpha<1$, so \Cref{lem:goodset} applies.

The next lemma proves the lower bound of \Cref{thm:2sweeplower} under the following simplifying assumption. We assume that in the first sweep, a node can see an arbitrary distribution of earlier-neighbor choices. Further, we assume that in the second sweep, all earlier neighbors will choose a worst-case color from their sets. This assumption is not true in general, but it allows us to identify a numerical obstruction for every candidate palette. A full proof of \Cref{thm:2sweeplower} is given in \Cref{sec:TechnicalDetailsLowerBound-Appendix}, where we show that this obstruction can also be realized in an execution. We identify actually realizable palettes and the colors that earlier neighbors can actually choose and we apply \Cref{lemma:LowerBound} to these effective-color sets.

\begin{lemma}[Numerical obstruction]
    \label{lemma:LowerBound}
    Let $\calS$ be any nonempty family of nonempty subsets of a color space $\mathcal C$ of size $C$. If
    \begin{equation}\label{eq:baddefectcondition}
    C\rho^2 < 1,
    \end{equation}
    then there are nonnegative masses $\alpha_S$ with $\alpha:=\sum_{S\in\calS}\alpha_S<1$ such that, writing $\alpha_x:=\sum_{S:x\in S}\alpha_S$,
    \[
        \sum_{x\in S}\max\set{0,\rho-\alpha_x}<1-\alpha
        \qquad\text{for every }S\in\calS.
    \]
\end{lemma}
\begin{proof}
    We use a constrained minimization argument and choose its objective by working backwards from these inequalities. Fix $\lambda:=(1+C\rho^2)/2$, so that $C\rho^2<\lambda<1$. Our aim is to obtain the stronger bounds
    \begin{align}\label{eq:potentialWithLambda}
        \sum_{x\in S}\max\set{0,\rho-\alpha_x}\leq\lambda(1-\alpha)
        \qquad\text{for every }S\in\calS
    \end{align}
    by showing that any violation would allow us to decrease a suitable potential. Define
    \begin{equation}\label{eq:potentialfct}
        \Phi(\vec\alpha)
        :=\sum_{x\in\mathcal{C}}\max\set{0,\rho-\alpha_x}^2
        -\lambda(1-\alpha)^2,
    \end{equation}
    where $\vec\alpha$ denotes $(\alpha_S)_{S\in\calS}$ with nonnegative entries fulfilling $\sum_{S\in\calS}\alpha_S\leq 1$. We will now argue that a vector minimizing this potential also fulfills \Cref{eq:potentialWithLambda}.

    By compactness of the feasible simplex and continuity of $\Phi$, there exists a non-negative vector that minimizes $\Phi$. Fix a minimizing vector $\vec\alpha$. At the zero vector, we have $\Phi(\vec 0)=C\rho^2-\lambda<0$, whereas every vector with $\alpha=1$ has $\Phi(\vec\alpha)\geq 0$. Thus, the minimizer satisfies $\alpha<1$.

    Suppose that \Cref{eq:potentialWithLambda} fails for some $S\in\calS$. Write $r_x:=\max\set{0,\rho-\alpha_x}$ and $R_S:=\sum_{x\in S}r_x$, and let $\eta:=R_S-\lambda(1-\alpha)>0$ be the amount of the violation. Increase only $\alpha_S$ by $\eps$, where $0<\eps<1-\alpha$, and denote the resulting feasible vector by $\vec\alpha(\eps)$. Its total mass is $\alpha+\eps$, and the remaining capacity of each $x\in S$ becomes $\max\set{0,r_x-\eps}$; all other capacities remain unchanged. Substituting directly into \Cref{eq:potentialfct} gives
    \begin{align*}
        \Phi(\vec\alpha(\eps))-\Phi(\vec\alpha)
        &=\sum_{x\in S}\left(\max\set{0,r_x-\eps}^2-r_x^2\right)
        +2\lambda\eps(1-\alpha)-\lambda\eps^2\\
        &\leq-2\eps R_S+|S|\eps^2+2\lambda\eps(1-\alpha)-\lambda\eps^2\\
        &=-2\eta\eps+(|S|-\lambda)\eps^2.
    \end{align*}
    The inequality uses $\max\set{0,r_x-\eps}^2\leq(r_x-\eps)^2$. Since $\eta>0$ and $|S|-\lambda>0$, choosing
    \[
        0<\eps<\min\set{1-\alpha,\frac{2\eta}{|S|-\lambda}}
    \]
    makes this potential difference strictly negative, contradicting the minimality of $\vec\alpha$. Hence, \Cref{eq:potentialWithLambda} holds for every $S\in\calS$. Using $\lambda<1$ and $\alpha<1$, we conclude that
    \[
        \sum_{x\in S}\max\set{0,\rho-\alpha_x}
        \leq\lambda(1-\alpha)<1-\alpha.
    \]
    This holds for every $S\in\calS$, so \Cref{lem:goodset} implies that no candidate palette can guarantee relative defect $\rho$ for these masses.
\end{proof}

\subsection{Towards the upper bound: formulation as a linear program}
\label{sec:towardsUpperBound}
For a fixed family $\mathcal S$, consider the following relaxed max-min linear
program:
\begin{equation}
\label{eq:generalQualityLP}
\begin{aligned}
    \text{maximize } z\\
    \text{subject to}\\
     \forall S \in \mathcal{S} \quad &Q_{\vec{\alpha}}(S) \geq z \\
     &\sum_{S \in \mathcal{S}} \alpha_S \leq 1 \\
      \forall S \in \mathcal{S}: \quad &\alpha_S \geq 0
\end{aligned}   
\end{equation}

Its optimum is
$\max_{\vec\alpha}\min_{S\in\mathcal S}Q_{\vec\alpha}(S)$ over the
displayed simplex. Every vector $\vec\alpha$ arising from the relative
predecessor-neighbor counts at a node is
feasible for this LP. Hence, by \Cref{eq:NearEqualityOfDefectAndQuality},
the LP optimum upper-bounds the worst-case relative defect of the actual
algorithm. The converse need not hold because the relaxation may contain
vectors $\vec\alpha$ that are not realizable by an execution.

\subsection{The Two-Bucket Two-Sweep Algorithm}
For the two-bucket construction, we partition the colors in $\mathcal{C}$ into two disjoint \textit{buckets} of sizes $C_1$ and $C_2$ (with $C=C_1+C_2$) and independently construct all possible sets of size $s_1:=\lfloor\sqrt C\rfloor$ in the first bucket and $s_2:=\lceil\sqrt C\rceil$ in the second bucket. More precisely:
\begin{align*}
    \mathcal{S} := \underbrace{\binom{\{1, ..., C_1\}}{\lfloor \sqrt{C} \rfloor}}_{\text{ first bucket }} \bigcup \underbrace{\binom{\{C_1 + 1, ..., C\}}{\lceil \sqrt{C} \rceil}}_{\text{ second bucket }} 
\end{align*}

Within bucket $i$, select the $s_i$ colors with the smallest counts $k_v(x)$, compute the resulting palette's quality, and choose the better of the two palettes.

Note that this construction has the handy property that sets from different buckets are guaranteed to be disjoint. Furthermore, the number of candidate palettes in the first bucket is exactly $\binom{C_1}{\lfloor\sqrt{C}\rfloor}$, and $\binom{C_2}{\lceil\sqrt{C}\rceil}$ sets can be constructed from the second bucket. \par
For purposes of obtaining an upper bound, the relaxed two-player game
simplifies as follows: Player one chooses $C_1$ and $C_2$, which determines
all palettes, and the relaxed adversary maximizes the LP objective over the
$\alpha_S$ values. We therefore seek a partition of the $C$ colors that
minimizes this worst-case upper bound. For some small values of $C$, we
computed the optimal integer splits using \Cref{claim:optimalIntegerBuckets}
and list them in \Cref{sec:full-table}.

By \Cref{lemma:EqualAlphas}, the relaxed LP has an optimal solution in which
all $\alpha_S$ values within a bucket are equal.  Restrict to such a solution,
and write $\alpha_i$ for the common value in bucket $i$.  The total
masses of the two buckets are
\begin{align*}
    a_1 := \binom{C_1}{s_1} \cdot \alpha_1 \text{ and } a_2 := \binom{C_2}{s_2} \cdot\alpha_2,
\end{align*}
respectively.  For arbitrary palettes $S_1$ and $S_2$ from the two buckets,
symmetry and~\Cref{eq:qualityOfSet} give the following two qualities:
\begin{align}
    Q_{\vec{\alpha}}(S_1) &= \frac{1}{s_1} \cdot \left( 1 + \sum_{S' \in \mathcal{S}} \left(|S_1 \cap S'| - 1\right) \cdot \alpha_{S'} \right) \notag \\ 
    &= \frac{1}{s_1} \cdot \left( 1 - \binom{C_1}{s_1} \alpha_1 - \binom{C_2}{s_2} \alpha_2 + \sum_{x \in S_1} \binom{C_1 - 1}{s_1 - 1} \cdot \alpha_1 \right) \notag\\
    &= \frac{1 - a_1  - a_2}{s_1} +  \frac{s_1}{C_1} \cdot a_1 \label{eq:qualityS1} \\
    Q_{\vec{\alpha}}(S_2) &= \frac{1 - a_1  - a_2}{s_2} +  \frac{s_2}{C_2} \cdot a_2 \label{eq:qualityS2}
\end{align}

The relaxed max-min linear program for the two-bucket construction thus
simplifies to the subsequent form.
\begin{equation}
\label{eq:LPbukets}
\begin{aligned}
    \text{maximize } z\\
    \text{subject to}\\
     \frac{1 - a_1  - a_2}{s_1} +  \frac{s_1}{C_1} \cdot a_1 &\geq z \\
     \frac{1 - a_1  - a_2}{s_2} +  \frac{s_2}{C_2} \cdot a_2  &\geq z \\
     a_1 + a_2 &\leq 1 \\
     a_1 \geq 0 \text{ and } a_2 &\geq 0
\end{aligned}   
\end{equation}

\begin{lemma}[Relaxed LP upper bound]
    \label{lemma:optimalityOfLP}
    Let $z_{2B}$ be the optimum of the linear program in
    \Cref{eq:LPbukets}. The relative defect guaranteed by the Two-Bucket
    Two-Sweep Algorithm is at most $z_{2B}$.
\end{lemma}
\begin{proof}
    By \Cref{eq:NearEqualityOfDefectAndQuality}, the
    optimum of the general relaxed LP upper-bounds the defect of every
    execution.  By \Cref{lemma:EqualAlphas}, that LP has a bucket-symmetric
    optimum.  Under the substitution
    $a_i=\binom{C_i}{s_i}\alpha_i$,
    \Cref{eq:qualityS1,eq:qualityS2} show that its restriction to
    bucket-symmetric vectors $\vec\alpha$ is the LP in \Cref{eq:LPbukets}. Its value $z_{2B}$ is
    therefore the required upper bound.
\end{proof}

\subsection{Solving the linear program}
Analyzing the general quality LP in \Cref{eq:generalQualityLP} for a single bucket with all subsets of one fixed size $p\in\set{\lfloor\sqrt C\rfloor,\lceil\sqrt C\rceil}$ is straightforward, and we defer the analysis to \Cref{sec:singleBucketAnalysis}. When $C$ is a perfect square, this single-bucket design attains the $1/\sqrt C$ lower bound for the Two-Sweep model of \Cref{thm:2sweeplower}. When $C$ differs from a perfect square by exactly one, one of the two stated single-bucket designs attains the corresponding upper bound in \Cref{thm:2bucketasymptotic}.

For completeness, consider two feasible buckets when $C=s^2$, so
$s_1=s_2=s$. By \Cref{eq:qualityS1,eq:qualityS2},
\[
  \min\{Q_{\vec\alpha}(S_1),Q_{\vec\alpha}(S_2)\}
  \leq \frac{C_1}{C}Q_{\vec\alpha}(S_1)
      +\frac{C_2}{C}Q_{\vec\alpha}(S_2)
  =\frac{1-a_1-a_2}{s}
      +\frac{s}{C}(a_1+a_2)
  =\frac1s.
\]
Taking $a_1=a_2=0$ attains $1/s$, so the LP optimum is
exactly $1/s$. The first term in the maximum in \Cref{thm:2bucketexact}
is $1/s$, and the second is $s^2/(C_2+s^3)\leq1/s$.
Together with \Cref{lemma:optimalityOfLP}, this proves the square case.

In the remainder of this section, we assume that $C$ is not a square number, and we analyze the linear program when both buckets are nonempty. The following lemma states a property of the linear program that we will later use to simplify our calculations.

\begin{lemma}
\label{lemma:optimalLPSolution}
    For nonsquare $C$ and any choice of $C_1 \geq s_1$ and $C_2 \geq s_2$, the LP has an optimal solution with either $a_1 = 0$ or $a_1 + a_2 = 1$ and
    \begin{align}
        z_{2B} = \frac{1 - a_1  - a_2}{s_1} +  \frac{s_1}{C_1} \cdot a_1 = \frac{1 - a_1  - a_2}{s_2} +  \frac{s_2}{C_2} \cdot a_2.
    \end{align}
\end{lemma}
\begin{proof}
    Denote the two expressions in \Cref{eq:LPbukets} by
    $f_i(a_1,a_2)$. The LP maximizes
    $\min\set{f_1,f_2}$ subject to $a_1,a_2\geq0$ and
    $a_1+a_2\leq1$.
    Since $C_2<C<s_2^2$, increasing $a_2$ decreases $f_1$ and
    increases $f_2$. Moreover, when $a_2=0$, we have
    \begin{align}\label{eq:optimalDifferenceWithAlpha2IsZero}
        f_1-f_2
        =(1-a_1)\left(\frac1{s_1}-\frac1{s_2}\right)
            +\frac{s_1}{C_1}a_1>0,
    \end{align}
        
    so equality is impossible there.

    We next show that an optimal solution can be chosen with $f_1=f_2$.
    If $f_1<f_2$, then $a_2>0$, and slightly decreasing $a_2$
    increases the smaller value $f_1$. If $f_2<f_1$ and
    $a_1+a_2<1$, then slightly increasing $a_2$
    increases the smaller value $f_2$. Finally, if $f_2<f_1$ and
    $a_1+a_2=1$, then $a_1>0$ (otherwise
    $a_2=1$ and $f_1=0<f_2=s_2/C_2$, contradicting $f_2<f_1$). Choose a sufficiently
    small $\delta>0$ and replace $a_1$ by $a_1-\delta$ and
    $a_2$ by $a_2+\delta$. This preserves feasibility, increases
    $f_2$, and, for sufficiently small $\delta$, leaves $f_2$ smaller than
    $f_1$. Hence it strictly increases $\min\set{f_1,f_2}$, contradicting
    optimality. Thus none of these strict inequalities can hold at an
    optimum, and we may assume that $f_1=f_2=z_{2B}$.

    The equality $f_1=f_2$ describes a line, and the feasible points on this
    line form an interval. The common value $f_1=f_2$ changes linearly along
    this interval, so its maximum is attained at one of the interval's
    endpoints. Such an endpoint satisfies $a_1=0$, $a_2=0$,
    or $a_1+a_2=1$. It remains to argue that $a_2 = 0$ cannot lead to an optimal solution. We know that \Cref{eq:optimalDifferenceWithAlpha2IsZero} is always larger than $0$ for any choice of $a_1$, but since for optimal solutions $f_1-f_2 = 0$, this is not possible.  Hence a maximum is attained with either $a_1=0$ or $a_1+a_2=1$, as claimed.
\end{proof}

Based on that, we compute the LP upper bound in \Cref{thm:2bucketexact} by
considering the two possible cases for $a_1$ and $a_2$. The
details are given in the subsequent lemma.

\begin{lemma}
 \label{lemma:ApproxFinal}
    For nonsquare $C$, for any choice of $ C_1 \geq s_1$ and $ C_2 \geq s_2$, the maximum value of $z_{2B}$ is
    \begin{align*}
        z_{2B} = \max \left\{\frac{s_1 \cdot s_2}{s_1 \cdot C_2 + s_2 \cdot C_1}, \frac{s_2^2}{C_2 + s_1 \cdot s_2^2} \right\}.
    \end{align*}
\end{lemma}
\begin{proof}
    Consider the two cases from \Cref{lemma:optimalLPSolution}.

    \textbf{Case $a_1 + a_2 = 1$:}

    In this case, we know that
    \begin{align*}
        z_{2B} = \frac{s_1}{C_1} a_1 = \frac{s_2}{C_2} a_2 = \frac{s_2}{C_2} (1 - a_1).
    \end{align*}
    Solving this equation for $a_1$ we get that $a_1 = C_1 \cdot \frac{s_2}{s_1C_2+s_2C_1}$. Substituting this value into $z_{2B}$ we get
    \begin{align*}
        z_{2B} = \frac{s_1 \cdot s_2}{s_1 \cdot C_2 + s_2 \cdot C_1}.
    \end{align*}
    \textbf{Case $a_1 = 0$:}
    \begin{align*}
        z_{2B} = \frac{1 - a_2}{s_1} = \frac{1}{s_2} + a_2 \cdot \left(\frac{s_2}{C_2} - \frac{1}{s_2} \right).
    \end{align*}
    This leads, using that $s_2 = s_1 +1$ if $C$ is not a square number, to $a_2 = \frac{C_2}{C_2 + s_1s_2^2}$ and thus
    \begin{align*}
        z_{2B} = \left( \frac{C_2 + s_1s_2^2}{C_2 + s_1s_2^2} - \frac{C_2}{C_2 + s_1s_2^2} \right)/s_1 = \frac{s_1s_2^2}{C_2s_1 + s_1^2s_2^2} = \frac{s_2^2}{C_2 + s_1s_2^2}.
    \end{align*}
\end{proof}

Together with \Cref{lemma:optimalityOfLP}, this proves the upper bound in
\Cref{thm:2bucketexact}.

To optimize this bound for a fixed nonsquare $C$, write $C_2=C-C_1$ and
regard the LP optimum as a function of the first bucket size:
\begin{align}\label{eq:OptimalValueAsFunctionOfC1}
    z_{2B}(C_1)
    =\max\left\{
        \frac{s_1s_2}{s_1C+C_1},
        \frac{s_2^2}{C-C_1+s_1s_2^2}
    \right\}.
\end{align}
The first argument decreases with $C_1$, while the second increases. We
first identify their intersection and then use this monotonicity to find
an optimal integer choice.

\begin{lemma}
    \label{lemma:RealValuesC1AndC2}
    Assume that $C$ is not a square number. The values
    \begin{align*}
        C_1^*:=\frac{s_1^2(s_2^2-C)}{s_1+s_2}
        \quad\text{and}\quad
        C_2^*:=\frac{s_2^2(C-s_1^2)}{s_1+s_2}
    \end{align*}
    sum to $C$ and make the two arguments defining $z_{2B}(C_1^*)$ equal.
    Their common value is
    \begin{align*}
        z_{2B}^*=\frac{s_1+s_2}{C+s_1s_2}
            =\frac{2s_1+1}{C+s_1^2+s_1}.
    \end{align*}
    For $C\geq3$, let
    \begin{align*}
        \widehat C_1:=\min\left\{C-s_2,\max\{s_1,C_1^*\}\right\}.
    \end{align*}
    Then the feasible interval $[s_1,C-s_2]$ is nonempty, and
    $z_{2B}(\widehat C_1)$ is the minimum value of $z_{2B}(C_1)$ over this
    interval.
\end{lemma}
\begin{proof}
    Since $s_2-s_1=1$, we have
    \begin{align*}
        C_1^*+C_2^*
            =\frac{(s_2^2-s_1^2)C}{s_1+s_2}=C, \quad
        s_1C_2^*+s_2C_1^*
            =\frac{s_1s_2(C+s_1s_2)}{s_1+s_2},\quad
        C_2^*+s_1s_2^2
            =\frac{s_2^2(C+s_1s_2)}{s_1+s_2}.
    \end{align*}
    Since $C_2^*=C-C_1^*$ and $s_2=s_1+1$, the denominator of the first
    argument satisfies $s_1C+C_1^*=s_1C_2^*+s_2C_1^*$. Substituting the
    identities above into the two arguments therefore gives
    \begin{align*}
        \frac{s_1s_2}{s_1C+C_1^*}
            &=\frac{s_1s_2}{s_1C_2^*+s_2C_1^*}
            =s_1s_2\cdot\frac{s_1+s_2}{s_1s_2(C+s_1s_2)}
            =\frac{s_1+s_2}{C+s_1s_2},\\
        \frac{s_2^2}{C-C_1^*+s_1s_2^2}
            &=\frac{s_2^2}{C_2^*+s_1s_2^2}
            =s_2^2\cdot\frac{s_1+s_2}{s_2^2(C+s_1s_2)}
            =\frac{s_1+s_2}{C+s_1s_2}.
    \end{align*}
    Thus both arguments are equal at $C_1^*$, and their common value is
    \begin{align*}
        z_{2B}^*=\frac{s_1+s_2}{C+s_1s_2}
            =\frac{2s_1+1}{C+s_1^2+s_1},
    \end{align*}
    where the last equality again uses $s_2=s_1+1$.
    For nonsquare $C\geq3$, we have
    $C\geq s_1+s_2$, so the feasible interval is nonempty. The two
    arguments defining $z_{2B}(C_1)$ are positive, with the first strictly
    decreasing and the second strictly increasing. Their maximum therefore
    decreases up to $C_1^*$ and increases after $C_1^*$. Its constrained
    minimum is attained at the intersection if it is feasible, and at the
    nearest endpoint otherwise, which is exactly $\widehat C_1$.
\end{proof}

\begin{claim}[Optimal integer bucket choice]
    \label{claim:optimalIntegerBuckets}
    Let $C\geq3$ be nonsquare, let $\widehat C_1$ be as in
    \Cref{lemma:RealValuesC1AndC2} and let $z_{2B}(C_1)$ be the function defined in \Cref{eq:OptimalValueAsFunctionOfC1}. An optimal feasible integer split is
    obtained by comparing
    \begin{align*}
        C_1\in\left\{\lfloor\widehat C_1\rfloor,
                         \lceil\widehat C_1\rceil\right\},
        \qquad C_2=C-C_1.
    \end{align*}
    Comparing these at most two splits with the single-bucket choices of
    palette sizes $s_1$ and $s_2$ gives the smallest relaxed-LP upper bound
    among all these constructions. In particular, we can guarantee
    \begin{align*}
        \rho\leq\min\left\{
            \frac1{s_1},\frac{s_2}{C},
            z_{2B}(\lfloor\widehat C_1\rfloor),
            z_{2B}(\lceil\widehat C_1\rceil)
        \right\}.
    \end{align*}
\end{claim}
\begin{proof}
    Both endpoints of $[s_1,C-s_2]$ are integers, so rounding
    $\widehat C_1$ down or up preserves feasibility. By the monotonicity in
    \Cref{lemma:RealValuesC1AndC2}, every other feasible integer has an LP
    value at least as large as one of these two choices. The single-bucket
    LP values are $1/s_1$ and $s_2/C$ by
    \Cref{lemma:SingleBucketAnalysis}. Taking the smallest of these values
    proves the claim.
\end{proof}

For square $C$, use a single bucket with palettes of size $\sqrt C$; for
$C=2$, the best single-bucket bound is $1$, and two nonempty buckets are
infeasible. Together with \Cref{claim:optimalIntegerBuckets}, this gives
an exact rule using at most four candidate constructions for every $C$.
Ties may be resolved arbitrarily; preferring a single bucket in a tie gives
the choices displayed in \Cref{tab:short_defect_data,tab:defect_table_complete}.
The rule optimizes the LP upper bound within the analyzed family. We next
show that it also satisfies the asymptotic bounds of
\Cref{thm:2bucketasymptotic}.

\restatetwobucketasymptotic*
\begin{proof}
    For square $C$, \Cref{lemma:SingleBucketAnalysis} gives
    $\rho\leq1/\sqrt C$.

    If $C=s_2^2-1$ or $C=s_1^2+1$, the two single-bucket choices give,
    respectively,
    \begin{align*}
        \rho&\leq\frac{s_2}{C}
            =\frac{\sqrt{C+1}}{C}
            \leq\frac1{\sqrt C}+\frac1{2C^{3/2}},\\
        \rho&\leq\frac1{s_1}
            =\frac1{\sqrt{C-1}}
            =\frac1{\sqrt{C\left(1-\frac1C\right)}}
            =\frac1{\sqrt C}\left(1-\frac1C\right)^{-1/2}
            =\frac1{\sqrt C}+\frac1{2C^{3/2}}+O(C^{-5/2}).
    \end{align*}
    Here we used $\sqrt{1+x}\leq1+x/2$ for $x\geq0$ and the expansion
    $(1-x)^{-1/2}=1+x/2+O(x^2)$. Thus, in both cases,
    $\rho\leq1/\sqrt C+1/(2C^{3/2})+O(C^{-5/2})$, proving the second
    case of the theorem.

    For every remaining nonsquare $C$, we have
    $s_1^2+2\leq C\leq s_2^2-2$ and hence $C\geq6$. In this range,
    \begin{align*}
        C_1^*&\geq\frac{2s_1^2}{s_1+s_2}
            =s_1-\frac{s_1}{s_1+s_2}>s_1-1,\\
        C_2^*&\geq\frac{2s_2^2}{s_1+s_2}>s_2.
    \end{align*}
    Thus $C_1=\lceil C_1^*\rceil$ and $C_2=C-C_1$ are feasible integer
    bucket sizes. The exact rule in \Cref{claim:optimalIntegerBuckets}
    gives a bound no larger than that of this choice, so it suffices to
    bound $z_{2B}(C_1)$.

    Let $\varepsilon:=C_1-C_1^*\in[0,1)$, so
    $C_2=C_2^*-\varepsilon$. Since $C_1\geq C_1^*$, the second argument
    defining $z_{2B}(C_1)$ attains the maximum. Consequently,
    \begin{align}
        z_{2B}(C_1)-z_{2B}^*
            &=\frac{s_2^2}{C_2+s_1s_2^2}
                -\frac{s_2^2}{C_2^*+s_1s_2^2}\notag\\
            &=\frac{s_2^2\bigl[(C_2^*+s_1s_2^2)-(C_2+s_1s_2^2)\bigr]}
                {(C_2+s_1s_2^2)(C_2^*+s_1s_2^2)}\notag\\
            &=\frac{s_2^2(C_2^*-C_2)}
                {(C_2+s_1s_2^2)(C_2^*+s_1s_2^2)}\notag\\
            &=\frac{\varepsilon}{C_2+s_1s_2^2}
                \cdot\underbrace{\frac{s_2^2}{C_2^*+s_1s_2^2}}_{=\,z_{2B}^*}\notag\\
            &=\frac{\varepsilon z_{2B}^*}{C_2+s_1s_2^2}
            \leq\frac1{s_1^2s_2^2}
            \leq\frac4{C^2}.\label{eq:integerBucketRounding}
    \end{align}
    The first inequality uses $\varepsilon<1$ and
    $z_{2B}^*=(s_1+s_2)/(C+s_1s_2)\leq1/s_1$; the last uses
    $s_1\geq\sqrt C/2$ and $s_2\geq\sqrt C$. This rounding term contributes
    $O(C^{-2})$ in addition to the $O(C^{-5/2})$ remainder in the bound
    on $z_{2B}^*$. Combining this with
    \Cref{lemma:technicalLemmaZStar} proves the third case:
    \begin{align*}
        \rho\leq z_{2B}(C_1) \leq z_{2B}^* + \frac4{C^2} \leq\frac1{\sqrt C}+\frac1{8C^{3/2}}
                +\frac1{32C^{5/2}}+\frac4{C^2} \leq\frac1{\sqrt C}+\frac1{8C^{3/2}}+\frac5{C^2}.
    \end{align*}
\end{proof}

\subsection{Improvement: Three Buckets}
\label{sec:threeBucketsImprovement}

So far, we have considered using one or two buckets. It turns out that
a third bucket can improve our bounds when $C$ is very close to a square.
For $C=s^2\pm1$ with $s\geq4$, using palette sizes $s-1,s,s+1$ reduces
the error above $1/\sqrt C$ from order $C^{-3/2}$ to order $C^{-2}$.
\Cref{thm:threeBucketsImprovement} summarizes this improvement together
with the bounds for the other cases from the previous section.

For example, at $C=15$ we can make the split
$(C_1,C_2,C_3)=(4,6,5)$ with palette sizes $(s_1,s_2,s_3)=(3,4,5)$.
By \Cref{eq:threeBucketCertificate}, this gives
\[
    \rho\leq\frac{40}{151}\approx0.26490
        <\frac4{15}\approx0.26667.
\]
This strictly improves the $C=15$ entry in \Cref{tab:defect_table_complete}
for the analyzed Single-Bucket and Two-Bucket constructions.

\paragraph*{The three-bucket construction.}
Let $C=s^2+\varepsilon$, where $s\geq4$ is an integer and
$\varepsilon\in\{-1,1\}$. Set $s_1:=s-1$, $s_2:=s$, and $s_3:=s+1$.
As in the two-bucket construction, partition the colors
$\mathcal C=\set{1,\dots,C}$ into consecutive buckets $B_1,B_2,B_3$ of
sizes $C_1,C_2,C_3$, with $C_i\geq s_i$ and $C_1+C_2+C_3=C$, and use
the common candidate family
\[
    \calS:=\binom{B_1}{s_1}\cup\binom{B_2}{s_2}\cup\binom{B_3}{s_3}.
\]
In Phase~I, node $v$ selects the $s_i$ colors with the smallest counts
$k_v(x)$ in each bucket $i$, computes the quality of the three resulting
palettes, and chooses the best one. Phase~II is unchanged. We call this
the \emph{Three-Bucket Two-Sweep Algorithm}.

\begin{lemma}[Three-bucket LP upper bound]
    \label{lem:threeBucketCertificate}
    Let $C=s^2+\varepsilon$, where $s\geq4$ is an integer and
    $\varepsilon\in\{-1,1\}$. For every integer split $C_1+C_2+C_3=C$
    with $C_1\geq s-1$, $C_2\geq s$, and $C_3\geq s+1$, the optimum
    $z_{3B}(C_1,C_3)$ of the relaxed quality LP in
    \Cref{eq:generalQualityLP} for $\calS$ satisfies\footnote{In fact,
    one can even show that equality holds for \Cref{eq:threeBucketCertificate}.}
    \begin{align}
        z_{3B}(C_1,C_3)\leq\max\left\{
            \frac{1}{C_1/(s-1)+C_2/s+C_3/(s+1)},
            \frac{1}{s-1+C_2/s^2+2C_3/(s+1)^2}
        \right\}.
        \label{eq:threeBucketCertificate}
    \end{align}
    The Three-Bucket Two-Sweep Algorithm has relative defect at most
    $z_{3B}(C_1,C_3)$.
\end{lemma}
\begin{proof}
     Let $\mathcal S_i:=\binom{B_i}{s_i}$. By
    \Cref{lemma:EqualAlphas}, applied with $m=3$, the LP in
    \Cref{eq:generalQualityLP} has an optimal solution $\vec\alpha$
    whose masses $\alpha_S$ are equal within each $\mathcal S_i$.
    Fix such a solution, let
    $a_i:=\sum_{S\in\mathcal S_i}\alpha_S$, and put
    $a:=a_1+a_2+a_3\leq1$. Since $a$ and $s_i$ do not depend on $S$,
    averaging the term $(1-a)/s_i$ sums $|\mathcal S_i|$ identical
    copies and divides by $|\mathcal S_i|$, so these factors cancel.
    Each color in $B_i$ belongs to exactly
    $\binom{C_i-1}{s_i-1}$ palettes in $\mathcal S_i$. Thus, for every
    $A\in\mathcal S_i$, summing $|S\cap A|$ over all $S\in\mathcal S_i$
    counts each of the $s_i$ colors in $A$ that many times. Substituting
    \Cref{eq:qualityOfSet} and exchanging the sums therefore gives
    \begin{align*}
        \overline Q_i
        &:=\frac1{|\mathcal S_i|}\sum_{S\in\mathcal S_i}Q_{\vec\alpha}(S)\\
        &=\frac{|\mathcal S_i|}{|\mathcal S_i|}\frac{1-a}{s_i}
            +\frac1{s_i|\mathcal S_i|}
                \sum_{A\in\mathcal S_i}\alpha_A
                \sum_{S\in\mathcal S_i}|S\cap A|\\
        &=\frac{1-a}{s_i}
            +\frac{s_i\binom{C_i-1}{s_i-1}}
                  {s_i\binom{C_i}{s_i}}
                \sum_{A\in\mathcal S_i}\alpha_A\\
        &=\frac{1-a}{s_i}+\frac{s_i}{C_i}a_i.
    \end{align*}
    Here, palettes from other buckets have empty intersection with $S$,
    and the last equality uses
    $\binom{C_i-1}{s_i-1}/\binom{C_i}{s_i}=s_i/C_i$.
    The uniform masses within each bucket make all its palette qualities
    equal, so $Q_{\vec\alpha}(S)=\overline Q_i$ for every
    $S\in\mathcal S_i$.

    Let $D_1,D_2$ denote the two denominators in
    \Cref{eq:threeBucketCertificate}, and set
    \[
        x_2:=\frac{C_2}{s^2},\qquad
        x_3:=\frac{C_3}{(s+1)^2},\qquad
        x_1:=\min\set{\frac{C_1}{(s-1)^2},1-x_2-x_3}.
    \]
    Since $x_2+x_3\leq(C-C_1)/s^2<1$, these weights satisfy
    $x_i\geq0$, $\sum_i x_i\leq1$, and $s_i^2x_i/C_i\leq1$.
    Consequently,
    \begin{equation}
        \sum_{i=1}^3 s_i x_i\overline Q_i
        =(1-a)\sum_{i=1}^3x_i
            +\sum_{i=1}^3\frac{s_i^2x_i}{C_i}a_i
        \leq1.
        \label{eq:threeBucketWeightedQualityBound}
    \end{equation}
    Since $s-1>0$, we can multiply the minimum defining $x_1$ by $s-1$
    and add the remaining terms to obtain
    \begin{align*}
        \sum_{i=1}^3s_i x_i
        &=(s-1)x_1+s x_2+(s+1)x_3\\
        &=\min\left\{
            \frac{C_1}{s-1}+s x_2+(s+1)x_3,\,
            (s-1)(1-x_2-x_3)+s x_2+(s+1)x_3
        \right\}\\
        &=\min\left\{
            \frac{C_1}{s-1}+s x_2+(s+1)x_3,\,
            s-1+x_2+2x_3
        \right\}.
    \end{align*}
    In the second expression, the coefficients simplify because
    $s-(s-1)=1$ and $(s+1)-(s-1)=2$. Substituting $x_2$ and $x_3$ now gives
    \begin{align*}
        \sum_{i=1}^3s_i x_i
        =\min\left\{
            \underbrace{\frac{C_1}{s-1}+\frac{C_2}{s}
                +\frac{C_3}{s+1}}_{D_1},\,
            \underbrace{s-1+\frac{C_2}{s^2}
                +\frac{2C_3}{(s+1)^2}}_{D_2}
        \right\}
        =\min\set{D_1,D_2}.
    \end{align*}
    Both denominators are positive: every term in $D_1$ is positive,
    and $D_2\geq s-1>0$. Hence $\sum_i s_i x_i>0$.
    Each bucket average $\overline Q_i$ is at least the minimum palette
    quality $\min_{S\in\calS}Q_{\vec\alpha}(S)$. Multiplying these
    inequalities by the nonnegative weights $s_i x_i$, summing over $i$,
    and dividing by their positive sum gives the first inequality below.
    The second uses \Cref{eq:threeBucketWeightedQualityBound} to bound
    the numerator by $1$:
    \begin{align}
        \min_{S\in\calS}Q_{\vec\alpha}(S)
            &\leq\frac{\sum_{i=1}^3s_i x_i\overline Q_i}
                       {\sum_{i=1}^3s_i x_i}
            \leq\frac1{\min\set{D_1,D_2}}
            =\max\set{1/D_1,1/D_2}.
        \label{eq:threeBucketScoreBound}
    \end{align}
    This proves the required upper bound on the LP optimum.

    The defect guarantee follows from
    \Cref{eq:NearEqualityOfDefectAndQuality}.
\end{proof}

We choose the bucket sizes by comparing two integer splits, just as in
the two-bucket construction:
\begin{itemize}
    \item If $C=s^2+1$, fix $C_1=s-1$ and consider
    \begin{align}
        C_3\in\set{\lfloor C_3^*\rfloor,\lceil C_3^*\rceil},
        \qquad
        C_3^*:=\frac{(s+1)^2(3s-2)}{(2s+1)(s-1)}.
        \label{eq:threeBucketPlusSplit}
    \end{align}
    \item If $C=s^2-1$, fix $C_3=s+1$ and consider
    \begin{align}
        C_1\in\set{\lfloor C_1^*\rfloor,\lceil C_1^*\rceil},
        \qquad
        C_1^*:=\frac{(s-1)^2(3s+2)}{(2s-1)(s+1)}.
        \label{eq:threeBucketMinusSplit}
    \end{align}
\end{itemize}
In both cases, set $C_2=C-C_1-C_3$. Both roundings are feasible for
$s\geq4$ by \Cref{lem:threeBucketRounding}. Choose the rounding that
minimizes the upper bound in \Cref{eq:threeBucketCertificate};
ties may be resolved arbitrarily.

\begin{lemma}[Three-bucket asymptotic bound]
    \label{lem:threeBucketAsymptotic}
    Let $C=s^2+\varepsilon$, where $s\geq4$ is an integer and
    $\varepsilon\in\{-1,1\}$. The bucket sizes selected above are feasible
    and satisfy
    \[
        z_{3B}(C_1,C_3)
        \leq\frac1{\sqrt C}+\frac3{2C^2}+O(C^{-5/2}).
    \]
\end{lemma}
\begin{proof}
    Let $D_1,D_2$ denote the denominators in
    \Cref{eq:threeBucketCertificate}. By \Cref{lem:threeBucketRounding},
    both candidate roundings are feasible, and there is an absolute
    constant $K>0$ such that, for all sufficiently large $C$, the selected
    split satisfies
    \[
        \min\set{D_1,D_2}\geq\sqrt C-\frac3{2C}-K C^{-3/2}.
    \]
    The right-hand side is positive for all sufficiently large $C$.
    Applying \Cref{eq:threeBucketCertificate}, factoring out $\sqrt C$,
    and using the geometric series $(1-u)^{-1}=\sum_{k=0}^{\infty}u^k$
    for $|u|<1$, we obtain
    \begin{align*}
        z_{3B}(C_1,C_3)
        &\leq\frac1{\min\set{D_1,D_2}}
        \leq\left(\sqrt C-\frac3{2C}-K C^{-3/2}\right)^{-1}\\
        &=\frac1{\sqrt C}
            \left(1-\frac3{2C^{3/2}}-\frac K{C^2}\right)^{-1}\\
        &=\frac1{\sqrt C}
            \left(1+\frac3{2C^{3/2}}+O(C^{-2})\right)\\
        &=\frac1{\sqrt C}+\frac3{2C^2}+O(C^{-5/2}).
    \end{align*}
    Here $u=3/(2C^{3/2})+K/C^2$.
    The same lemma gives $\min\set{D_1,D_2}>s$ for $\varepsilon=1$ and
    $\min\set{D_1,D_2}>C/s$ for $\varepsilon=-1$. Thus the certified
    bounds strictly improve the respective guarantees $1/s$ and $s/C$
    from \Cref{thm:2bucketasymptotic}.
\end{proof}

\begin{theorem}[Three-bucket improvement]
    \label{thm:threeBucketsImprovement}
    For every integer $C\geq3$, choosing between the Single-Bucket,
    Two-Bucket, and Three-Bucket Two-Sweep Algorithms gives a $C$-coloring
    with relative defect $\rho$ satisfying
    \[
        \rho\leq
        \begin{cases}
            1/\sqrt C,
                &\text{if $C=s^2$,}\\
            1/\sqrt C+3/(2C^2)+O(C^{-5/2}),
                &\text{if $C=s^2\pm1$ and $s\geq4$,}\\
            1/\sqrt C+1/(8C^{3/2})+O(C^{-2}),
                &\text{if $\lfloor\sqrt C\rfloor^2+2\leq C
                              \leq\lceil\sqrt C\rceil^2-2$,}\\
            r_C,
                &\text{if $C\in\set{3,5,8,10}$,}
        \end{cases}
    \]
    where $s$ is an integer and
    $(r_3,r_5,r_8,r_{10})=(2/3,1/2,3/8,1/3)$.
\end{theorem}
\begin{proof}
    The square and remaining nonsquare cases follow from
    \Cref{thm:2bucketasymptotic}. For $C=3,5,8,10$, use one bucket with
    palette size $2,2,3,3$, respectively, and apply
    \Cref{lemma:SingleBucketAnalysis}.

    For $C=s^2\pm1$ with $s\geq4$, use the Three-Bucket Two-Sweep
    Algorithm with the bucket sizes selected above. By
    \Cref{lem:threeBucketCertificate}, its relative defect is at most
    $z_{3B}(C_1,C_3)$, and \Cref{lem:threeBucketAsymptotic} gives the
    claimed bound.
\end{proof}

\bibliography{references}

\onecolumn
\appendix
\clearpage
\section{Coloring graphs of bounded neighborhood independence}
\label{sec:BoundedNeighborhood}

In this section, we apply our oriented list defective coloring result
\Cref{thm:OLDC2}, through the list-arbdefective reduction of
\Cref{coll:sweepAdaption}, to graphs of bounded neighborhood independence.
The neighborhood independence $\theta$ of a graph $G$ is the
maximum independence number of a graph induced by the neighborhood of a
node. Notable examples with $\theta=O(1)$ are line graphs of bounded-rank
hypergraphs and claw-free graphs. We follow the recursive framework of
\cite{FuchsK25,FuchsK25full}, but use the \CONGEST list-arbdefective base
case obtained from \Cref{thm:OLDC2}.

We first state the precise problem setup. Fix an undirected $n$-node
communication graph $G=(V,E)$ of maximum degree $\Delta$, neighborhood
independence at most $\theta$, and a supplied proper $q$-vertex coloring.
Every instance has nonempty lists $L_v\subseteq\mathcal C$ and integer defect
functions $d_v:L_v\to\mathbb N_0$. In a list defective coloring, each node
$v$ chooses $x_v\in L_v$ and has at most $d_v(x_v)$ neighbors of color
$x_v$. In a list arbdefective coloring, the algorithm additionally orients
every monochromatic edge and requires at most $d_v(x_v)$ same-colored
outneighbors at $v$; the orientation need not be acyclic.

An instance has slack at least a real value $S\geq1$ if, for every node $v$,
\begin{align*}
    \sum_{x\in L_v}(d_v(x)+1)>S\deg(v).
\end{align*}
For an integer $C\geq1$, let $P_D(S,C)$ and $P_A(S,C)$ denote these two
problem classes when $|\mathcal C|\leq C$. We write $T_D,T_A$ for the
upper bounds on deterministic round complexities for these instance classes
(with the ambient graph parameters understood) and $B_D,B_A$ for the maximum
number of bits sent over one edge in one round. 

We first use the slack-generation reduction, which is Lemma~5 in
\cite{FuchsK25} respectively Lemma~2.2 in the full version \cite{FuchsK25full}.
\begin{restatable}{lemma}{restateLemmaSlackGeneration}
    \label{lemma:SlackGeneration}
    Let $G=(V,E)$ be a graph equipped with a proper $q$-vertex coloring. For
    every real $S\geq1$, we have
    \begin{align}
        T_A(1,C)
        &\leq O(\log\Delta)\bigl(T_A(2,C)+\log^*\Delta\bigr)
            +O(\log^*q),\label{eq:slackGeneration1}\\
        T_A(2,C)
        &\leq O(S^2)T_A(S,C)+O(\log^*q).
        \label{eq:slackGeneration2}
    \end{align}
    The corresponding per-edge message bounds (in bits) are
    \begin{align*}
        B_A(1,C)
        &=O\bigl(\max\{\log q,\log\Delta,\log C,B_A(2,C)\}\bigr),\\
        B_A(2,C)
        &=O\bigl(\max\{\log q,\log C,\log S,B_A(S,C)\}\bigr).
    \end{align*}
\end{restatable}
\begin{proof}
    The runtime inequalities are exactly the cited slack-generation lemma, so we need to focus on the message sizes here.
    Its first reduction communicates the supplied $q$-coloring and an
    intermediate $O(\Delta^2)$-coloring. The second communicates the supplied
    coloring and one of $O(S^2)$ subgraph indices. These labels require
    $O(\log q+\log\Delta)$ and $O(\log q+\log(S))$ bits, respectively.
    The final arbdefective colors require $O(\log C)$ bits, in addition to
    the messages of the recursive calls. This gives the displayed bounds.
\end{proof}

We next use the color-space reduction of Lemma~6 in \cite{FuchsK25}.
\begin{restatable}{lemma}{restateLemmaColorSpaceReduction}
    \label{lemm:colorSpaceReduction}
    Let $C\geq1$ and $p\in\{1,\ldots,C\}$ be integers, and let
    $1\leq\sigma\leq S$. Then
    \begin{align}
        T_A(S,C)
        \leq T_D(\sigma,p)
            +T_A\left(\frac S\sigma,
                \left\lceil\frac Cp\right\rceil\right).
        \label{eq:colorSpaceReduction}
    \end{align}
    The corresponding message bound (in bits) is
    \begin{align*}
        B_A(S,C)=O\left(\max\left\{\log C,B_D(\sigma,p),
        B_A\left(\frac S\sigma,\left\lceil\frac Cp\right\rceil\right)
        \right\}\right).
    \end{align*}
\end{restatable}
\begin{proof}
    Except for the message bound, everything is shown in Lemma 6 of \cite{FuchsK25}. Here the only messages that need to be communicated are the color that refers to the color space of size $p \leq C$ and the final color. Thus, the number of bits is $O(\log C)$ plus the bits needed for the recursive calls.
\end{proof}

The last imported reduction is Theorem~7 in \cite{FuchsK25}, respectively
Theorem~2.4 in \cite{FuchsK25full}. This is the only one of the three
reductions that uses the neighborhood-independence bound.
\begin{restatable}{lemma}{restateforth}
    \label{lemma:ArbToDef}
    Let $\theta\geq1$, and let $G=(V,E)$ be an $n$-node graph of maximum
    degree $\Delta$ and neighborhood independence at most $\theta$, equipped
    with a proper $q$-vertex coloring. Then, for every real $S\geq1$,
    \begin{align}
        T_D(16\theta S,C)
        \leq O(\log^3\Delta)T_A(S,C)+O(\log^*q).
        \label{eq:defToArbdef}
    \end{align}
    The corresponding message bound (in bits) is
    \begin{align*}
        B_D(16\theta S,C)=O\bigl(\max\{\log C,\log q,\log\log\Delta,
        B_A(S,C)\}\bigr).
    \end{align*}
\end{restatable}
\begin{proof}
    The proof is essentially the same as Theorem 7 of \cite{FuchsK25}, however, we need to check the communication complexity. First, we need to split the graph into multiple subgraphs. This is done using the defective coloring algorithm from \cite{Kuhn2009} that requires messages of $O(\log \log \Delta)$ bits. An initial exchange of the supplied vertex coloring costs $O(\log q)$ bits. In later active phases, nodes exchange the $O(\log\log\Delta)$-bit subgraph labels and their final $O(\log C)$-bit colors, in addition to messages from recursive calls. This gives the stated bound.
\end{proof}

We now apply our Two-Sweep Algorithm as the list arbdefective base case.
\begin{corollary}
    \label{coll:sweepAdaption}
    There is a deterministic algorithm to solve list arbdefective coloring
    instances $P_A(1,C)$ on any $n$-node graph equipped with a proper $q$-vertex
    coloring in
    \begin{align*}
        T_A(1,C)=O\left(\sqrt C\log\Delta
        \bigl(\log^3C+\log C\log^*q\bigr)\right)
    \end{align*}
    rounds, using messages of
    $B_A(1,C)=O(\log C+\log q+\log n)$ bits. Thus this is a
    \CONGEST algorithm whenever $q,C\leq\poly(n)$.
\end{corollary}
\begin{proof}
    We use the general list-arbdefective reduction of Theorem~4 in
    \cite{FK23}. Its preprocessing obtains a proper $O(\Delta^2)$-coloring
    for the arbdefective decomposition. Starting from the supplied proper
    $q$-coloring, Linial's color reduction~\cite{Linial1987} takes
    $O(\log^*q)$ rounds; if the supplied coloring already has sufficiently
    few colors, we retain it. The original proper $q$-coloring remains
    available for the OLDC calls on all subgraphs. Thus the preprocessing
    term is $O(\log^*q)$ instead of $O(\log^*n)$.
    For $\nu=0$, the reduction transforms an oriented list defective coloring algorithm
    with slack factor $\kappa$, runtime $T$, and message size $B$ into an
    algorithm for $P_A(1,C)$ with runtime
    \begin{align*}
        O\bigl(\kappa\log\Delta\cdot T+\log^*q\bigr)
    \end{align*}
    The reduction applies to arbitrary integer defect functions satisfying the strict
    slack condition, not only to the zero-defect special case.

    We apply it with $\kappa=3\sqrt C$ and use \Cref{thm:OLDC2}, for which
    \begin{align*}
        T&=O\bigl(\log^3C+\log C\log^*q\bigr),
        &B&=O(\log q+\log C).
    \end{align*}
    Substitution gives the stated round bound, absorbing the additive
    $O(\log^*q)$ term. For communication, Theorem~4 of \cite{FK23}
    guarantees $O(B+\log n)$-bit messages. The preprocessing from the
    supplied coloring uses $O(\log q+\log\Delta)$-bit messages, which fit
    this bound because $\Delta<n$. Hence the resulting message size is
    $O(\log C+\log q+\log n)$, as claimed.
\end{proof}

We are now ready to prove the recursive bounds.
\begin{lemma}
  \label{lemm:mainNeighborhood}
  In $n$-node graphs of maximum degree $\Delta$ and neighborhood independence $\theta$, we have the following two bounds on $T_A(1,C)$ in the \CONGEST model, assuming that $C$ is at most $\poly \Delta$. 
  \begin{align}
    T_A(1,C) = &  (\theta\log\Delta\big)^{O\left(1+\frac{\log\log C}{\max\{1, \log\log\log \Delta - \log \log \theta\} } \right)}+ O(\log^* n),\label{eq:mainBound1}\\
    T_A(1,C) = & O\big(\theta^2 \cdot C^{1/4} \cdot \log^8 \Delta + \log^* n\big).\label{eq:mainBound2}
  \end{align}
\end{lemma}
\begin{proof}
    Let us first show how to solve $P_A(2, C)$; the first bound of \Cref{lemma:SlackGeneration} then completes the proof. We initially compute a proper $q=O(\Delta^2)$-coloring using \cite{Linial1987}. Starting from $P_A(2,C)$, we generate slack $S = 64 \cdot \theta$ using the second bound of \Cref{lemma:SlackGeneration}. We then apply the color-space reduction of \Cref{lemm:colorSpaceReduction} with parameters $\sigma = S/2=32\theta$ and $p=C^\varepsilon$ for an appropriately chosen $0<\varepsilon\leq1/2$. Here and below, we write the recursive sizes without integer rounding.\footnote{Note that one can run into rounding issues here. To avoid them, one may enlarge the ambient color space to the next power of two,
    increasing $C$ by less than a factor of two without changing any input list.
    Rounding the split exponents then gives integer recursive sizes and,
    after a constant enlargement of the base-case threshold, preserves
    the stated asymptotic bounds.}
    The resulting list defective instance has slack $32\theta=16\theta\cdot2$, so \Cref{lemma:ArbToDef} applies with its parameter $S=2$.
  \begin{align} 
      T_A(2, C) &\leq O(\theta^2) \cdot T_A \left(64 \cdot \theta, C \right) + O(\log^* q) \notag \\
         &\leq O(\theta^2) \cdot \left( T_D\big(32 \theta, C^\eps\big) + T_A\big(2, C^{1 - \eps}\big)\right) + O(\log^* q)  \notag \\
         &\leq O(\theta^2) \cdot \left( \log^3 \Delta \cdot T_A\big(2, C^{\eps}\big) + T_A\big(2, C^{1 - \eps}\big)\right)  \label{eq:splitting}
  \end{align}
  In the last step, the additive $O(\theta^2\log^*q)$ term from
  \Cref{lemma:ArbToDef} is absorbed by the first displayed term, using
  $q=O(\Delta^2)$, $\theta\geq1$, and $T_A(2,C^\eps)\geq1$.

  We now argue that all these steps have low message complexities. Since the
  initial coloring has $q=O(\Delta^2)$ colors, its contribution satisfies
  $\log q=O(\log\Delta)=O(\log n)$.
  \begin{align} 
      B_A(2, C) &\leq O(\max\{\log n, \log C, \log q, \log \theta, B_A(64 \theta, C) \}) \\
      &\leq O(\max\{\log n, \log C, \log q, \log \theta, B_D(32\theta, C^\eps), B_A(2, C^{1-\eps}) \}) \\
      &\leq O(\max\{\log n, \log C, \log q, \log \theta, \log \log \Delta, B_A(2, C^\eps), B_A(2, C^{1-\eps}) \}) \label{eq:spliitingMessages}
  \end{align}

  With $F(C):=T_A(2,C)$ and constants absorbed into $A$ and $B$,
  \Cref{eq:splitting} has the form
  \[
      F(C)\leq B F(C^\eps)+A F(C^{1-\eps}),
      \qquad A=O(\theta^2),\quad B=O(\theta^2\log^3\Delta).
  \]
  The recurrence analysis of Lemma~4.3 and Corollary~4.5 in
  \cite{FuchsK25full} chooses the split to balance these two contributions.
  We use $C \leq \max\{2,\log\Delta\}$
  as the base-case threshold.
  In order to obtain \Cref{eq:mainBound1}, we set $\eps$ appropriately to get the following bound,
  \begin{equation}\label{eq:mainBound1Slack}
    T_A(2,C) =  (\theta\log\Delta\big)^{O\left(1 + \frac{\log\log C}{\max\{1, \log\log\log \Delta - \log \log \theta\} } \right)}.
  \end{equation}
  For the base cases, $q=O(\Delta^2)$ and \Cref{coll:sweepAdaption} give
  \[
    T_A(2,C)\leq T_A(1,C)
    =O\!\left(\sqrt{\log\Delta}\,\log\Delta\cdot
      \bigl((\log\log\Delta)^3+\log\log\Delta\cdot\log^*q\bigr)\right)
    =O(\log^4\Delta),
  \]
  matching the base-case bound used in its underlying Lemma~4.3.
  We have to argue that the construction also works in \CONGEST. Since $q,C\leq\poly\Delta$,
  this base case works in \CONGEST. Before reaching the base case, the
  construction gradually applies the steps in \Cref{eq:splitting} to reduce
  the color space. By
  \Cref{lemma:SlackGeneration,lemm:colorSpaceReduction,lemma:ArbToDef}, the
  nonrecursive messages contain labels of
  $O(\max\{\log C,\log q,\log\theta,\log\log\Delta\})$ bits, and the
  base-case messages have size $O(\log C+\log q+\log n)$ by
  \Cref{coll:sweepAdaption}. The reductions execute their subcalls in
  separate phases, so the message size is bounded by the maximum over
  these messages, without multiplying by the recursion depth. Therefore,
  \Cref{eq:spliitingMessages} gives
  $O(\max\{\log n,\log C,\log q,\log\theta,\log\log\Delta\})$ bits
  throughout the recursion. Since $\theta\leq\Delta$ and
  $C,q\leq\poly\Delta$, this is $O(\log n)$. Computing the initial
  coloring also uses $O(\log n)$-bit messages. Hence, the complete
  procedure works in \CONGEST.

  We now show the second statement of the lemma.  We use
  \Cref{eq:splitting} with $\varepsilon=1/2$ and solve the remaining instance
  with \Cref{coll:sweepAdaption}. As above,
  $P_A(2,\sqrt C)\subseteq P_A(1,\sqrt C)$, so applying the corollary to color
  space $\sqrt C$ yields
  \begin{align*}
    T_A(2, C)
    &\leq O(\theta^2\log^3\Delta)\cdot T_A(2,\sqrt C)\\
    &\leq O(\theta^2\log^3\Delta)\cdot T_A(1,\sqrt C)\\
    &\leq O(\theta^2\log^3\Delta)\cdot
      O\left(C^{1/4}\log\Delta\cdot
      \big(\log^3(C)+\log(C)\cdot\log^*q\big)\right)\\
    &\leq O\left(\theta^2C^{1/4}\log^7\Delta\right),
  \end{align*}
  where we used $C=\poly\Delta$.  The first equation of
  \Cref{lemma:SlackGeneration} contributes one further $O(\log\Delta)$ factor.
  It remains to argue about the message complexity again, but since the same subprocedures are used as above, we can upper bound the message size again by $O(\log n)$ bits.
\end{proof}

To obtain \Cref{thm:boundedNeighborhoodMain}, set all defects to zero and
$C=\Delta+1$. The fraction in the exponent of \Cref{eq:mainBound1} is then
at least one, so its additive constant can be absorbed into the $O(\cdot)$.

\section{Technical details of the Two-Sweep lower bound}
\label{sec:TechnicalDetailsLowerBound-Appendix}

We complete \Cref{thm:2sweeplower} by realizing the numerical obstruction
from \Cref{sec:lowerBound} in finite undirected trees. Recall from
\Cref{sec:SweepAndBuckets} that $\mathcal C=\set{1,\dots,C}$ is the
color space and $\calS\subseteq 2^{\mathcal C}\setminus\set{\varnothing}$
is the fixed nonempty family of candidate palettes available to every
node. A palette is a nonempty set of colors: in Phase~I, node $v$
announces one palette $S_v\in\calS$, and in Phase~II it chooses its final
color $x_v\in S_v$. Membership in $\calS$ only makes a palette available;
we must still construct neighbors that actually announce it.

\subsection{Model and local counts}
\label{sec:lowerBoundModel}
We use the uniform-relative-defect setting of \Cref{sec:lowerBound} and
fix a global deterministic tie-breaking rule in each phase, based only on
local counts. Thus, the same local counts give the same decisions. All
edges below are undirected, so both endpoints count each other as
neighbors. The words \emph{earlier} and \emph{later} always refer to
Phase~I; Phase~II processes the vertices in reverse order. Assigning
distinct initial colors realizes any prescribed finite vertex order.

\paragraph*{Integer counts and their normalized counterparts.}
The graph constructions use numbers of neighbors rather than fractions.
For a node $v$, introduce the nonnegative integer counts
\begin{equation}
    \label{eq:lowerBound-local-counts}
    n_A:=\bigl|\set{u\in N_<(v):S_u=A}\bigr|\quad(A\in\calS),
    \qquad b:=|N_>(v)|,
    \qquad r_x:=r_v(x)\quad(x\in\mathcal C).
\end{equation}
Here $N_<(v)$ and $N_>(v)$ are the earlier and later neighbor sets from
\Cref{sec:genericTwoSweep}, and $r_v(x)$ counts later neighbors that
finish in color $x$. We abbreviate the two vectors of counts by
\[
    \vec{n}:=(n_A)_{A\in\calS},
    \qquad \vec{r}:=(r_x)_{x\in\mathcal C}.
\]
In particular,
\[
    b=\sum_{x\in\mathcal C}r_x,
    \qquad \deg(v)=\sum_{A\in\calS}n_A+b,
    \qquad k_v(x)=\sum_{\substack{A\in\calS\\x\in A}}n_A.
\]
At positive degree, these are exactly the unnormalized versions of the
quantities in \Cref{sec:lowerBound}:
\[
    \alpha_A=\frac{n_A}{\deg(v)},
    \qquad 1-\alpha=\frac{b}{\deg(v)},
    \qquad \delta_x=\frac{r_x}{\deg(v)}.
\]
The same notation will be used when the node under consideration is a
root $u$ instead of $v$.

By \Cref{eq:qualityOfSet}, choosing a palette at positive degree is
equivalent to minimizing
\[
    \frac{b+\sum_{A\in\calS}|S\cap A|n_A}{|S|}
    =\deg(v)\,Q_{\vec\alpha}(S)
    \qquad(S\in\calS).
\]
Formally, fix two deterministic selection functions $T_{\mathrm I}$
and $T_{\mathrm{II}}$, used by every vertex. For all nonnegative integer
counts, they satisfy
\[
    T_{\mathrm I}(\vec{n},b)\in
    \operatorname*{arg\,min}_{S\in\calS}
    \frac{b+\sum_{A\in\calS}|S\cap A|n_A}{|S|}
\]
and
\[
    T_{\mathrm{II}}(\vec{n},\vec{r})\in
    \operatorname*{arg\,min}_{x\in T_{\mathrm I}(\vec{n},\sum_{y\in\mathcal C}r_y)}
    \left(\sum_{\substack{A\in\calS\\x\in A}}n_A+r_x\right).
\]
These functions include the respective tie-breaking choices, depend
only on the displayed counts, and are fixed throughout the proof.
Thus $(\vec{n},b)$ determines the Phase~I choice, and
$(\vec{n},\vec{r})$ determines the Phase~II choice, since $\vec{r}$
also determines $b$ and hence the selected palette.
Preserving these integer counts preserves both
decisions, including ties. This exact preservation is the reason the
neighbor replacements below work.

\paragraph*{Attaching a rooted tree.}
To make one neighbor behave as required, we construct a finite tree with
a distinguished vertex $u$, called its \emph{root}, and a specified
order of its vertices. We connect it to an external vertex $w$ by the
single edge $\set{u,w}$. Only $u$ is adjacent to $w$; the other tree
vertices help determine the behavior of $u$. We call this tree,
together with its order and its connection to $w$, a \emph{witness}
for the behavior stated in each construction. Every vertex in it runs
the same algorithm with the same candidate family $\calS$.

The external vertex $w$ is not part of the witness. It may have arbitrary
other neighbors, provided that no additional edge enters the witness.
An \emph{earlier witness} places all its vertices before $w$; a
\emph{later witness} places all its vertices after $w$. These terms
refer to the order relative to $w$, not relative to the root $u$:
an earlier witness can contain neighbors that are later than $u$.
Every conditional guarantee below holds in any attachment respecting
this order. When several witnesses are attached to one vertex, their
vertex sets (including their roots) are disjoint, and their internal
orders are preserved. We also refer to these vertex sets as their
\emph{interiors}, to distinguish them from the shared external vertex.

\subsection{Realizable palettes}
\label{sec:realizablePalettes}
We need to know which palettes can be supplied by an earlier neighbor
of a vertex whose other neighbors have not yet been constructed.
Define $\mathcal R$ to be the family of palettes with this property.
More precisely, a palette $A\in\calS$ belongs to $\mathcal R$ if there
exists an earlier witness whose root $u$ announces $S_u=A$ in Phase~I
in every attachment at its external vertex $w$.
\Cref{fig:onePortPredecessorWitness} illustrates the definition.

Thus $\mathcal R\subseteq\calS$ is a family of \emph{palettes}, not a
set of individual colors. For each $A\in\mathcal R$ we may choose a
different witness. Membership certifies that we can attach a copy of
that tree and obtain one earlier neighbor announcing $A$. This says
nothing yet about its final color. The announcement is independent of
$w$'s choices because every witness vertex acts before $w$ in Phase~I;
the edge $\set{u,w}$ still contributes one later neighbor to the root's
degree and to $b$. We do not assume that $\mathcal R=\calS$.

\begin{figure}[htbp]
    \centering
    \begin{tikzpicture}[
    vertex/.style={circle,draw,minimum size=6mm,inner sep=0pt},
    root/.style={vertex,fill=black!10,thick},
    internal/.style={circle,draw,fill=white,minimum size=3mm,inner sep=0pt}
]
    \node[internal] (a1) at (0,0.7) {};
    \node[internal] (a2) at (0,-0.5) {};
    \node[internal] (a) at (1.3,0.1) {};
    \node[root] (u) at (3.1,0.1) {$u$};
    \node[internal] (b1) at (4.7,0.7) {};
    \node[internal] (b2) at (4.7,-0.5) {};
    \node[vertex] (w) at (7,0.1) {$w$};

    \draw (a1) -- (a) -- (a2);
    \draw (a) -- (u);
    \draw (b1) -- (u) -- (b2);
    \draw[thick] (u) -- (w)
        node[pos=0.8,above] {$\set{u,w}$};
    \draw[densely dotted] (w) -- (8.1,0.1);
    \node[anchor=west] at (8.1,0.1) {$\cdots$};

    \draw[dashed,rounded corners]
        (-0.4,-1.25) rectangle (5.4,1.15);
    \node[anchor=south] at (2.5,1.15) {earlier witness for $A$};
    \node[align=center] at (3.1,-0.95) {root $u$: $S_u=A$};
    \node[align=center] at (7,-0.65) {external vertex};

    \draw[-{Latex[length=2mm]}] (-0.2,-1.7) -- (7.2,-1.7)
        node[midway,below] {initial coloring order};
\end{tikzpicture}
    \caption{An earlier witness certifying $A\in\mathcal R$. The shaded
    vertex $u$ is the root, and the dashed boundary contains all witness
    vertices, including $u$. They all precede the external vertex $w$
    in Phase~I, although some may follow $u$. The edge $\set{u,w}$ is
    the witness's only connection to the surrounding graph; $w$ may
    have other neighbors outside it. The root announces $S_u=A$
    independently of $w$'s choices. The arrow indicates sweep order,
    not an edge orientation.}
    \label{fig:onePortPredecessorWitness}
\end{figure}
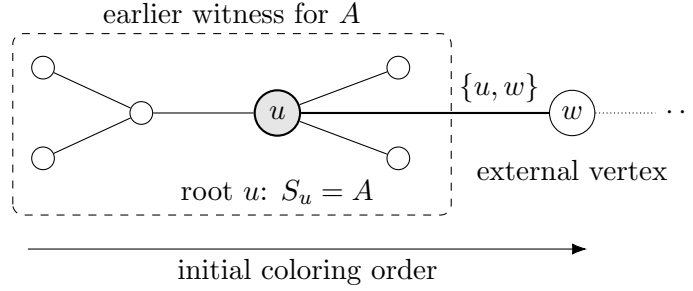

\textit{Example: realizability may exclude candidate palettes.} Consider two colors and the candidate family
    \[
        \mathcal C=\set{1,2},
        \qquad
        \calS=\bigl\{\set{1},\set{1,2}\bigr\}.
    \]
    Then
    \[
        \mathcal R=\bigl\{\set{1,2}\bigr\}\subsetneq\calS.
    \]
    Indeed, any root $u$ of an earlier witness has the external vertex
    $w$ as a later neighbor, and hence its normalized earlier-neighbor
    mass satisfies $\alpha_{\set{1}}+\alpha_{\set{1,2}}<1$. The two
    Phase~I qualities are
    \[
        Q_{\vec\alpha}(\set{1})=1,
        \qquad
        Q_{\vec\alpha}(\set{1,2})
        =\frac{1+\alpha_{\set{1,2}}}{2}<1.
    \]
    Thus every such root strictly chooses $\set{1,2}$, independently
    of the earlier witnesses and of tie-breaking. Consequently,
    $\set{1}\notin\mathcal R$. On the other hand, a root with no
    earlier neighbors and the single later neighbor $w$ announces
    $\set{1,2}$, so $\set{1,2}\in\mathcal R$.

    The condition that the witness has a later external vertex matters: a vertex with no later neighbors could choose $\set{1}$ when we force $\alpha_{\set{1,2}} = 1$.

\begin{lemma}[Realizable palettes]
    \label{lem:lowerBound-realizable-palettes}
    The family $\mathcal R$ is nonempty. For any nonnegative integer
    counts $n_A$, with $n_A=0$ for $A\notin\mathcal R$, and any integer
    $b\geq1$, the palette selected from $(\vec{n},b)$ also belongs
    to $\mathcal R$.
\end{lemma}
\begin{proof}
    A single root with its external vertex as its only, later neighbor
    witnesses some palette, so $\mathcal R\neq\varnothing$.
    For the second claim, take $n_A$ earlier witnesses for each
    $A\in\mathcal R$ and attach them to a new root $u$. Place all their
    vertices before $u$. Add $b-1$ leaves after $u$, and put the external
    vertex $w$ after all these vertices. The edge $\set{u,w}$ supplies
    the last later neighbor. The counts at $u$ are exactly $(\vec{n},b)$,
    independently of $w$'s choices. This is an earlier witness for the
    selected palette; see \Cref{fig:predecessorWitnessConstruction}.
\end{proof}

\begin{figure}[htbp]
    \centering
    \begin{tikzpicture}[
    vertex/.style={circle,draw,minimum size=6mm,inner sep=0pt},
    leaf/.style={vertex,minimum size=8mm},
    root/.style={vertex,fill=black!10,thick}
]
    \draw[dashed,rounded corners]
        (-5.6,-1.9) rectangle (4.2,2.05);
    \node[anchor=south] at (-0.7,2.05) {new earlier witness};

    \draw[rounded corners,fill=black!3]
        (-5.35,0.45) rectangle (-1.4,1.45);
    \draw[rounded corners,fill=black!3]
        (-5.35,-1.45) rectangle (-1.4,-0.45);
    \node[align=center] at (-4.15,0.95) {witness\\for $A$};
    \node[align=center] at (-4.15,-0.95) {witness\\for $A$};
    \node at (-3.8,0) {$\vdots$};
    \node[vertex,minimum size=8mm] (u1) at (-2.2,0.95) {$u_1$};
    \node[vertex,minimum size=8mm] (uL) at (-2.2,-0.95) {$u_{n_A}$};
    \node[align=center,font=\small] at (-3.45,1.75)
        {$n_A$ copies for each $A\in\mathcal R$};

    \node[root] (u) at (0,0) {$u$};
    \node at (0,-0.7) {new root $u$};

    \node[leaf] (l1) at (2.9,0.95) {$\ell_1$};
    \node at (2.9,0) {$\vdots$};
    \node[leaf] (lLast) at (2.9,-0.95) {$\ell_{b-1}$};
    \node[align=center,font=\small] at (2.8,1.75) {$b-1$ later leaves};

    \node[vertex] (w) at (6.2,0) {$w$};
    \node at (6.2,-0.65) {external vertex};

    \draw (u1) -- (u);
    \draw (uL) -- (u);
    \draw (u) -- (l1);
    \draw (u) -- (lLast);
    \draw[thick] (u) -- (w)
        node[pos=0.8,above] {$\set{u,w}$};

    \draw[-{Latex[length=2mm]}] (-5.3,-2.35) -- (6.2,-2.35)
        node[midway,below]
        {initial coloring order};
\end{tikzpicture}
    \caption{The witness construction of
    \Cref{lem:lowerBound-realizable-palettes}. Shaded boxes represent
    earlier witnesses for a palette $A$, with roots $u_1,\dots,u_{n_A}$.
    The dashed boundary excludes the external vertex $w$.}
    \label{fig:predecessorWitnessConstruction}
\end{figure}
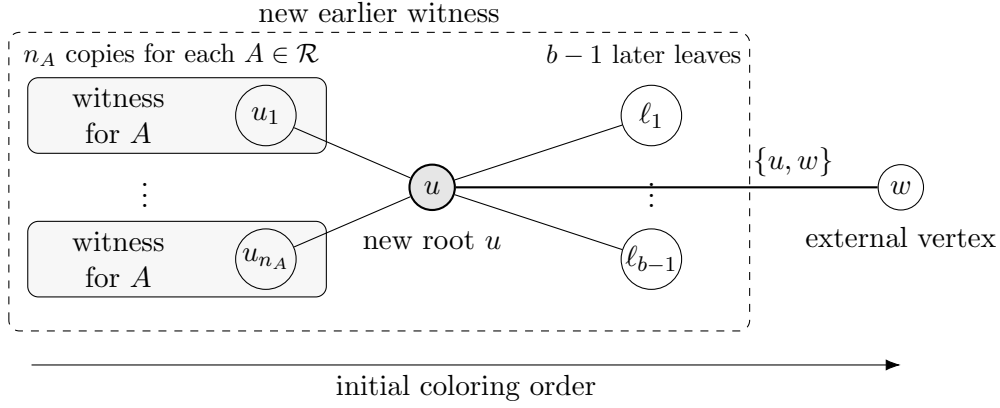

\subsection{Realizable later colors}
\label{sec:realizableLaterColors}
We next need later neighbors whose final colors are known before the
central node acts in Phase~II. A \emph{later witness} has the same form
as an earlier witness: a rooted tree joined to the external vertex $w$
by the single edge $\set{u,w}$. Its placement is \emph{mirrored relative
to $w$} in Phase~I: the entire earlier witness precedes $w$, while the
entire later witness follows $w$. The later root $u$ can therefore use
$w$'s announced palette, but chooses
its final color before $w$ in the reversed Phase~II order. An earlier
witness supplies $w$ with a palette for its Phase~I choice; a later
witness supplies it with a final color for its Phase~II choice.
This mirroring describes the order: merely reversing an earlier witness
need not give the required final color. The proof of
\Cref{lem:lowerBound-realizable-colors} constructs suitable later witnesses.

Define $\mathcal X\subseteq\mathcal C$
to be the set of colors $x$ with the
following property: for every $A\in\mathcal R$, there is a later witness
whose root outputs $x$ whenever its external vertex announces $A$.
This is a conditional guarantee: the witness does not force its external
vertex to announce $A$. When using the witness, we establish that
announcement separately from the external vertex's Phase~I counts.
The witness may depend on both $A$ and $x$. Its output is independent of
the external vertex's final color, which is chosen only afterward.
Here $x$ need not belong to $A$: the witness root chooses from its own
palette, whereas $A$ is the external vertex's palette.

\begin{lemma}[Realizable later colors]
    \label{lem:lowerBound-realizable-colors}
    The set $\mathcal X$ is nonempty. Moreover, for any nonnegative
    integer count vectors $(\vec{n},\vec{r})$ satisfying
    $n_A\geq1$ for $A\in\mathcal R$, $n_A=0$ for
    $A\in\calS\setminus\mathcal R$, and $r_x=0$ for
    $x\in\mathcal C\setminus\mathcal X$, the selected final color
    belongs to $\mathcal X$ as well.
\end{lemma}
\begin{proof}
    First attach one earlier witness for every $A\in\mathcal R$ to a
    root $u$, with no later neighbors. Let $x$ be the color chosen by
    $u$. For each $A\in\mathcal R$, take a fresh copy of this tree,
    remove the entire earlier witness of type $A$, and replace it by
    an external earlier neighbor $w$. Put $w$ before
    all remaining vertices; see
    \Cref{fig:laterColorWitnessConstruction}. If $w$ announces $A$,
    the counts at $u$ have not changed, so it still outputs $x$.
    This gives a later witness for $x$ for every external palette $A$, proving
    $x\in\mathcal X$. The root's palette need not belong to
    $\mathcal R$ in this step: $b=0$, so no later color witness needs
    to respond to that palette.

    For the closure claim, first set $b:=\sum_{x\in\mathcal C}r_x$
    and let $S$ be the palette selected by the Phase~I rule from the
    prescribed counts $(\vec{n},b)$. Thus $S$ is determined before we
    construct the tree. If $b>0$, then $S\in\mathcal R$ by
    \Cref{lem:lowerBound-realizable-palettes}.

    We now realize these counts at a new root $u$. Attach $n_A$
    earlier witnesses of each type $A\in\mathcal R$. If $b>0$, also
    attach $r_x$ later witnesses for each $x\in\mathcal X$, choosing
    the versions that output $x$ whenever their external vertex
    announces $S$. Such versions exist because $S\in\mathcal R$.
    Place all earlier witnesses before $u$ and all later witnesses
    after $u$, preserving their internal orders.

    Each earlier root announces its prescribed palette independently
    of $u$'s choices, so $u$ has exactly the earlier-palette counts
    $\vec{n}$. Each later witness contributes one later neighbor of
    $u$, giving exactly $\sum_x r_x=b$ later neighbors. The Phase~I
    rule uses only these counts: the internal structure, palettes,
    and final colors of the later witnesses do not enter this
    decision. Hence $u$ sees precisely the input $(\vec{n},b)$ used
    to define $S$ and, with the same deterministic tie-breaking rule,
    announces $S$. This establishes the condition required by the
    later witnesses, so they output exactly the prescribed colors
    before $u$ acts in Phase~II. If $b=0$, all $r_x$ are zero and no
    later witnesses are needed. In either case, $u$ has the specified
    Phase~II counts $(\vec{n},\vec{r})$. Denote its resulting output
    by $y$.

    For each $A\in\mathcal R$, start with a fresh copy of the resulting
    tree and replace one of the $n_A\geq1$ earlier witnesses by an
    external earlier vertex $w$, exactly as above.
    Conditional on $w$ announcing $A$, the root's degree and earlier
    palette counts are unchanged. Its palette, the outputs of its later witnesses, and
    its final color $y$ are consequently unchanged. These replacements
    give the later witnesses required for $y\in\mathcal X$.
\end{proof}

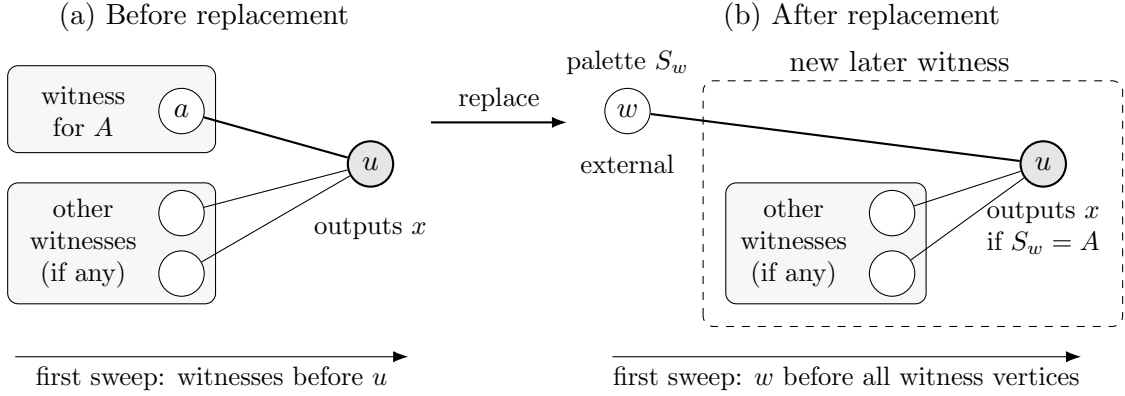
\begin{figure}[htbp]
    \centering
    \begin{tikzpicture}[
    vertex/.style={circle,draw,fill=white,minimum size=6mm,inner sep=0pt},
    root/.style={vertex,fill=black!10,thick},
    branch/.style={rounded corners,draw,fill=black!3},
    annotation/.style={align=center,font=\small}
]
    \node at (2,2.05) {(a) Before replacement};

    \draw[branch] (-0.6,0.25) rectangle (2.15,1.4);
    \node[annotation] at (0.4,0.8) {witness\\for $A$};
    \node[vertex] (a) at (1.7,0.8) {$a$};

    \draw[branch] (-0.6,-1.75) rectangle (2.15,-0.15);
    \node[annotation] at (0.4,-0.95) {other\\witnesses\\(if any)};
    \node[vertex] (p1) at (1.7,-0.55) {};
    \node[vertex] (p2) at (1.7,-1.35) {};

    \node[root] (uBefore) at (4.2,0.1) {$u$};
    \node[annotation] at (4.2,-0.75) {outputs $x$};
    \draw[thick] (a) -- (uBefore);
    \draw (p1) -- (uBefore);
    \draw (p2) -- (uBefore);

    \draw[-{Latex[length=2mm]}] (-0.5,-2.45) -- (4.7,-2.45)
        node[midway,below,annotation] {first sweep: witnesses before $u$};

    \draw[thick,-{Latex[length=2mm]}] (5,0.65) -- (6.8,0.65)
        node[midway,above,annotation] {replace};

    \begin{scope}[xshift=7.6cm]
        \node at (3.1,2.05) {(b) After replacement};
        \draw[dashed,rounded corners]
            (1,-2.05) rectangle (6.55,1.2);
        \node[anchor=south] at (3.6,1.2) {new later witness};

        \node[vertex] (w) at (0,0.8) {$w$};
        \node[annotation] at (0,1.45) {palette $S_w$};
        \node[annotation] at (0,0.1) {external};

        \draw[branch] (1.3,-1.75) rectangle (3.95,-0.15);
        \node[annotation] at (2.2,-0.95) {other\\witnesses\\(if any)};
        \node[vertex] (q1) at (3.5,-0.55) {};
        \node[vertex] (q2) at (3.5,-1.35) {};

        \node[root] (uAfter) at (5.5,0.1) {$u$};
        \node[annotation] at (5.5,-0.75) {outputs $x$\\if $S_w=A$};
        \draw[thick] (w) -- (uAfter);
        \draw (q1) -- (uAfter);
        \draw (q2) -- (uAfter);

        \draw[-{Latex[length=2mm]}] (-0.2,-2.45) -- (6,-2.45)
            node[midway,below,annotation] {first sweep: $w$ before all witness vertices};
    \end{scope}
\end{tikzpicture}
    \caption{Replacing one earlier witness by the external vertex $w$,
    shown for $b=0$. The construction does not force $w$'s palette:
    it guarantees that $u$ outputs $x$ if $S_w=A$, with no guarantee
    when $S_w\neq A$.
    The dashed boundary encloses the later witness; the closure step
    uses the same replacement with any later witnesses retained.}
    \label{fig:laterColorWitnessConstruction}
\end{figure}

\subsection{Effective colors and copying witnesses}
\label{sec:effectiveColorsAndCopying}
The same replacement idea works in the opposite direction. If a root
already chooses $x$ while one later neighbor supplies $x$, replacing
that neighbor by an external vertex of color $x$ preserves the root's
decision. We therefore identify local counts in which such a neighbor
is available.

For $S\in\mathcal R$, define its \emph{effective colors} $E(S)$ to be
the colors selected in Phase~II at some nonnegative integer count
vectors $(\vec{n},\vec{r})$ for which Phase~I, using
$b=\sum_{x\in\mathcal C}r_x$, selects $S$ and
\begin{equation}
    \label{eq:lowerBound-positive-counts}
    n_A\geq1\quad(A\in\mathcal R), \qquad n_A=0\quad(A\not\in\mathcal R),
    \qquad
    r_x\geq1\quad(x\in\mathcal X), \qquad r_x=0\quad(x\not\in\mathcal X).
\end{equation}
The set $E(S)$ may be a proper subset of $S$ or
even empty. By
\Cref{lem:lowerBound-realizable-colors}, $E(S)\subseteq S\cap\mathcal X$.
There exists at least one $S\in\mathcal R$ with $E(S)\neq\varnothing$.
To see this, set $n_A=1$ for every $A\in\mathcal R$ and $r_x=1$ for
every $x\in\mathcal X$, with all other counts zero. Let $S$ and $y$
be the palette and final color selected from these counts. Since
$\mathcal X\neq\varnothing$, we have $b=|\mathcal X|\geq1$, so
$S\in\mathcal R$ by \Cref{lem:lowerBound-realizable-palettes}.
The counts satisfy \Cref{eq:lowerBound-positive-counts}, and hence
$y\in E(S)$ by definition.

More generally, the construction in the proof of
\Cref{lem:lowerBound-realizable-colors} realizes any counts satisfying
\Cref{eq:lowerBound-positive-counts} at a root of a finite tree:
attach $n_A$ earlier palette witnesses of each type $A$ and $r_x$
later color witnesses of each type $x$, with the latter chosen for
the root's Phase~I palette. The two positivity conditions provide
the witnesses needed for the replacement arguments. The condition
$n_A\geq1$ ensures that an earlier witness announcing $A$ is available
to remove in the proof of \Cref{lem:lowerBound-realizable-colors}.
For the next lemma, if $x\in E(S)$, then $x\in\mathcal X$ and thus
$r_x\geq1$. Consequently, there is a later witness outputting $x$
whose entire tree can be removed and replaced by the external vertex.

\begin{lemma}[Earlier neighbors that copy a color]
    \label{lem:lowerBound-copying-witness}
    For every $S\in\mathcal R$ and $x\in E(S)$, there is an earlier
    witness whose root announces $S$ independently of the external
    vertex's choices and outputs $x$ whenever the external vertex
    outputs $x$.
\end{lemma}
\begin{proof}
    Choose counts for which the root selects $S$ and then $x$, as in
    the definition of $E(S)$, and realize them using earlier palette
    witnesses and later color witnesses as in the proof of
    \Cref{lem:lowerBound-realizable-colors}.
    Because $r_x\geq1$, a later neighbor is the root of a witness
    that outputs $x$. Remove that entire witness and replace it by the
    external vertex $w$, placed after all remaining vertices.

    The earlier palette counts and the number of later neighbors are
    unchanged, so the root still announces $S$, independently of $w$.
    The remaining later witnesses therefore still output their
    prescribed colors. If $w$ outputs $x$, the root's later-color
    counts also equal their original values, so the root still
    outputs $x$. All remaining vertices precede $w$, making this an
    earlier witness. They form a tree because the removed witness was
    a branch attached by a single edge. The guarantee is conditional
    on $w$ choosing $x$; no claim is needed when $w$ chooses another
    color.
\end{proof}

The three witness types used in the final construction are summarized
below. In every case, the root $u$ is the only witness vertex adjacent
to the external vertex $w$. The order column compares $u$ and $w$ in
Phase~I; all other witness vertices are on the same side of $w$ as $u$.
Phase~II reverses the order.
\begin{center}
\small
\begin{tabularx}{\linewidth}{@{}p{0.24\linewidth}p{0.14\linewidth}X@{}}
\toprule
\textbf{Witness type} & \textbf{Phase~I order} & \textbf{Root guarantee} \\
\midrule
Earlier palette, $A\in\mathcal R$
(\Cref{lem:lowerBound-realizable-palettes})
& $u$ before $w$
& Announces $S_u=A$ independently of $w$'s choices; no final color is prescribed. \\
\addlinespace
Later color, $A\in\mathcal R$, $x\in\mathcal X$
(\Cref{lem:lowerBound-realizable-colors})
& $u$ after $w$
& Outputs $x_u=x$ if $S_w=A$; independent of $w$'s final color. The witness is constructed for the given $A$ and $x$. \\
\addlinespace
Earlier copying, $S\in\mathcal R$, $x\in E(S)$
(\Cref{lem:lowerBound-copying-witness})
& $u$ before $w$
& Announces $S_u=S$ independently of $w$'s choices, and outputs $x_u=x$ if $x_w=x$. \\
\bottomrule
\end{tabularx}
\end{center}

\subsection{The obstruction and the tree construction}
\label{sec:obstructionAndTree}
We now connect the preceding constructions to \Cref{lemma:LowerBound}.
For masses $\alpha_A\geq0$, $A\in\mathcal R$, with all other masses zero,
we retain $\alpha$ and $\alpha_x$ and define $\gamma_x$ by
\begin{equation}
    \label{eq:lowerBound-copying-mass}
    \alpha=\sum_{A\in\mathcal R}\alpha_A,
    \qquad
    \alpha_x=\sum_{\substack{A\in\mathcal R\\x\in A}}\alpha_A,
    \qquad
    \gamma_x:=\sum_{\substack{A\in\mathcal R\\x\in E(A)}}\alpha_A.
\end{equation}
The quantity $\alpha_x$ counts all earlier palette incidences of $x$.
The quantity $\gamma_x$ counts only the earlier witnesses that the
preceding lemma can make copy $x$ if their external vertex chooses $x$.
Both are fractions of the degree of the vertex to which the witnesses
will be attached, and $\gamma_x\leq\alpha_x$ because $E(A)\subseteq A$.

\begin{lemma}[Obstruction for effective colors]
    \label{lem:lowerBound-effective-obstruction}
    If $0\leq\rho<1/\sqrt C$, there are positive rational masses
    $\alpha_A$, $A\in\mathcal R$, with $\alpha<1$ such that
    \begin{equation}
        \label{eq:lowerBound-effective-capacity}
        \sum_{x\in E(S)}\max\set{0,\rho-\gamma_x}<1-\alpha
        \qquad(S\in\mathcal R).
    \end{equation}
\end{lemma}
\begin{proof}
    Let $\mathcal F$ be the family of distinct nonempty sets among
    $\{E(S):S\in\mathcal R\}$. This family is nonempty by the
    preceding argument. Since $C\rho^2<1$, we can apply
    \Cref{lemma:LowerBound} to $\mathcal F$. It gives nonnegative
    masses $\mu_B$, $B\in\mathcal F$, with
    \[
        \mu:=\sum_{B\in\mathcal F}\mu_B<1,
        \qquad
        \mu_x:=\sum_{\substack{B\in\mathcal F\\x\in B}}\mu_B,
    \]
    such that
    \[
        \sum_{x\in B}\max\set{0,\rho-\mu_x}<1-\mu
        \qquad(B\in\mathcal F).
    \]

    To turn these masses on effective-color sets into masses on
    palettes, choose one representative $A_B\in\mathcal R$ with
    $E(A_B)=B$ for each $B\in\mathcal F$. Set
    $\alpha_{A_B}:=\mu_B$ and give every other palette mass zero.
    This preserves both the total mass and the mass contributing to
    each effective color:
    \[
        \alpha=\mu,
        \qquad
        \gamma_x
        =\sum_{\substack{A\in\mathcal R\\x\in E(A)}}\alpha_A
        =\sum_{\substack{B\in\mathcal F\\x\in B}}\mu_B
        =\mu_x.
    \]
    Now fix any $S\in\mathcal R$ with $E(S)\neq\varnothing$.
    Then $E(S)\in\mathcal F$, so applying the displayed inequality
    to $B=E(S)$ and substituting $\mu_x=\gamma_x$ and
    $\mu=\alpha$ gives \Cref{eq:lowerBound-effective-capacity}
    for this $S$. This also covers palettes that were not chosen as
    representatives: the inequality depends only on their effective
    set $E(S)$. If $E(S)=\varnothing$, the left-hand side is zero,
    and the inequality follows from $\alpha<1$.

    It remains to ensure that the masses are positive and rational,
    as required by the statement; the masses constructed so far may
    be zero or irrational. The finitely many inequalities in
    \Cref{eq:lowerBound-effective-capacity} and the bound $\alpha<1$
    are strict and depend continuously on the masses. They therefore
    remain valid after sufficiently small changes. We may thus
    replace the masses by sufficiently close positive rational
    numbers, obtaining all the claimed properties.
\end{proof}

\begin{proof}[Proof of \Cref{thm:2sweeplower}]
    Fix $0\leq\rho<1/\sqrt C$ and choose the masses from
    \Cref{lem:lowerBound-effective-obstruction}. For every
    possible palette $S\in\mathcal R$ at the central node, choose
    positive rational later-neighbor masses
    $(\delta_x^S)_{x\in\mathcal X}$ satisfying
    \begin{equation}
        \label{eq:lowerBound-later-masses}
        \sum_{x\in\mathcal X}\delta_x^S=1-\alpha,
        \qquad
        \delta_x^S>\max\set{0,\rho-\gamma_x}
        \quad(x\in E(S)).
    \end{equation}
    The superscript $S$ labels an allocation for that palette; each
    $\delta_x^S$ has the same meaning as the earlier $\delta_x$, namely
    the fraction of all neighbors that are later and finish in $x$.
    This is possible by \Cref{eq:lowerBound-effective-capacity}: the
    strict slack leaves room to give positive mass to every color in
    $\mathcal X$. Rational choices are possible because $1-\alpha$ is
    rational and the inequalities are strict. Since there are only
    finitely many pairs $(S,x)$ with $S\in\mathcal R$ and $x\in E(S)$,
    and at least one such pair, the constant
    \[
        \eps:=\min_{\substack{S\in\mathcal R\\x\in E(S)}}
            \bigl(\gamma_x+\delta_x^S-\rho\bigr)>0
    \]
    is well defined and satisfies
    \begin{equation}
        \label{eq:lowerBound-defect-gap}
        \gamma_x+\delta_x^S\geq\rho+\eps
        \qquad(S\in\mathcal R,\ x\in E(S)).
    \end{equation}

    Fix trees with the following guarantees, writing $u$ for each
    tree's root and $w$ for its external vertex:
    \begin{itemize}
        \item For every $A\in\mathcal R$, fix an \emph{earlier
        palette witness}: its root $u$ announces $A$ in Phase~I,
        independently of $w$'s choices. Such a tree exists by the
        definition of $\mathcal R$; no particular final color is
        required of its root.
        \item For every $A\in\mathcal R$ and $x\in E(A)$, fix an
        \emph{earlier copying witness} from
        \Cref{lem:lowerBound-copying-witness}: its root $u$ announces
        $A$ independently of $w$'s choices and outputs $x$ in
        Phase~II whenever $w$ outputs $x$.
        \item For every $A\in\mathcal R$ and $x\in\mathcal X$, fix a
        \emph{later color witness}: its root $u$ outputs $x$ in
        Phase~II whenever $w$ announces $A$ in Phase~I. Such a tree
        exists by the definition of $\mathcal X$, using the
        constructions of \Cref{lem:lowerBound-realizable-colors}.
        The tree may depend on both $A$ and $x$; its root's own
        palette is not prescribed.
    \end{itemize}
    Every vertex of an earlier witness precedes $w$ in Phase~I,
    whereas every vertex of a later witness follows $w$. Thus, in
    the reversed Phase~II order, later witnesses output their colors
    before $w$, while earlier copying witnesses act after $w$.

    Let $D$ be an arbitrarily large common multiple of the denominators
    of all the $\alpha_A$ and $\delta_x^S$. This will be the degree of
    the central node $v$. We first determine its decisions from the
    intended counts. Set $n_A:=D\alpha_A$ for all $A\in\mathcal R$ and
    $n_A=0$ otherwise, and let $b:=D(1-\alpha)$. Let $S$ be the palette
    selected from $(\vec{n},b)$, and hence
    $S\in\mathcal R$ by \Cref{lem:lowerBound-realizable-palettes}.
    Set $r_x:=D\delta_x^S$ for $x\in\mathcal X$ and $r_x=0$ otherwise,
    and let $x_v$ be the final color selected from
    $(\vec{n},\vec{r})$. These counts
    satisfy \Cref{eq:lowerBound-positive-counts}, hence $x_v\in E(S)$.

    We now realize these decisions at a central node $v$ of degree $D$.
    Attach the following witnesses, with disjoint interiors and $v$ as
    their external vertex (see also \Cref{fig:lowerBoundTreeConstruction}):
    \begin{itemize}
        \item For each $A\in\mathcal R$ with $x_v\notin E(A)$,
        attach $n_A$ copies of the fixed earlier palette witness
        announcing $A$.
        \item For each $A\in\mathcal R$ with $x_v\in E(A)$,
        attach $n_A$ copies of the fixed earlier copying witness
        announcing $A$ and outputting $x_v$ whenever $v$ outputs
        $x_v$.
        \item For each $x\in\mathcal X$, attach $r_x$ later
        witnesses outputting $x$ when their external vertex announces
        $S$.
    \end{itemize}
    Place every earlier witness before $v$ and every later witness
    after $v$, preserving their internal orders. Joining disjoint
    trees to $v$ by one edge each gives a finite simple tree, with
    \[
        \deg(v)=\sum_{A\in\mathcal R}n_A
            +\sum_{x\in\mathcal X}r_x=D.
    \]
    With $C$, $\rho$, $\calS$, and the selection functions
    $T_{\mathrm I},T_{\mathrm{II}}$ fixed, the witnesses chosen above
    form a finite collection of finite trees, fixed before $D$ varies.
    Their sizes are therefore bounded by a constant independent of $D$.
    Since the construction attaches exactly
    $\sum_A n_A+\sum_x r_x=D$ such witnesses, the resulting tree has
    $O(D)=O(\Delta)$ vertices.

    The initial proper coloring can also use a number of colors independent
    of $D$, but possibly depending on $C$, $\rho$, $\calS$, and the fixed
    selection functions. For each of the finitely many fixed witness types,
    fix ranks that realize its prescribed internal order, and let $q_0$ be the maximum
    number of ranks needed. Reuse colors $1,\dots,q_0$ in every earlier
    witness copy, give $v$ color $q_0+1$, and reuse colors
    $q_0+2,\dots,2q_0+1$ in every later witness copy. The copies have disjoint
    interiors and no edges between them, so this coloring is proper. It
    preserves every comparison along an edge and hence every local count and
    decision in both phases. Thus the initial coloring uses at most
    $2q_0+1$ colors, independently of $D$.

    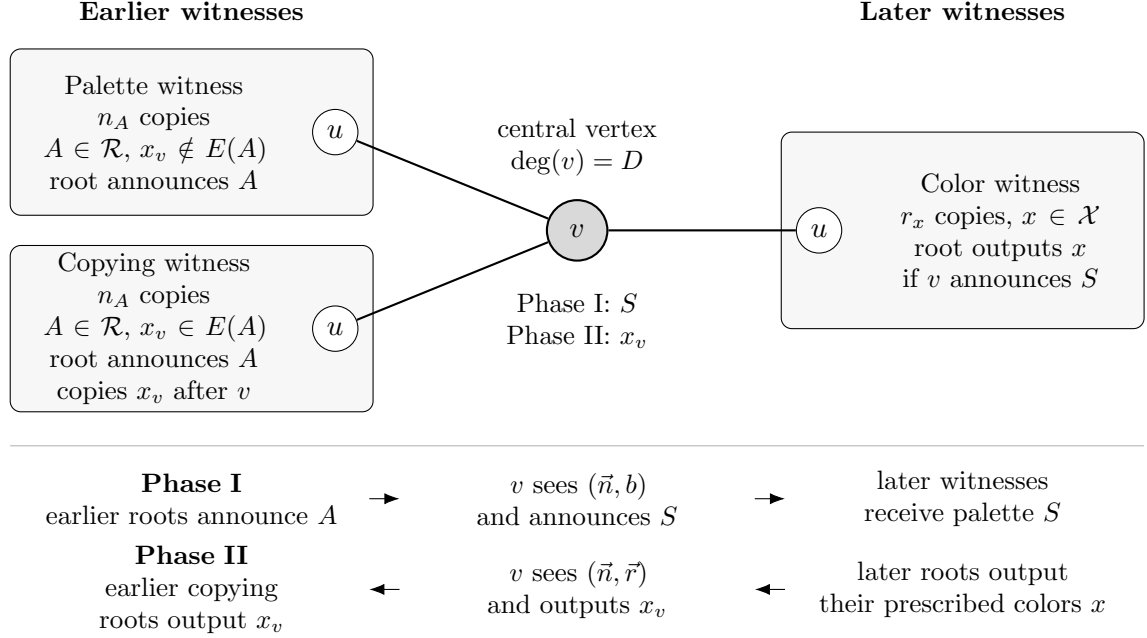
\begin{figure}[htbp]
        \centering
        \begin{tikzpicture}[
    vertex/.style={circle,draw,fill=white,minimum size=6mm,inner sep=0pt},
    central/.style={vertex,fill=black!15,thick,minimum size=8mm},
    witness/.style={rounded corners,draw,fill=black!3},
    annotation/.style={align=center,font=\small},
    sweepStep/.style={annotation,text width=4.5cm}
]
    \node[annotation,font=\small\bfseries] at (2.4,3.6)
        {Earlier witnesses};
    \node[annotation,font=\small\bfseries] at (12.6,3.6)
        {Later witnesses};

    \draw[witness] (0,0.9) rectangle (4.8,3.1);
    \node[annotation,text width=3.4cm] at (1.9,2)
        {Palette witness\\
        $n_A$ copies\\
        $A\in\mathcal R$, $x_v\notin E(A)$\\
        root announces $A$};
    \node[vertex] (paletteRoot) at (4.3,2) {$u$};

    \draw[witness] (0,-1.7) rectangle (4.8,0.5);
    \node[annotation,text width=3.4cm] at (1.9,-0.6)
        {Copying witness\\
        $n_A$ copies\\
        $A\in\mathcal R$, $x_v\in E(A)$\\
        root announces $A$\\
        copies $x_v$ after $v$};
    \node[vertex] (copyingRoot) at (4.3,-0.6) {$u$};

    \node[central] (v) at (7.5,0.7) {$v$};
    \node[annotation] at (7.5,1.8) {central vertex\\$\deg(v)=D$};
    \node[annotation] at (7.5,-0.5)
        {Phase I: $S$\\Phase II: $x_v$};

    \draw[witness] (10.2,-0.6) rectangle (15,2);
    \node[vertex] (colorRoot) at (10.7,0.7) {$u$};
    \node[annotation,text width=3.4cm] at (13.1,0.7)
        {Color witness\\
        $r_x$ copies, $x\in\mathcal X$\\
        root outputs $x$\\
        if $v$ announces $S$};

    \draw[thick] (paletteRoot) -- (v);
    \draw[thick] (copyingRoot) -- (v);
    \draw[thick] (v) -- (colorRoot);

    \draw[black!30] (0,-2.15) -- (15,-2.15);
    \node[sweepStep] at (2.4,-2.85)
        {\textbf{Phase I}\\earlier roots announce $A$};
    \node[sweepStep] at (7.5,-2.85)
        {$v$ sees $(\vec n,b)$\\and announces $S$};
    \node[sweepStep] at (12.6,-2.85)
        {later witnesses\\receive palette $S$};
    \draw[-{Latex[length=2mm]}] (4.75,-2.85) -- (5.15,-2.85);
    \draw[-{Latex[length=2mm]}] (9.85,-2.85) -- (10.25,-2.85);

    \node[sweepStep] at (2.4,-4.05)
        {\textbf{Phase II}\\earlier copying roots output $x_v$};
    \node[sweepStep] at (7.5,-4.05)
        {$v$ sees $(\vec n,\vec r)$\\and outputs $x_v$};
    \node[sweepStep] at (12.6,-4.05)
        {later roots output\\their prescribed colors $x$};
    \draw[-{Latex[length=2mm]}] (5.15,-4.05) -- (4.75,-4.05);
    \draw[-{Latex[length=2mm]}] (10.25,-4.05) -- (9.85,-4.05);
\end{tikzpicture}
        \caption{The tree construction for \Cref{thm:2sweeplower}.
        Each shaded box represents one witness tree, with its root
        $u$ shown and its other vertices suppressed; the indicated
        number of disjoint copies is attached to $v$. The two earlier-witness cases partition
        $\mathcal R$, so each palette contributes exactly $n_A$
        earlier neighbors. All vertices in the left witnesses precede
        $v$ in Phase~I, and all vertices in the right witnesses follow
        it.}
        \label{fig:lowerBoundTreeConstruction}
    \end{figure}

    In Phase~I, every earlier root announces its prescribed palette
    independently of $v$, so $v$ sees $(\vec{n},b)$ and announces
    $S$. The later roots consequently output their prescribed colors
    before $v$ acts in Phase~II. Thus $v$ sees
    $(\vec{n},\vec{r})$ and chooses the previously determined
    $x_v$. Finally, each earlier copying witness for $x_v$ outputs
    $x_v$. There is no circular dependency: the earlier roots'
    palettes are fixed before $v$ acts, whereas their final colors
    are chosen after $v$.

    Recall from \Cref{eq:NearEqualityOfDefectAndQuality} that $h_v(x_v)$
    and $r_v(x_v)$ count the earlier and later neighbors that finish
    in $x_v$. For every $A\in\mathcal R$ with $x_v\in E(A)$, the
    $n_A$ earlier copying roots output $x_v$. In addition, exactly
    $r_{x_v}$ later roots output $x_v$. Other earlier roots may also
    output $x_v$, so these contributions give a lower bound on the
    actual defect. Substituting $n_A=D\alpha_A$ and
    $r_{x_v}=D\delta_{x_v}^S$, and then using the definition of
    $\gamma_{x_v}$ in \Cref{eq:lowerBound-copying-mass}, gives
    \begin{align*}
        h_v(x_v)+r_v(x_v)
        &\geq\sum_{\substack{A\in\mathcal R\\x_v\in E(A)}}n_A+r_{x_v}\\
        &=\sum_{\substack{A\in\mathcal R\\x_v\in E(A)}}D\alpha_A
            +D\delta_{x_v}^S\\
        &=D\left(\sum_{\substack{A\in\mathcal R\\x_v\in E(A)}}\alpha_A
            +\delta_{x_v}^S\right)\\
        &=D\bigl(\gamma_{x_v}+\delta_{x_v}^S\bigr)\\
        &\geq(\rho+\eps)D,
    \end{align*}
    where the last inequality is \Cref{eq:lowerBound-defect-gap}.
    The internal degrees of the fixed witnesses, including the external
    edge, have a bound independent of $D$. For sufficiently large $D$, the
    central node therefore has maximum degree $\Delta=D$.

    Taking arbitrarily large multiples $D$ contradicts a guarantee of
    $\rho\deg(v)+o(\deg(v))$ as defined in
    \Cref{eq:defRelativeDefect}, since $\eps>0$ is fixed. This holds
    for every $\rho<1/\sqrt C$ and every fixed candidate family and
    pair of tie-breaking rules, proving the theorem.
\end{proof}

\section{Technical details of the coloring algorithms}
\label{sec:TechnicalDetails-Appendix}

\subsection{Proof of the directed defective coloring lemma}
\label{sec:directedDefectiveColoringProof}
\begin{proof}[Proof of \Cref{lemm:defColorBlackBox}]
    We give the directed version of the polynomial color reduction used in
    \cite{Kuhn2009} and Appendix~A of \cite{KawarabayashiS18}, including its
    iteration for nonconstant $\alpha$. Set
    $\eta:=\min\{\alpha,1/(2e)\}$.

    First consider any current $M$-coloring, which need not be proper, and
    an error allowance $0<\delta\leq\eta$. Put
    $h:=\max\{1,\log_{1/\delta}M\}$ and choose a power of two $s$ with
    $2h/\delta\leq s<4h/\delta$. Assign each old color $c$ a distinct
    polynomial $P_c$ over $\mathbb F_s$ of degree at most
    $k:=\lceil\log_s M\rceil-1$ (the case $M=1$ needs no reduction).
    There are at least $M$ such polynomials, and $k\leq h$.
    After learning its outneighbors' old colors, node $v$ chooses
    $t_v\in\mathbb F_s$ minimizing
    \[
        A_v(t):=\bigl|\{u\in N^+(v):c_u\ne c_v,
                    P_{c_u}(t)=P_{c_v}(t)\}\bigr|
    \]
    and takes the new color $(t_v,P_{c_v}(t_v))$. Fix all choices
    deterministically. Distinct polynomials agree at at most $k$ points,
    so, writing $b_v:=|\{u\in N^+(v):c_u\ne c_v\}|$, we have
    \[
        \sum_{t\in\mathbb F_s}A_v(t)\leq k b_v,
        \qquad A_v(t_v)\leq\frac{k}{s}b_v\leq\delta\beta_v.
    \]
    Every newly monochromatic outgoing arc is counted by $A_v(t_v)$:
    equal new colors require equal evaluation points as well as equal
    polynomial values. Previously monochromatic arcs contribute at most
    their previous number. Thus this step increases the outgoing defect
    at $v$ by at most $\delta\beta_v$ and uses at most
    \begin{equation}\label{eq:directedColorReduction}
        D\delta^{-2}\max\{1,\log_{1/\delta}M\}^2
    \end{equation}
    colors, for an absolute constant $D\geq16$. This argument is separate
    for each vertex and also applies when antiparallel arcs are present.

    We next allocate the error allowances. Let $a:=\eta/4$ and
    $K:=8\sqrt D$. If $q\leq(K/a)^2$, retain the initial proper coloring;
    it already has $O(\eta^{-2})$ colors and zero defect. Otherwise, let
    $x_i:=\ln^{(i)}q$, and let $T$ be the first index with $x_T\leq K/a$.
    For $i=1,\dots,T$, apply the reduction with
    $\delta_i:=a/2^{T-i}$. Here $T\geq1$ and $x_i>1$ for $i\leq T$.
    The minimality of $T$ and $2\ln x\leq x$ give
    \[
        \frac{K}{\delta_i}
        =\frac{K}{a}2^{T-i}
        <2^{T-i}x_{T-1}\leq x_{i-1}.
    \]
    If $M_i$ is the palette-size bound after step $i$, induction using
    \Cref{eq:directedColorReduction} yields
    \[
        M_i\leq16D\delta_i^{-2}x_i^2.
    \]
    Indeed, the first step satisfies this bound directly. For $i>1$,
    the induction hypothesis, $\delta_{i-1}=\delta_i/2$, and
    $x_{i-1}\geq K/\delta_i$ imply $M_{i-1}\leq x_{i-1}^4$.
    Since $\ln(1/\delta_i)\geq1$ and $x_i>1$, the maximum in
    \Cref{eq:directedColorReduction} is then at most $4x_i$.
    Consequently, $M_T=O(\eta^{-4})$. One final reduction with
    $\delta=\eta/2$ gives $O(\eta^{-2})$ colors, because
    $\log_{2/\eta}M_T=O(1)$. If a palette has only one color, the
    remaining steps can simply keep it.

    The initial defect is zero, and the total increase at each node is
    at most
    \[
        \left(\sum_{i=1}^T\delta_i+\frac{\eta}{2}\right)\beta_v
        \leq\eta\beta_v\leq\alpha\beta_v.
    \]
    Also, $\eta^{-2}=O(\alpha^{-2})$ for $0<\alpha\leq1$.
    There are $T+1=O(\log^*q)$ reductions. Each takes one round to
    exchange the current colors; the polynomial assignment and evaluation
    are local computations. The bounds above show that every color and
    evaluation pair has an encoding of
    $O(\log q+\log(1/\eta)+T)=O(\log q+\log(1/\alpha))$ bits.
    We start directly from the supplied proper coloring.
\end{proof}

\subsection{Proof of the strict-slack color-space reduction lemma}
\label{sec:strictSlackColorSpaceReductionProof}
\begin{proof}[Proof of \Cref{lemm:colorSpaceReductionBox}]
    Put $K:=\kappa(\lambda)$ and $k:=\lceil\log_\lambda C\rceil$.
    Fix a common depth-$k$ partition tree of the ambient color space,
    with at most $\lambda$ children per part and at most $\lambda^{h}$
    colors in a part with $h$ levels remaining. Empty ambient positions
    may be padded conceptually, but are never added to any input list.
    We prove by induction on the number of remaining levels that total
    budget strictly greater than $K^k\beta_v$ suffices.

    For $k=1$, the ambient space and all lists have size at most
    $\lambda$. The assumed algorithm applies, by monotonicity of
    $\kappa$, with round and message bounds $T(\lambda)$ and $M(\lambda)$.

    For $k>1$, let $L_{v,i}$ be the intersection of $L_v$ with child
    part $i$, and omit every empty intersection. For each nonempty
    intersection, define
    \[
        W_{v,i}:=\sum_{x\in L_{v,i}}(d_v(x)+1),
        \qquad t_{v,i}:=\frac{W_{v,i}}{K^{k-1}},
        \qquad b_{v,i}:=\lceil t_{v,i}\rceil-1.
    \]
    Since $W_{v,i}>0$ and $K>0$, each auxiliary defect $b_{v,i}$ is
    a nonnegative integer. The auxiliary list consists of the indices
    of the nonempty child parts, with total budget
    \[
        \sum_i(b_{v,i}+1)
        =\sum_i\lceil t_{v,i}\rceil
        \geq\sum_i t_{v,i}
        =\frac{\sum_{x\in L_v}(d_v(x)+1)}{K^{k-1}}
        >K\beta_v.
    \]
    Apply the assumed algorithm to these at most $\lambda$ indices.
    If node $v$ chooses index $i$, its outdegree $\beta'_v$ in the
    subgraph induced by that index satisfies
    \[
        \beta'_v\leq b_{v,i}<t_{v,i},
        \qquad\text{and hence}\qquad
        W_{v,i}>K^{k-1}\beta'_v.
    \]
    This is the strict induction hypothesis for the selected child
    part, whose list at $v$ remains nonempty. The index classes recurse
    independently and in parallel. The argument also covers
    $\beta_v=0$ without division by the outdegree.

    Each of the $k$ levels invokes the base algorithm once. Between
    levels, nodes exchange their chosen child index to identify the
    induced subgraphs; the common labels lie in $\{1,\ldots,\lambda\}$
    and require $O(\log\lambda)$ bits in one additional round. This
    proves the claimed round and message bounds. In our application
    $\lambda=4$, so this overhead is constant.
\end{proof}

\subsection{Relative defect of single buckets}
\label{sec:singleBucketAnalysis}
In this section, we specialize the general quality LP in \Cref{eq:generalQualityLP} to a single bucket containing all $C$ colors. We combine the resulting bounds with the two-bucket analysis to prove \Cref{thm:2bucketasymptotic}.

\begin{lemma}
    \label{lemma:SingleBucketAnalysis}
    If $C$ is a square number, the single-bucket relaxed LP has value
    $1/\sqrt C$, and the algorithm has relative defect
    \begin{align*}
        \rho \leq \frac{1}{\sqrt{C}}.
    \end{align*}
    Otherwise, the single bucket approach leads to a relative defect of 
    \begin{align*}
        \rho \leq \min \left\{\frac{1}{\lfloor \sqrt{C} \rfloor}, \frac{\lceil \sqrt{C} \rceil}{C} \right\}.
    \end{align*}
\end{lemma}
\begin{proof}
    If $C$ is a square number, the LP objective simplifies to
    \begin{align*}
        z = \frac{1 - a_1 }{\sqrt{C}} +  \frac{\sqrt{C}}{C} \cdot a_1 = \frac{1}{\sqrt{C}}.
    \end{align*}
    By the single-bucket analogue of \Cref{lemma:optimalityOfLP}, this LP
    value upper-bounds the relative defect of the algorithm. Together with
    \Cref{thm:2sweeplower}, this guarantee is optimal.
    If $C$ is not a square, we put all colors into one bucket. Let us first consider the case $C_1 = C$. 
    \begin{align*}
        z = \frac{1 - a_1}{s_1} + \frac{s_1}{C} \cdot a_1 = \frac{1}{s_1} + a_1 \left(\frac{s_1}{C} - \frac{1}{s_1} \right)
    \end{align*}
    Note that because $s_1^2 < C$ implies $s_1/C < 1/s_1$, the factor of
    $a_1$ is negative. Thus the relaxed LP is maximized at
    $a_1 = 0$ and has value $1/s_1$. In the case of the second bucket,
    we get
    \begin{align*}
        z = \frac{1 - a_2}{s_2} + \frac{s_2}{C} \cdot a_2 = \frac{1}{s_2} + a_2 \left(\frac{s_2}{C} - \frac{1}{s_2} \right).
    \end{align*}
    Note that because $s_2^2 > C$ implies $s_2/C > 1/s_2$, the factor of
    $a_2$ is positive. Thus the relaxed LP is maximized at
    $a_2 = 1$ and has value $s_2/C$. Taking the better of the two
    single-bucket constructions concludes the proof.
\end{proof}

\subsection{Auxiliary inequalities and bucket symmetry}
\label{sec:usefulLemmas}
\begin{lemma}
    \label{lemma:technicalLemmaZStar}
    For every nonsquare integer $C\geq2$, the value of $z_{2B}^*$ defined in
    \Cref{lemma:RealValuesC1AndC2} satisfies
    \begin{align*}
        z_{2B}^*
        \leq \frac{1}{\sqrt{C-1/4}}
        \leq \frac{1}{C^{1/2}}+\frac{1}{8C^{3/2}}
            +\frac{1}{32C^{5/2}}.
    \end{align*}
\end{lemma}
\begin{proof}
    Put $u:=s_1+1/2$. By \Cref{lemma:RealValuesC1AndC2},
    \begin{align*}
        z_{2B}^*=\frac{2u}{u^2+C-1/4}.
    \end{align*}
    The arithmetic--geometric mean inequality gives
    \begin{align*}
        \frac{u^2+C-\frac14}{2}
        \geq u\sqrt{C-\frac14},
    \end{align*}
    and hence $z_{2B}^*\leq1/\sqrt{C-1/4}$.

    For the second inequality, let $r:=1/(4C)\leq1/8$. A direct
    calculation shows that
    \begin{align*}
        (1-r)\left(1+\frac r2+\frac{r^2}{2}\right)^2-1
        =\frac{r^2}{4}(1-3r-r^2-r^3)\geq0,
    \end{align*}
    where the last inequality follows from
    $1-3r-r^2-r^3\geq311/512>0$. Since both sides are positive,
    \begin{align*}
        \frac{1}{\sqrt{1-r}}
        \leq1+\frac r2+\frac{r^2}{2}.
    \end{align*}
    Multiplying by $1/\sqrt C$ proves
    \begin{align*}
        \frac{1}{\sqrt{C-1/4}}
        \leq\frac{1}{C^{1/2}}+\frac{1}{8C^{3/2}}
            +\frac{1}{32C^{5/2}}.
    \end{align*}
\end{proof}

The following symmetry fact holds for any number of buckets.
\begin{lemma}
    \label{lemma:EqualAlphas}
    Let $m\geq1$, let
    \[
        \mathcal C=\mathcal C_1\mathbin{\dot\cup}\cdots
            \mathbin{\dot\cup}\mathcal C_m,
        \qquad
        \mathcal B_i:=\binom{\mathcal C_i}{s_i},
        \qquad
        \mathcal S:=\bigcup_{i=1}^m\mathcal B_i,
    \]
    where $1\leq s_i\leq|\mathcal C_i|$.  The general relaxed LP in
    \Cref{eq:generalQualityLP} (where colors are partitioned into buckets) has an optimal solution $\vec\alpha$ such that
    \[
        \alpha_S=\alpha_T
        \qquad\text{whenever }S,T\in\mathcal B_i
        \text{ for some }i\in\set{1,\dots,m}.
    \]
    \textit{Remark: We use this statement for the two-bucket ($m=2$)
    construction and for the three-bucket ($m=3$) construction in
    \Cref{sec:threeBucketsImprovement}.}
\end{lemma}
\begin{proof}
    Let $(\vec\alpha^{\,\mathrm{opt}},z^{\mathrm{opt}})$ be an optimal LP
    solution.  For each bucket, let
    \[
        a_i:=\sum_{T\in\mathcal B_i}\alpha_T^{\mathrm{opt}},
    \]
    and distribute this total mass uniformly within the bucket by defining
    \[
        \bar\alpha_T:=\frac{a_i}{|\mathcal B_i|}
        \qquad(T\in\mathcal B_i).
    \]
    This preserves nonnegativity and the total $\alpha$-mass.

    Fix a bucket $i$.  For every $T\in\mathcal B_i$, double counting the pairs
    $(x,S)$ with $x\in T\cap S$ and $S\in\mathcal B_i$ gives
    \[
        \frac1{|\mathcal B_i|}\sum_{S\in\mathcal B_i}|S\cap T|
        =\frac{s_i\binom{|\mathcal C_i|-1}{s_i-1}}
            {\binom{|\mathcal C_i|}{s_i}}
        =\frac{s_i^2}{|\mathcal C_i|}.
    \]
    If $T\in\mathcal B_j$ for $j\neq i$, then $S\cap T=\varnothing$ for every
    $S\in\mathcal B_i$.  Therefore, averaging the original qualities and
    using \Cref{eq:qualityOfSet} gives
    \begin{align*}
        \frac1{|\mathcal B_i|}\sum_{S\in\mathcal B_i}
            Q_{\vec\alpha^{\,\mathrm{opt}}}(S)
        &=\frac1{s_i}\left(
            1+\sum_{T\in\mathcal S}
            \left(\frac1{|\mathcal B_i|}
                \sum_{S\in\mathcal B_i}|S\cap T|-1\right)
            \alpha_T^{\mathrm{opt}}\right)\\
        &=\frac1{s_i}\left(
            1+\left(\frac{s_i^2}{|\mathcal C_i|}-1\right)a_i
            -\sum_{j\neq i}a_j\right)\\
        &=\frac1{s_i}\left(1-\sum_{j=1}^m a_j
            +\frac{s_i^2}{|\mathcal C_i|}a_i\right).
    \end{align*}
    On the other hand, for each fixed $R\in\mathcal B_i$, the same counting
    identity gives
    \[
        \sum_{T\in\mathcal B_i}|R\cap T|
        =|\mathcal B_i|\frac{s_i^2}{|\mathcal C_i|}.
    \]
    Since $\bar\alpha_T=a_i/|\mathcal B_i|$ within bucket $i$ and palettes
    from all other buckets are disjoint from $R$, it follows that
    \begin{align}
        Q_{\bar{\vec\alpha}}(R)
        =\frac1{s_i}\left(1-\sum_{j=1}^m a_j
            +\frac{a_i}{|\mathcal B_i|}
                \sum_{T\in\mathcal B_i}|R\cap T|\right)
        =\frac1{|\mathcal B_i|}\sum_{S\in\mathcal B_i}
            Q_{\vec\alpha^{\,\mathrm{opt}}}(S)
            \label{eq:SetRWithNewAlpha}
    \end{align}

    By \Cref{eq:SetRWithNewAlpha}, every quality after uniformization
    is an average of original qualities, each at least $z^{\mathrm{opt}}$.
    Hence $(\bar{\vec\alpha},z^{\mathrm{opt}})$ remains feasible and is
    optimal. The equality of masses within each bucket holds by construction.
\end{proof}

\begin{lemma}[Rounding the three-bucket split]
    \label{lem:threeBucketRounding}
    Let $C=s^2+\varepsilon$, where $s\geq4$ is an integer and
    $\varepsilon\in\{-1,1\}$. Form the two candidate splits from
    \Cref{eq:threeBucketPlusSplit,eq:threeBucketMinusSplit} for the
    corresponding sign, set $C_2=C-C_1-C_3$, and choose the candidate
    minimizing the right-hand side of \Cref{eq:threeBucketCertificate}.
    Let $D_1,D_2$ denote the two denominators in that bound.

    Both candidate splits are feasible, with $C_1\geq s-1$, $C_2\geq s$,
    and $C_3\geq s+1$. There is an absolute constant $K>0$ such that,
    for all sufficiently large $C$, the selected split satisfies
    \begin{equation}
        \min\set{D_1,D_2}\geq\sqrt C-\frac3{2C}-K C^{-3/2}.
        \label{eq:threeBucketRounding}
    \end{equation}
    The selected split also satisfies
    $\min\set{D_1,D_2}>s$ for $\varepsilon=1$ and
    $\min\set{D_1,D_2}>C/s$ for $\varepsilon=-1$.
\end{lemma}
\begin{proof}
    Solving $D_1=D_2$ gives $C_3^*$ in
    \Cref{eq:threeBucketPlusSplit} when $C_1=s-1$, and $C_1^*$ in
    \Cref{eq:threeBucketMinusSplit} when $C_3=s+1$.
    For $s\geq4$, these values satisfy
    \[
        s+1<C_3^*<s^2-2s+2,
        \qquad s-1<C_1^*<s^2-2s-2.
    \]
    The endpoints are integers and describe exactly the feasible ranges
    with $C_2\geq s$, so rounding either way preserves feasibility.

    Write $z_{3B}^*$ for the value of the right-hand side of
    \Cref{eq:threeBucketCertificate} at the balanced real-valued split.
    At this split, $D_1=D_2>0$, so $z_{3B}^*=1/D_1=1/D_2$.
    We compute their common value separately for the two signs.
    For $\varepsilon=1$, we have $C_1=s-1$ and
    $C_2=s^2-s+2-C_3^*$. Substituting the value of $C_3^*$ from
    \Cref{eq:threeBucketPlusSplit} into $D_1$ gives
    \begin{align*}
        D_1
        &=1+\frac{s^2-s+2-C_3^*}{s}+\frac{C_3^*}{s+1}\\
        &=s+\frac2s
            -\frac{(s+1)(3s-2)}{s(2s+1)(s-1)}\\
        &=s+\frac{s-3}{(s-1)(2s+1)}
         =s+\frac1{2s+1}-\frac2{(s-1)(2s+1)}.
    \end{align*}
    For $\varepsilon=-1$, we have $C_3=s+1$ and
    $C_2=s^2-s-2-C_1^*$. Using \Cref{eq:threeBucketMinusSplit}, we obtain
    \begin{align*}
        D_1
        &=\frac{C_1^*}{s-1}+\frac{s^2-s-2-C_1^*}{s}+1\\
        &=s-\frac2s
            +\frac{(s-1)(3s+2)}{s(2s-1)(s+1)}\\
        &=s-\frac{s+3}{(s+1)(2s-1)}
         =s-\frac1{2s-1}-\frac2{(s+1)(2s-1)}.
    \end{align*}
    Thus both cases can be written as
    \begin{equation}
        \frac1{z_{3B}^*}
        =D_1=D_2
        =s+\frac{\varepsilon}{2s+\varepsilon}
            -\frac{2}{(s-\varepsilon)(2s+\varepsilon)}.
        \label{eq:threeBucketBalancedDenominator}
    \end{equation}
    To control the rounding error, keep $s$ and $\varepsilon$ fixed and
    vary the bucket size being rounded, adjusting $C_2$ so that the total
    remains $C$. For $\varepsilon=1$, we fix $C_1=s-1$ and set
    $C_2=s^2-s+2-C_3$. Both denominators equal $1/z_{3B}^*$ at
    $C_3=C_3^*$, so subtracting their balanced values gives
    \begin{align*}
        D_1&=\frac1{z_{3B}^*}-\frac{C_3-C_3^*}{s(s+1)},\\
        D_2&=\frac1{z_{3B}^*}
            +\left(\frac2{(s+1)^2}-\frac1{s^2}\right)(C_3-C_3^*).
    \end{align*}
    The coefficient of $C_3-C_3^*$ in $D_1$ has absolute value
    $1/[s(s+1)]$, and, since $s\geq4$, the coefficient in $D_2$ satisfies
    \[
        0<\frac2{(s+1)^2}-\frac1{s^2}
        =\frac{s^2-2s-1}{s^2(s+1)^2}
        \leq\frac1{s(s+1)}.
    \]
    For $\varepsilon=-1$, we fix $C_3=s+1$ and set
    $C_2=s^2-s-2-C_1$. Expanding around $C_1=C_1^*$ similarly gives
    \begin{align*}
        D_1&=\frac1{z_{3B}^*}+\frac{C_1-C_1^*}{s(s-1)},\\
        D_2&=\frac1{z_{3B}^*}-\frac{C_1-C_1^*}{s^2}.
    \end{align*}
    The absolute values of the coefficients of $C_1-C_1^*$ are
    $1/[s(s-1)]$ and $1/s^2\leq1/[s(s-1)]$.
    Define a common bound on the nearest-integer distances for both signs by
    \[
        \delta_s:=\max\left\{
            \min_{m\in\mathbb Z}|C_3^*-m|,
            \min_{m\in\mathbb Z}|C_1^*-m|
        \right\}.
    \]
    For any real $t$, its distance to the nearest integer is
    $\min\set{t-\lfloor t\rfloor,\lceil t\rceil-t}$.
    If $t$ is an integer, both distances are zero; otherwise they sum to
    $1$, so their minimum is at most $1/2$. Hence $0\leq\delta_s\leq1/2$.
    For either sign, rounding to the nearest integer gives a feasible
    candidate at distance at most $\delta_s$ from the balanced bucket size.
    The displayed formulas show that each denominator at this rounding is at least
    $1/z_{3B}^*-\delta_s/[s(s+\varepsilon)]$.
    Our chosen split has an upper bound no larger than that candidate's.
    Hence
    \begin{align}
        \min\set{D_1,D_2}
        &\geq\frac1{z_{3B}^*}
            -\frac{\delta_s}{s(s+\varepsilon)}
        \geq\frac1{z_{3B}^*}-\frac1{2s(s+\varepsilon)}.
        \label{eq:threeBucketRoundingError}
    \end{align}
    For $s\geq4$, the final expression in \Cref{eq:threeBucketRoundingError}
    is greater than $s$ when $\varepsilon=1$, and greater than $C/s$
    when $\varepsilon=-1$. This proves the stated strict lower bounds.

    \emph{The asymptotic bound.}
    Polynomial division of the formulas in
    \Cref{eq:threeBucketPlusSplit,eq:threeBucketMinusSplit} gives the exact identities
    \begin{align*}
        C_3^*
        &=\frac{(s+1)^2(3s-2)}{(2s+1)(s-1)}\\
        &=\frac{3s^3+4s^2-s-2}{2s^2-s-1}\\
        &=\frac32s+\frac{11}4
            +\frac{13s+3}{4(2s^2-s-1)},\\
        C_1^*
        &=\frac{(s-1)^2(3s+2)}{(2s-1)(s+1)}\\
        &=\frac{3s^3-4s^2-s+2}{2s^2+s-1}\\
        &=\frac32s-\frac{11}4
            +\frac{13s-3}{4(2s^2+s-1)}.
    \end{align*}
    Both remainder fractions are positive and $O(s^{-1})$, since their
    numerators are linear and their denominators are quadratic in $s$.
    For integer $s$, both $\frac32s+\frac{11}4$ and
    $\frac32s-\frac{11}4$ have distance exactly $1/4$ from the nearest
    integer. Adding a remainder changes this distance by at most the
    absolute value of the remainder, by the triangle inequality.
    Applying this to both split sizes and taking the maximum gives
    \[
        \left|\delta_s-\frac14\right|
        \leq\max\left\{
            \frac{13s+3}{4(2s^2-s-1)},
            \frac{13s-3}{4(2s^2+s-1)}
        \right\}
        =O(s^{-1}).
    \]
    Thus $\delta_s=1/4+O(s^{-1})$.
    Expanding the two fractions in \Cref{eq:threeBucketBalancedDenominator}
    and using $\varepsilon^2=1$ gives
    \begin{align*}
        \frac1{z_{3B}^*}
        &=s+\frac{\varepsilon}{2s}-\frac1{4s^2}
            -\frac1{s^2}+O(s^{-3})\\
        &=s+\frac{\varepsilon}{2s}-\frac5{4s^2}+O(s^{-3}).
    \end{align*}
    Since $C=s^2+\varepsilon$, we have
    $\sqrt C=s+\varepsilon/(2s)+O(s^{-3})$ and
    $s^{-2}=C^{-1}+O(C^{-2})$. Combining these expansions with the
    estimate for $\delta_s$ yields
    \[
        \frac1{z_{3B}^*}=\sqrt C-\frac5{4C}+O(C^{-3/2}),
        \qquad
        \frac{\delta_s}{s(s+\varepsilon)}
            =\frac1{4C}+O(C^{-3/2}).
    \]
    Substituting these estimates into \Cref{eq:threeBucketRoundingError}
    gives \Cref{eq:threeBucketRounding} for some absolute $K>0$.
\end{proof}

\subsection{Optimal partitioning of the color space into two buckets}
\label{sec:full-table}
\Cref{tab:defect_table_complete} extends \Cref{tab:short_defect_data} and lists the optimal choices within the analyzed family for the displayed values of $C$, computed by \Cref{claim:optimalIntegerBuckets}. A single bucket is used when $C$ is a square or differs from a square by one.

\begingroup
\small
\renewcommand{\arraystretch}{1.4}
\begin{longtable}{|c|c|c|c|c|c|}
\caption{Best relaxed-LP bounds within the analyzed Single-Bucket and Two-Bucket Two-Sweep constructions, with bucket choices computed by \Cref{claim:optimalIntegerBuckets}.}\label{tab:defect_table_complete}\\
\hline
\textbf{C} & \textbf{one-/two-bucket bound} & $\approx$ & \textbf{$C_1$} & \textbf{$C_2$} & \textbf{single bucket} \\
\hline
\endfirsthead

\multicolumn{6}{c}%
{{\tablename\ \thetable{} -- Continued from previous page}} \\
\hline
\textbf{C} & \textbf{one-/two-bucket bound} & $\approx$ & \textbf{$C_1$} & \textbf{$C_2$} & \textbf{single bucket} \\
\hline
\endhead

\hline \multicolumn{6}{|r|}{{\textit{Continued on next page}}} \\ \hline
\endfoot

\endlastfoot

\textbf{3} & $2/3$ & 0.66667 & 0 & 3 & * \\ \hline
\textbf{4} & $1/2$ & 0.50000 & 4 & 0 & * \\ \hline
\textbf{5} & $1/2$ & 0.50000 & 5 & 0 & * \\ \hline
\textbf{6} & $3/7$ & 0.42857 & 2 & 4 &  \\ \hline
\textbf{7} & $9/23$ & 0.39130 & 2 & 5 &  \\ \hline
\textbf{8} & $3/8$ & 0.37500 & 0 & 8 & * \\ \hline
\textbf{9} & $1/3$ & 0.33333 & 9 & 0 & * \\ \hline
\textbf{10} & $1/3$ & 0.33333 & 10 & 0 & * \\ \hline
\textbf{11} & $4/13$ & 0.30769 & 6 & 5 &  \\ \hline
\textbf{12} & $12/41$ & 0.29268 & 5 & 7 &  \\ \hline
\textbf{13} & $16/57$ & 0.28070 & 4 & 9 &  \\ \hline
\textbf{14} & $16/59$ & 0.27119 & 3 & 11 &  \\ \hline
\textbf{15} & $4/15$ & 0.26667 & 0 & 15 & * \\ \hline
\textbf{16} & $1/4$ & 0.25000 & 16 & 0 & * \\ \hline
\textbf{17} & $1/4$ & 0.25000 & 17 & 0 & * \\ \hline
\textbf{18} & $5/21$ & 0.23810 & 12 & 6 &  \\ \hline
\textbf{19} & $25/108$ & 0.23148 & 11 & 8 &  \\ \hline
\textbf{20} & $25/111$ & 0.22523 & 9 & 11 &  \\ \hline
\textbf{21} & $20/91$ & 0.21978 & 7 & 14 &  \\ \hline
\textbf{22} & $20/93$ & 0.21505 & 5 & 17 &  \\ \hline
\textbf{23} & $25/119$ & 0.21008 & 4 & 19 &  \\ \hline
\textbf{24} & $5/24$ & 0.20833 & 0 & 24 & * \\ \hline
\textbf{25} & $1/5$ & 0.20000 & 25 & 0 & * \\ \hline
\textbf{26} & $1/5$ & 0.20000 & 26 & 0 & * \\ \hline
\textbf{27} & $6/31$ & 0.19355 & 20 & 7 &  \\ \hline
\textbf{28} & $15/79$ & 0.18987 & 18 & 10 &  \\ \hline
\textbf{29} & $36/193$ & 0.18653 & 16 & 13 &  \\ \hline
\textbf{30} & $9/49$ & 0.18367 & 14 & 16 &  \\ \hline
\textbf{31} & $15/83$ & 0.18072 & 11 & 20 &  \\ \hline
\textbf{32} & $30/169$ & 0.17751 & 9 & 23 &  \\ \hline
\textbf{33} & $18/103$ & 0.17476 & 7 & 26 &  \\ \hline
\textbf{34} & $36/209$ & 0.17225 & 5 & 29 &  \\ \hline
\textbf{35} & $6/35$ & 0.17143 & 0 & 35 & * \\ \hline
\textbf{36} & $1/6$ & 0.16667 & 36 & 0 & * \\ \hline
\textbf{37} & $1/6$ & 0.16667 & 37 & 0 & * \\ \hline
\textbf{38} & $7/43$ & 0.16279 & 30 & 8 &  \\ \hline
\textbf{39} & $49/305$ & 0.16066 & 28 & 11 &  \\ \hline
\textbf{40} & $49/309$ & 0.15858 & 25 & 15 &  \\ \hline
\textbf{41} & $21/134$ & 0.15672 & 22 & 19 &  \\ \hline
\textbf{42} & $42/271$ & 0.15498 & 19 & 23 &  \\ \hline
\textbf{43} & $49/320$ & 0.15313 & 17 & 26 &  \\ \hline
\textbf{44} & $49/324$ & 0.15123 & 14 & 30 &  \\ \hline
\textbf{45} & $42/281$ & 0.14947 & 11 & 34 &  \\ \hline
\textbf{46} & $21/142$ & 0.14789 & 8 & 38 &  \\ \hline
\textbf{47} & $49/335$ & 0.14627 & 6 & 41 &  \\ \hline
\textbf{48} & $7/48$ & 0.14583 & 0 & 48 & * \\ \hline
\textbf{49} & $1/7$ & 0.14286 & 49 & 0 & * \\ \hline
\textbf{50} & $1/7$ & 0.14286 & 50 & 0 & * \\ \hline
\textbf{51} & $8/57$ & 0.14035 & 42 & 9 &  \\ \hline
\textbf{52} & $56/403$ & 0.13896 & 39 & 13 &  \\ \hline
\textbf{53} & $64/465$ & 0.13763 & 36 & 17 &  \\ \hline
\textbf{54} & $64/469$ & 0.13646 & 33 & 21 &  \\ \hline
\textbf{55} & $28/207$ & 0.13527 & 29 & 26 &  \\ \hline
\textbf{56} & $28/209$ & 0.13397 & 26 & 30 &  \\ \hline
\textbf{57} & $32/241$ & 0.13278 & 23 & 34 &  \\ \hline
\textbf{58} & $32/243$ & 0.13169 & 20 & 38 &  \\ \hline
\textbf{59} & $56/429$ & 0.13054 & 16 & 43 &  \\ \hline
\textbf{60} & $56/433$ & 0.12933 & 13 & 47 &  \\ \hline
\textbf{61} & $64/499$ & 0.12826 & 10 & 51 &  \\ \hline
\textbf{62} & $64/503$ & 0.12724 & 7 & 55 &  \\ \hline
\textbf{63} & $8/63$ & 0.12698 & 0 & 63 & * \\ \hline
\textbf{64} & $1/8$ & 0.12500 & 64 & 0 & * \\ \hline
\textbf{65} & $1/8$ & 0.12500 & 65 & 0 & * \\ \hline
\end{longtable}
\endgroup

\end{document}